\documentclass[11pt,dvipsnames]{article}

\usepackage{setspace}
\usepackage{jheppub} 
\usepackage[utf8]{inputenc}
\usepackage{epsfig}
\usepackage{bbm}
\usepackage{bbding}
\usepackage{amsfonts}
\usepackage{amssymb}
\usepackage{mathtools}
\usepackage{dsfont}
\usepackage{bm}
\usepackage{graphicx}
\usepackage{longtable}
\usepackage{pdflscape}
\usepackage{subcaption}
\usepackage{psfrag}
\usepackage[capitalise]{cleveref}
\usepackage{tikz}
\usepackage{pgfplots}
\usepackage{pgfplotstable}
\usepackage{subfiles}
\usepackage{longtable}
\usepackage{tabularx}
\usepackage{ltablex}
\usepackage{booktabs}
\usepackage[export]{adjustbox}
\usepackage{listings}
\usepackage{multirow} 
\usepackage{cancel,slashed}
\usepackage{bbm}
\usepackage{graphicx}
\usepackage{latexsym}
\usepackage{transparent}
\usepackage{mathtools}
\usepackage{array}
\usepackage{makecell}
\usepackage{graphics,psfrag}
\usepackage{placeins}
\usepackage{nowidow}
\usepackage{listings}
\usepackage[normalem]{ulem}
\usepackage{environ}
\usepackage{tikz}
\usepackage{enumitem}   
\usepackage{stmaryrd}
\usetikzlibrary{positioning,arrows.meta}
\usetikzlibrary{positioning,trees,decorations.pathmorphing,decorations.markings,decorations.pathreplacing,calc,shapes,shapes.geometric,shapes.symbols,patterns,arrows}

\usepackage{fontawesome5}
\makeatletter
\newcommand{\github}[1]{%
   \href{#1}{{\normalsize \color{black}\faGithub}}%
}
\makeatother

\usetikzlibrary{patterns}
\usetikzlibrary{patterns.meta}

\usetikzlibrary{fadings}
\usetikzlibrary{decorations.shapes}
\usepackage{amsmath}
\usepackage{amssymb}
\usepackage{amsthm}
\usepackage{algorithm}
\usepackage{algpseudocode}
\usepackage{placeins}

\usepackage{xstring}

\usepackage{float}

\newcommand{\newshortstack}[1]
{\begingroup\renewcommand{\arraystretch}{1.1}
\ifmmode
\begin{array}{c}#1\end{array}%
\else
\begin{tabular}{c}#1\end{tabular}%
\fi
\endgroup}

\pgfplotsset{
        compat=1.9,
        compat/bar nodes=1.8,
    }

\newcommand{\be}{\begin{equation}}
	\newcommand{\ee}{\end{equation}}
\newcommand{\bea}{\begin{eqnarray}}
	\newcommand{\eea}{\end{eqnarray}}

\renewcommand{\epsilon}{\varepsilon}

\newcommand{\ben}{\begin{enumerate}}
	\newcommand{\een}{\end{enumerate}}
\newcommand{\bei}{\begin{itemize}}
	\newcommand{\eei}{\end{itemize}}
	
	\makeatletter
\tikzset{
    dot diameter/.store in=\dot@diameter,
    dot diameter=3pt,
    dot spacing/.store in=\dot@spacing,
    dot spacing=10pt,
    dots/.style={
        line width=\dot@diameter,
        line cap=round,
        dash pattern=on 0pt off \dot@spacing
    }
}

\usepackage{setspace}
\usepackage{graphicx} 
\usepackage{lmodern}
\usepackage{xspace}
\usepackage{array} 
\usepackage{etoolbox}
\usetikzlibrary{positioning,trees,decorations.pathmorphing,decorations.markings,decorations.pathreplacing,calc,shapes,shapes.geometric,shapes.symbols,patterns,arrows}
	
\tikzset{decorate sep/.style 2 args=
{decorate,decoration={shape backgrounds,shape=circle,shape size=#1,shape sep=#2}}}
\usepackage{tikz}
\usepackage{latexsym}
\usepackage{graphicx}
\usepackage{svg}
\usepackage{tikz-3dplot}
\usetikzlibrary{arrows.meta}
\usetikzlibrary{patterns}

\usepackage{mathrsfs}
\usepackage{nicematrix}

\usepackage[most]{tcolorbox}
\usepackage{xcolor}
\usepackage{etoolbox}

\newcounter{example}[section]
\renewcommand{\theexample}{\thesection.\arabic{example}}

\colorlet{examplebarcolor}{black!60}

\newtcolorbox{examplebox}[1][]{%
  enhanced,
  breakable,
  colback=white,
  colframe=white,
  borderline west={2pt}{0pt}{examplebarcolor},
  sharp corners,
  boxrule=0pt,
  before skip=10pt,
  after skip=10pt,
  left=8pt,
  right=2pt,
  top=4pt,
  bottom=4pt,
  fonttitle=\bfseries,
  coltitle=black,
  attach boxed title to top left={yshift=-2mm},
  boxed title style={empty, size=small, top=2pt, bottom=2pt},
  title=#1,
  parbox=false
}

\newenvironment{example}[1][]{%
\refstepcounter{example}%
\def\exampletitle{Example~\theexample}%
\ifstrempty{#1}%
    {\begin{examplebox}\textbf{\exampletitle}\\} 
    {\begin{examplebox}\textbf{\exampletitle:~#1}\\} %
}{%
  \end{examplebox}%
}

\definecolor{bluetto}{HTML}{0088ff}
\definecolor{snowfl}{HTML}{00ace6}
\definecolor{dgreen}{HTML}{298A08}
\usepackage{subcaption}

\theoremstyle{definition} 

\newtheorem{cor}{Corollary}

\newtheorem{lem}{Lemma}
\newtheorem{prop}{Proposition}

\newtheorem*{theorem*}{Theorem}
\newtheorem{defn}{Definition}
\DeclareMathOperator{\ConvOp}{conv}
\DeclareMathOperator{\ConeOp}{cone}
\DeclareMathOperator{\SuppOp}{supp}
\DeclareMathOperator{\RelintOp}{relint}
\newcommand{\conv}[1]{\ConvOp\!\left(#1\right)}
\newcommand{\cone}[1]{\ConeOp\!\left(#1\right)}
\newcommand{\supp}[1]{\SuppOp\!\left(#1\right)}
\newcommand{\relint}[1]{\RelintOp\!\left(#1\right)}
\newcommand{\setbuilder}[2]{\left\{ #1 \; \middle| \; #2 \right\}}

\crefname{defn}{Def.}{Defs.}
\crefname{cor}{Cor.}{Cors}
\crefname{prop}{Prop.}{Props.}
\crefname{conj}{Conj.}{Conjs.}
\crefname{lem}{Lem.}{Lems.}

\crefformat{section}{#2§#1#3}
\crefformat{appendix}{#2§#1#3}

\makeatletter
\renewcommand*\env@matrix[1][\arraystretch]{%
  \edef\arraystretch{#1}%
  \hskip -\arraycolsep
  \let\@ifnextchar\new@ifnextchar
  \array{*\c@MaxMatrixCols c}}
\makeatother

\usepackage{hyperref}
\begin{document}

	\pagestyle{plain}

	\makeatletter
	\@addtoreset{equation}{section}
	\makeatother
	\renewcommand{\theequation}{\thesection.\arabic{equation}}
	\pagestyle{empty}

\vspace{2cm}

\begin{center}
\phantom{a}\\
\vspace{0.8cm}
\scalebox{0.90}[0.90]{%
  \parbox{\textwidth}{\centering\fontsize{24}{30}\selectfont\bfseries
  Calabi--Yau Threefolds \\ from Vex Triangulations
}%
}\\
\end{center}

\vspace{0.4cm}
\begin{center}
\scalebox{0.95}[0.95]{{\fontsize{12}{30}\selectfont  
Nate MacFadden and Elijah Sheridan
}} 
\end{center}

\begin{center}
\vspace{0.25 cm}

\textsl{Department of Physics, Cornell University, Ithaca, NY 14853, USA}\\

  \vspace{1.1cm}
	\normalsize{\bf Abstract} \\[8mm]
\end{center}

\begin{center}
\begin{minipage}[h]{15.0cm}
    We study the birational geometry (i.e., K\"ahler moduli space) of Calabi--Yau (CY) threefold hypersurfaces in toric varieties arising from four-dimensional reflexive polytopes.
    In particular, it has been observed that the birational classes of these geometries are not exhausted by toric hypersurfaces arising from fine, regular, star triangulations (FRSTs). 
    We begin by introducing a classification problem: enumeration of birational classes of toric varieties, which is equivalent to enumeration of certain triangulations/fans. 
    We consider this problem from the complementary perspectives of triangulation theory and toric geometry, reviewing both theories in detail; this culminates in an explanation of how to generate all fine regular triangulations of a vector configuration (i.e., fine regular simplicial fans), which we implement in a Python package \texttt{regfans} \cite{regfans}.
    We then apply this theory to the Kreuzer--Skarke (KS) database, where we encounter both FRSTs and vex triangulations. 
    We study the non-weak-Fano toric varieties arising from vex triangulations, along with their CY hypersurfaces. 
    In particular, we show that all fine regular triangulations of a fixed 4D reflexive polytope give rise to smooth birational CY hypersurfaces, extending Batyrev's result from FRSTs to vex triangulations. 
    We exhaustively enumerate all $24,023,940$ fine regular triangulations in the KS database with $h^{1,1}\leq 7$, of which over $70\%$ are vex triangulations, and provide an upper bound of $10^{979}$ for fine regular triangulations in the entire KS database. 
    We conclude that vex triangulations of four-dimensional reflexive polytopes give rise to a large number of smooth Calabi--Yau threefolds and importantly provide toric descriptions for novel regions in the K\"ahler moduli space.
    \github{https://github.com/natemacfadden/regfans}
\end{minipage}
\end{center}

	\newpage

	\setcounter{page}{1}
	\pagestyle{plain}
	\setcounter{footnote}{0}
\tableofcontents

\newpage

\section{Introduction}
\label{sec:intro}

Calabi--Yau (CY) manifolds are ubiquitous in mathematics and theoretical physics. Generalizations of the pervasive elliptic curve, they demarcate the boundary between Fano (positively curved) and general-type (negatively curved) varieties and exhibit mirror symmetry, an intricate duality connecting complex, symplectic, and enumerative geometry. Moreover, they are the most convenient and well-studied category of spaces employed for compactifications of superstring theory, which are important examples of theories of quantum gravity.

For both mathematics and physics, then, it is important to generate many examples of CY manifolds. Moreover, it is especially helpful when the topological and geometric properties of these manifolds can be computed efficiently, in particular because such properties are often required to construct the low-energy effective theories of the associated string compactifications. 

Toric geometry has proven to be an invaluable tool for constructing large datasets of computationally tractable CY geometries. The enumeration of toric varieties and analysis of their properties reduce to combinatorial calculations --- e.g., constructing the face lattice or counting the integral points of a polytope --- which can be performed efficiently and at scale. However, while toric geometry --- and algebraic geometry more broadly --- provides convenient constructions of algebraic varieties satisfying the Calabi--Yau condition (trivial canonical bundle), these varieties are generically not manifolds in that they are not smooth. String compactifications on smooth CY varieties are substantially better understood than their singular counterparts, so ensuring smoothness is an important challenge.

Several large ensembles of smooth CY manifolds embedded in toric varieties have been constructed \cite{Candelas:1987kf, Scholler:2018apc}. In this paper, we will focus on the toric hypersurface CY threefolds which arise from the Kreuzer--Skarke database. In \cite{Batyrev:1993oya}, Batyrev identified that the toric varieties induced by fine, regular, star triangulations (FRSTs) of four-dimensional reflexive polytopes admit smooth CY threefold hypersurfaces. Subsequently, Kreuzer and Skarke enumerated all 473,800,776 four-dimensional reflexive polytopes \cite{Kreuzer:2000xy}, which give rise to as many as $10^{928}$ FRSTs \cite{Demirtas:2020dbm}\footnote{Though many give rise to diffeomorphic CY hypersurfaces \cite{Demirtas:2020dbm, Gendler:2023ujl, Chandra:2023afu, bigicy}.}, providing an impressive CY playground for the mathematics and physics communities.

Some topologically distinct CY manifolds can be related by birational maps, or isomorphisms defined on dense open subsets. A complete set of CY manifolds closed under such maps form a birational equivalence class, and are organized into a single continuous K\"ahler moduli space. In particular, this moduli space is a cone --- the extended K\"ahler cone --- which itself decomposes as a union of cones --- K\"ahler cones --- each corresponding to a particular CY in the equivalence class. Each individual cone parameterizes the geometry of the associated CY manifold (via the K\"ahler form, which measures volumes of holomorphic cycles); crossing a facet (codimension-one face) separating two chambers enacts a flop between the two CY manifolds, an elementary type of birational map. 

Usefully, the FRSTs of a fixed four-dimensional reflexive polytope give rise to birational CY threefold hypersurfaces; in this way, these FRSTs produce both many \textit{discrete topological representatives} of a birational class of CY threefolds and map out their shared \textit{continuous K\"ahler moduli space}. Indeed, the toric varieties associated to the distinct FRSTs of a reflexive polytope are themselves birational, such that the flops between their CY hypersurfaces descend from birational maps --- flips --- of the ambient varieties. Additionally, the K\"ahler cones of the ambient varieties approximate those of their CY hypersurfaces. In this way, FRSTs map out part of the toric hypersurface CY K\"ahler moduli space.

However, it is known that the FRSTs do not, in general, exhaust the birational equivalence classes of their CY hypersurfaces: the union of CY K\"ahler cones arising from FRSTs is only a subset of the full K\"ahler moduli space. Such non-FRST CYs have been discussed in the literature: in \cite{Gendler:2022ztv, Gendler:2023ujl}, for example, they were called ``non-toric phases.'' Even when non-FRST elements of the birational equivalence class can be constructed, without a toric hypersurface realization they are more conceptually and computationally challenging to enumerate and study. 

Moreover, FRSTs do not even exhaust the birational equivalence classes of the ambient toric varieties. The toric varieties associated to FRSTs were originally distinguished because they satisfied sufficient conditions for yielding \textit{smooth} CY hypersurfaces \cite{Batyrev:1993oya}, but are not general. Non-FRST toric varieties have previously featured in the string compactifications literature. In particular, they have predominantly been studied in the context of non-reflexive polytopes \cite{Berglund:2016yqo, Berglund:2016nvh, Berglund:2022dgb, Berglund:2022zto, Berglund:2024zuz}; their relevance for reflexive polytopes has been identified \cite{Huang:2019pne, Jefferson:2022ssj, vex_notes} and has come up in specific examples \cite{Gross1997DeformationObstructed, Ruan1996TopologicalSigma} (as we will discuss before \cref{ex:intro}), but has not been studied systematically.

Our goal in this paper is to construct new smooth Calabi--Yau threefolds --- especially those with the computability endowed by toric embeddings --- and map out more of the K\"ahler moduli space of toric hypersurface Calabi--Yau threefolds arising from the Kreuzer--Skarke database. Because the birational geometry of the ambient toric variety and its CY hypersurface are so tightly related, and because the traditional restriction to FRSTs leaves gaps in the birational geometry of both, we proceed by first considering the following classification questions. 
\begin{itemize}
    \item What is the birational equivalence class\footnote{In this paper, we adopt a restricted notion of birational maps. In particular, we require them to be isomorphisms in codimension one (as is done at points in \cite{cls}, for example, as well as in \cite{Gendler:2022ztv}). This is discussed further in \cref{sec:toric_bir}.} of a toric variety?
    \item What fine regular simplicial fans can be constructed from a fixed set of rays? 
    \item What fine regular triangulations\footnote{We note that by ``triangulation'' here we do not mean triangulations of polytopes, which is the usual meaning of the word in the context of Batyrev's construction: we will clarify this in \cref{sec:triang}.} can be constructed from a vector configuration?
\end{itemize}
These three questions are actually equivalent: the combinatorial nature of toric geometry reduces the first question to the second, while the third question is a restating in the language of triangulation theory. While separate, toric geometry and triangulation theory study many of the same objects and problems, and we find their distinct but complementary perspectives to be helpful to consider in tandem, both for conceptual understanding and the development of efficient computational methods. The above questions are standard and well-understood in their respective disciplines. We will review triangulation theory and toric geometry in detail, and also clarify how to translate between them. In particular, in \cref{tab:dict}, we provide a dictionary relating some of their respective vocabulary; further discussion of this table and subtleties in translation appears in \cref{sec:translate}.

Upon completing a general discussion of these standard classification questions, we apply this knowledge to the toric geometry of 4D reflexive polytopes. In so doing, we are naturally led to consider fans not arising from FRSTs. Motivated and guided by the earlier research on these objects \cite{Berglund:2016yqo, Berglund:2016nvh, Berglund:2022dgb, Berglund:2022zto, Berglund:2024zuz, Huang:2019pne, Jefferson:2022ssj, vex_notes}, we refer to such fans as vex triangulations\footnote{When originally introduced, the word ``vex'' was styled in uppercase; we adopt a lowercase styling to distinguish the vex triangulations discussed in this paper from the VEX polytopes/multitopes of \cite{Berglund:2016yqo, Berglund:2016nvh, Berglund:2022dgb, Berglund:2022zto, Berglund:2024zuz}. For example, VEX polytopes/multitopes were introduced to generalize toric geometry, while vex triangulations in this paper merely correspond to a particular class of toric varieties. We will comment briefly on the relationship between vex triangulations and VEX polytopes in \cref{sec:frst_vs_vex}. Our use of the word ``vex'' also closely aligns with \cite{Huang:2019pne, Jefferson:2022ssj, vex_notes}, where it is styled in lowercase.} and study the theory and computation of the induced toric varieties and the structure of their CY hypersurfaces, especially their possible singularities.

Our main results are summarized as follows.
\begin{enumerate}
    \item The toric varieties arising from FRSTs of four-dimensional reflexive polytopes generically constitute only a fraction of their birational equivalence class.
    \item Vex triangulations can be efficiently generated and studied, and the topology and geometry of the associated toric varieties --- and their CY hypersurfaces --- are readily computable.
    \item In particular, all fine regular simplicial fans arising from a fixed four-dimensional reflexive polytope (FRSTs and vex triangulations) induce non-weak-Fano Gorenstein toric varieties with smooth, birational CY hypersurfaces (\cref{main}). In particular, we provide evidence that as the number of vectors/rays increases, vex triangulations constitute an increasing fraction of fine regular simplicial fans. This has the following two immediate corollaries.
    \begin{enumerate}
        \item Vex triangulations give rise to many new smooth toric hypersurface Calabi--Yau threefolds topologies.
        \item Vex triangulations map out significant new regions in the K\"ahler moduli space of toric hypersurface CY manifolds: in particular, the \textit{moving cone} of the ambient toric variety must be contained in the extended K\"ahler cone (\cref{eq:movable} and \cref{moving}).
    \end{enumerate}
\end{enumerate}
It is worth stressing that while users of the Kreuzer--Skarke database have traditionally restricted to fans arising from FRSTs, vex triangulations are not significantly different from the perspective of triangulation theory and toric geometry, and their inclusion should not be considered a serious complication or generalization.\footnote{We do note that some of the more technical methods developed for FRST toric hypersurface CYs --- such as computational mirror symmetry algorithms \cite{Demirtas2024} --- may require some adaptation. We comment further on this in the conclusion.}

This paper is organized as follows. For the remainder of the introduction, we present a simple motivating example which makes our discussion more concrete: a CY threefold arising from an FRST of a 4D reflexive polytope whose K\"ahler moduli space contains a distinct topology arising from a non-FRST fan. Following this, we develop the basic theories of triangulations and of toric geometry successively and in detail, in \cref{sec:triang} and \cref{sec:toric}. The familiar reader can skip the majority of these sections. We then review the secondary fan in \cref{sec:sec_fan}, the primary object employed in both theories to organize the triangulation/fan classification problem, and discuss methods for efficiently constructing the secondary fan, yielding the solution to the classification problem. These methods are implemented in an open-source Python package \texttt{regfans} \cite{regfans}. \github{https://github.com/natemacfadden/regfans}

Up to this point, the content is almost entirely a review of known results; we break new ground in \cref{sec:ref_poly} by applying these results on triangulation/fan classification to the case of 4D reflexive polytopes. In particular, we review Batyrev's construction of CY hypersurfaces from FRSTs, develop the theory of toric varieties arising from vex triangulations, and study their anticanonical hypersurfaces (especially their smoothness). This culminates in the smoothness result summarized in \cref{main}. Using \texttt{CYTools} \cite{Demirtas:2022hqf} and \texttt{regfans} \cite{regfans}, we then exhaustively enumerate fans/triangulations for small Hodge numbers ($h^{1,1} \leq 7$), demonstrating the ubiquity of vex triangulations, and present upper bounds at large Hodge numbers, generalizing the methods of \cite{Demirtas:2020dbm}. Finally, we conclude in \cref{sec:conc}. In \cref{sec:translate} we spend some time clarifying subtleties in translation between triangulation theory and toric geometry, and in \cref{sec:gale} we briefly comment on how Gale duality lurks behind the scenes in triangulation theory and toric geometry.

\begin{table}[p]
    \centering
    {
    \renewcommand{\arraystretch}{1.1}
    \begin{tabular}{|>{\centering\arraybackslash}m{3.6cm}|
                >{\centering\arraybackslash}m{3.6cm}|
                >{\centering\arraybackslash}m{5.6cm}|}
        \hline
        \multicolumn{2}{|c|}{Toric Geometry} & Triangulation Theory \\
        Algebraic Geometric & Combinatorial & \\
        \hline\hline
        Birational class of toric varieties & Collection of rays $\Sigma(1)$ & Vector configuration (VC) $\mathbf{A}$ \\
        \hline
        Separated normal toric variety $V_\Sigma$ & Fan $\Sigma$ & VC subdivision $\mathscr{S}$ \\
        \hline
        $\mathbb{Q}$-factorial toric variety & Simplicial fan & VC triangulation $\mathscr{T}$ \\
        \hline
        Compact toric variety & Complete fan & Subdivision of a totally cyclic VC \\
        \hline
        Projective compact toric variety & Normal fan of a polytope & Regular subdivision of a totally cyclic VC \\
        \hline
        \multicolumn{2}{|>{\centering\arraybackslash}m{7.5cm}|}{Torus-invariant $\mathbb{R}$-divisor $D = \sum_\rho a_\rho D_\rho$ ($a \in \mathbb{R}^{|\Sigma(1)|}$)} & Height vector $\omega \in \mathbb{R}^{|\mathbf{A}|}$ \\
        \hline
        \multicolumn{2}{|>{\centering\arraybackslash}m{7.5cm}|}{Torus-invariant divisor $D = \sum_\rho a_\rho D_\rho$ ($a \in \mathbb{Z}^{|\Sigma(1)|}$)} & Height vector $\omega \in \mathbb{Z}^{|\mathbf{A}|}$ \\
        \hline
        \multicolumn{2}{|>{\centering\arraybackslash}m{7.5cm}|}{Effective torus-invariant divisor ($a \in \mathbb{Z}_{\geq0}^{|\Sigma(1)|}$)} & Height vector $\omega \in \mathbb{Z}_{\geq0}^{|\mathbf{A}|}$ \\
        \hline
        \multicolumn{2}{|>{\centering\arraybackslash}m{7.5cm}|}{Torus-invariant divisor with trivial class ($\beta(D) = 0$)} & Linear evaluation $\mathbf{A}^T \psi$ \\
        \hline
        \multicolumn{2}{|>{\centering\arraybackslash}m{7.5cm}|}{Torus-invariant divisor with class in effective cone ($\beta(D) \in \mathrm{Eff}(V_\Sigma)$)} & Valid height vector $\omega$: $\exists \psi$ s.t. $\omega+\mathbf{A}^T\psi\geq0$ \\
        \hline
        \multicolumn{2}{|>{\centering\arraybackslash}m{7.5cm}|}{Flip/flipping contraction} & Flip with signature $(a,b)$ s.t. $a, b > 1$ \\
        \hline
        \multicolumn{2}{|>{\centering\arraybackslash}m{7.5cm}|}{Divisorial contraction} & Deletion flip (signature $(a,1)$) \\
        \hline
        \multicolumn{2}{|>{\centering\arraybackslash}m{7.5cm}|}{Fibering contraction} & Non-flippable circuit \quad  (signature $(a,0)$) \\
        \hline
        \multicolumn{2}{|>{\centering\arraybackslash}m{7.5cm}|}{Pullback of secondary fan by class group grading $\beta$} & Secondary fan \\
        \hline
        \multicolumn{2}{|>{\centering\arraybackslash}m{7.5cm}|}{Secondary fan} & Chamber fan \\
        \hline
        \multicolumn{2}{|>{\centering\arraybackslash}m{7.5cm}|}{Effective cone $\mathrm{Eff}(V_\Sigma)$} & Support of chamber fan \\ 
        \hline
        \multicolumn{2}{|>{\centering\arraybackslash}m{7.5cm}|}{Moving cone $\mathrm{Mov}(V_\Sigma)$} & Union of cones in chamber fan assoc. to fine triangulations  \\
        \hline
        \multicolumn{2}{|>{\centering\arraybackslash}m{7.5cm}|}{Nef/K\"ahler cone $\Gamma(V_\Sigma)$} & Cone in chamber fan assoc. to fine triangulation \\
        \hline
    \end{tabular}
    }
    \caption{Dictionary between triangulation theory (as introduced in \cref{sec:triang}) and toric geometry (as introduced in \cref{sec:toric}). 
    }
    \label{tab:dict}
\end{table}
Now let us turn our attention to the motivating example. Aspects of this example were studied in \cite{Gross1997DeformationObstructed} and subsequently in \cite{Ruan1996TopologicalSigma}, though for largely unrelated reasons. However, smoothness of the CY hypersurface arising from the particular relevant non-FRST fan was shown in Lemma 1.1 of \cite{Gross1997DeformationObstructed}.
\begin{example}
    Let us consider a simple example of a 4D reflexive polytope whose FRSTs do not give rise to a complete birational equivalence class of CY threefolds (i.e., the entire K\"ahler moduli space is not inherited from ambient toric varieties associated to FRSTs). This will motivate us to consider a vex triangulation. For brevity, we will use some standard terminology in this section without definition, but this nomenclature will be clarified in the remainder of the text. 
    
    Consider the four-dimensional reflexive polytope $\Delta^\circ$ (i.e., an element of the Kreuzer--Skarke database) with non-zero points given by the columns of the following matrix.
    \begin{equation}
        \begin{pNiceMatrix}[first-row]
             1 &  2 & 3 & 4 & 5 & 6 \\
            -1 & -1 & 0 & 0 & 0 & 1 \\
            -1 &  1 & 0 & 0 & 1 & 0 \\
            -1 &  0 & 0 & 1 & 0 & 0 \\
             0 & -1 & 1 & 0 & 0 & 0
        \end{pNiceMatrix}.
        \label{eq:example_polytope}
    \end{equation}
    This polytope $\Delta^\circ$ has a unique FRST. By inducing a fan from this triangulation in the traditional way --- taking the cone over each simplex --- the (indices of the) points generating each of its maximal cones are the columns of the following matrix.
    \begin{equation}
        \begin{pmatrix}
            1 & 1 & 1 & 1 & 1 & 1 & 2 & 2 & 3 \\
            2 & 2 & 2 & 2 & 3 & 3 & 3 & 4 & 4 \\
            3 & 3 & 4 & 5 & 4 & 5 & 4 & 5 & 5 \\
            4 & 5 & 6 & 6 & 6 & 6 & 5 & 6 & 6
        \end{pmatrix}.
    \end{equation}
    Let $V$ denote the induced toric variety and $X$ its anticanonical Calabi--Yau hypersurface. We would like to construct the extended K\"ahler cone of $X$, which includes a K\"ahler cone for each element in the birational equivalence class of $X$. Let us choose the following representation $Q$ for the GLSM charge matrix / class group grading. In particular, this sets a basis for divisor classes.
    \begin{equation}
        Q = \begin{pmatrix}
            1 & 1 & 1 & 1 & 0 & 2 \\
            1 & 0 & 0 & 1 & 1 & 1
        \end{pmatrix}.
    \end{equation}
    In this basis, the K\"ahler cone of the ambient variety $V$ --- a subcone of the K\"ahler cone of $X$ --- is $K = \cone{(1,1), (2,1)}$, where we let $\mathrm{cone}$ denote the conical/positive hull. In this particular example, $K$ is exactly the K\"ahler cone of $X$ (this can be verified, for example, by computing Gopakumar-Vafa (GV) invariants \cite{Gendler:2022ztv,Demirtas2024}, as enabled by \texttt{CYTools} \cite{Demirtas:2022hqf} and the \texttt{cygv} package \cite{cygv}).
    
    We now turn our attention to the extended K\"ahler cone $\mathcal{K}$ of $X$: this can be constructed, for example, using the GV-based algorithm described in \cite{Gendler:2022ztv}. Because this polytope admitted a unique FRST, any other CY manifold belonging to the birational equivalence class of $X$ would not traditionally be associated to a hypersurface in a toric variety. That is, in the language of \cite{Gendler:2022ztv, Gendler:2023ujl}, they would be considered ``non-toric phases.'' For this particular polytope, we do in fact have that $X$ admits a flop to a second CY $\Tilde{X}$ with K\"ahler cone $\Tilde{K} = \cone{(2,1), (1,0)}$. Neither geometry admits further flops, so the extended K\"ahler cone is $\mathcal{K} = K \cup \Tilde{K} = \cone{(1,1), (1,0)}$.
    
    This polytope admits another regular simplicial fan which doesn't correspond to an FRST (e.g., the simplices induced by the fan do not cover $\conv{\Delta^\circ}$). That is, this polytope admits a vex triangulation, which gives rise to a toric variety $V'$. The indices of the points generating each of its maximal cones are the columns of the following matrix.
    \begin{equation}
        \begin{pmatrix}
            1 & 1 & 1 & 1 & 1 & 1 & 2 & 3 \\
            2 & 2 & 2 & 3 & 3 & 3 & 4 & 4 \\
            4 & 4 & 5 & 4 & 4 & 5 & 5 & 5 \\
            5 & 6 & 6 & 5 & 6 & 6 & 6 & 6
        \end{pmatrix}
        .
    \end{equation}
    We can compute the K\"ahler cone $K'$ of the ambient variety again: in particular, we find $K' = \cone{(2,1), (1,0)}$, which is exactly what we found for the ``non-toric phase'' moments ago. What's more, comparing the topological data of $X'$ --- the triple intersection numbers and second Chern class, which uniquely determine $\Tilde{X}$ up to diffeomorphism by Wall's theorem \cite{wall_classification_1966} --- to that of the ``non-toric phase'' $\Tilde{X}$ reveals that the two CY manifolds are the same: $\Tilde{X} = X'$. We comment in passing that this is the variety denoted $X_2$ in \cite{Gross1997DeformationObstructed, Ruan1996TopologicalSigma}.
    
    One is led to conclude that $\Tilde{X}$ isn't so ``non-toric'' after all: it merely corresponded to a non-FRST fan --- a vex triangulation. In summary, by considering all fine regular simplicial fans arising from the polytope $\Delta^\circ$, we were able to identify more topological classes and a greater portion of the extended K\"ahler cone than if we had restricted to FRSTs. Indeed, in this example we torically realized the full birational equivalence class and all of the extended K\"ahler cone. This motivates the study of fans beyond FRSTs we undertake in this work.
    \label{ex:intro}
\end{example}

\section{Review of Triangulations}
\label{sec:triang}

We now review the theory of triangulations, in particular those of convex polytopes/cones, following \cite{De_Loera2010-ss}. This section should be viewed as review of relevant material in \cite{De_Loera2010-ss}, except for a few comments particular to the triangulations relevant for our applications in toric and Calabi--Yau geometry. Such an extended discussion is provided since it is anticipated that ``vector configurations'' are new for much of the audience. For those familiar with FRSTs, this section should feel like a generalization of familiar concepts.

\subsection{Configurations and their Subdivisions}

Convex polytopes/cones can be formally described using the language of configurations. Define a \textit{configuration} $\mathbf{A}\subset\mathbb{R}^m$ as a finite, labeled collection of points $x\in\mathbb{R}^m$. For ease of discussion, order the points arbitrarily and let the $i$th point in said configuration have label $j_i$. One collects these labels into the label set $J = \{j_1, \dots, j_{|\mathbf{A}|}\}$. Typically, take $j_i = i$ so that the label set is $J = \{1, \dots, |\mathbf{A}|\}$. 
We make one exception: the origin (if included in $\mathbf{A}$) will always have label $0$. 
Treat a configuration $\mathbf{A}$ both as a matrix
\begin{equation}
    \mathbf{A}= \begin{bmatrix}
        \mathbf{A}_{j_1} & \cdots & \mathbf{A}_{j_{|\mathbf{A}|}}
    \end{bmatrix}
\end{equation}
and as a set
\begin{equation}
    \mathbf{A} = \{ \mathbf{A}_{j_1}, \dots, \mathbf{A}_{j_{|\mathbf{A}|}} \}.
\end{equation}
By this, we mean allow operations both like matrix-vector multiplication $\mathbf{A}\lambda = v\in\mathbb{R}^m$ and set intersection $\mathbf{A} \cap \mathbf{A}'$, for example.

Configurations describe geometric objects. There are two classes of configurations, depending on the type of object they describe. A \textit{point configuration} describes a convex polytope given by its convex hull,
\begin{equation}
    \conv{\mathbf{A}_\mathrm{PC}} = \setbuilder{\mathbf{\mathbf{A}_\mathrm{PC}\lambda}}{\lambda\geq0, \sum_{j\in J} \lambda_j = 1},
\end{equation}
while a \textit{vector configuration} describes a convex cone given by its positive hull,
\begin{equation}
    \cone{\mathbf{A}_\mathrm{VC}} = \setbuilder{\mathbf{\mathbf{A}_\mathrm{VC}\lambda}}{\lambda\geq0}.
\end{equation}
Here and whenever ambiguous, we subscript $\mathbf{A}_\mathrm{PC}$ or $\mathbf{A}_\mathrm{VC}$ to stress that $\mathbf{A}$ should be viewed as a point or vector configuration, respectively.

While convex polytopes/cones are interesting in their own right \cite{Ziegler:1995}, we introduce them as objects-to-be-triangulated. Let $\supp{\mathbf{A}}$ indicate the support of the configuration $\mathbf{A}$, either $\conv{\mathbf{A}_\mathrm{PC}}$ or $\cone{\mathbf{A}_\mathrm{VC}}$. 
\begin{defn}
\label{def:subdivision}
    A \textit{subdivision} of $\mathbf{A}$ is a collection $\mathscr{S}(\mathbf{A}) = \{F_1, F_2,\dots\}$ of subsets $F_i\subseteq\mathbf{A}$ called \textit{cells} such that
    \begin{enumerate}
        \item $\relint{\supp{F_i}}\cap \relint{\supp{F_j}}= \varnothing$ for all $F_i, F_j\in\mathscr{S}$ distinct,
        \item if $G$ is a face of $F\in\mathscr{S}$, then $G\in\mathscr{S}$, and
        \item $\bigcup_{F\in\mathscr{S}} \supp{F} = \supp{\mathbf{A}}$
    \end{enumerate}
    where $\relint{P}$ indicates the relative interior of the region $P$. The first two conditions ensure that $\mathscr{S}$ defines a polyhedral complex, while the last condition ensures that the support of this complex equals $\supp{\mathbf{A}}$.
\end{defn}
See any of \cref{fig:subdivision_triangulation,fig:diagonal_flip,fig:star_link,fig:starpc=vc} for examples. 
Treat subdivisions $\mathscr{S}$ as sets of cells. E.g., $\mathscr{S}'\subseteq\mathscr{S}$ indicates every $\sigma\in\mathscr{S}'$ is also in $\mathscr{S}$. Interchangeably treat elements $F\in\mathscr{S}$ both as geometric regions $\supp{F}$ and as the label subsets $F\subseteq J$ defining said regions. A \textit{triangulation} is a subdivision into simplicial cells (i.e., simplices or simplicial cones) in which case we denote it as $\mathscr{T}$. We will typically just specify the maximal cells in the subdivision; implicitly, one must add all of their faces, too.

\def\dx{1.7} 
\def\dy{1.5} 

\begin{figure}
    \centering
    \begin{subfigure}[b]{0.48\textwidth}
        \centering
        \begin{tikzpicture}[scale=1.8, every node/.style={circle, fill=black, inner sep=1.5pt}]

            \begin{scope}[shift={(0,0)}, scale=0.4]
                \coordinate (v1) at (1,1);
                \coordinate (v2) at (1,3);
                \coordinate (v3) at (2,3);
                \coordinate (v4) at (3,1);
                \coordinate (v5) at (3,2);
                \coordinate (v6) at (4,3);
                
                \node[
                label=below:{$1$}
                ] at (v1) {};
                \node[
                label=above:{$2$}
                ] at (v2) {};
                \node[
                label=above:{$3$}
                ] at (v3) {};
                \node[
                label=below:{$4$}
                ] at (v4) {};
                \node[
                label=above:{$5$}
                ] at (v5) {};
                \node[
                label=above:{$6$}
                ] at (v6) {};
                
                \draw (v1) -- (v2) -- (v3) -- cycle;
                \draw (v1) -- (v3) -- (v5) -- (v4) -- cycle;
                \draw (v3) -- (v5) -- (v6) -- cycle;
                \draw (v4) -- (v5) -- (v6) -- cycle;
                
                \node[draw=none, fill=none, circle=none] at (2.1,1.9) {F};
            \end{scope}
                
            \begin{scope}[shift={(\dx,0)}, scale=0.4]
                \coordinate (v1) at (1,1);
                \coordinate (v2) at (1,3);
                \coordinate (v3) at (2,3);
                \coordinate (v4) at (3,1);
                \coordinate (v5) at (3,2);
                \coordinate (v6) at (4,3);
                
                \node[] at (v1) {};
                \node[] at (v2) {};
                \node[] at (v3) {};
                \node[] at (v4) {};
                \node[] at (v5) {};
                \node[] at (v6) {};
                
                \draw (v1) -- (v2) -- (v3) -- cycle;
                \draw (v1) -- (v3) -- (v4) -- cycle;
                \draw (v3) -- (v4) -- (v5) -- cycle;
                \draw (v3) -- (v5) -- (v6) -- cycle;
                \draw (v4) -- (v5) -- (v6) -- cycle;
            \end{scope}
                
                \begin{scope}[shift={(0,-\dy)}, scale=0.4]
                \coordinate (v1) at (1,1);
                \coordinate (v2) at (1,3);
                \coordinate (v3) at (2,3);
                \coordinate (v4) at (3,1);
                \coordinate (v5) at (3,2);
                \coordinate (v6) at (4,3);
                
                \node[] at (v1) {};
                \node[] at (v2) {};
                \node[] at (v3) {};
                \node[] at (v4) {};
                \node[draw, fill=none] at (v5) {};
                \node[] at (v6) {};
                
                \draw (v1) -- (v2) -- (v3) -- cycle;
                \draw (v1) -- (v3) -- (v6) -- (v4) -- cycle;
                
                \node[draw=none, fill=none, circle=none] at (2.1,1.9) {F};
            \end{scope}
                
            \begin{scope}[shift={(\dx,-\dy)}, scale=0.4]
                \coordinate (v1) at (1,1);
                \coordinate (v2) at (1,3);
                \coordinate (v3) at (2,3);
                \coordinate (v4) at (3,1);
                \coordinate (v5) at (3,2);
                \coordinate (v6) at (4,3);
                
                \node[] at (v1) {};
                \node[] at (v2) {};
                \node[] at (v3) {};
                \node[] at (v4) {};
                \node[draw, fill=none] at (v5) {};
                \node[] at (v6) {};
                
                \draw (v1) -- (v2) -- (v3) -- cycle;
                \draw (v1) -- (v3) -- (v4) -- cycle;
                \draw (v3) -- (v4) -- (v6) -- cycle;
            \end{scope}
        \end{tikzpicture}
        \subcaption[Subdivisions of $\mathbf{A}$]{%
  Subdivisions of
  {\par\centering
  \vspace{0.5em}
  $\mathbf{A} =
  \begin{bNiceMatrix}
    1 & 1 & 2 & 3 & 3 & 4\\
    1 & 3 & 3 & 1 & 2 & 3
  \end{bNiceMatrix}.$
  \vspace{0.5em}
  \par}
  Left: non-triangulation subdivisions, characterized by a non-simplicial cell $F$;
  right: triangulations ``refining'' these subdivisions.
  Top: ``fine'' subdivisions; bottom: non-``fine'' subdivisions, characterized by missing point $5$.%
}
        \label{fig:subdivision_triangulation_a}
    \end{subfigure}
    \hfill
    \begin{subfigure}[b]{0.48\textwidth}
        \centering
        \begin{tikzpicture}[scale=1.8, every node/.style={circle, fill=black, inner sep=1.5pt}]

            \newcommand{\Ray}[8][]{%
                \pgfmathsetmacro{\Px}{#2/(#4)}
                \pgfmathsetmacro{\Py}{#3/(#4)}
                \coordinate (P#5) at (\Px,\Py,1);
                \draw[ray,->,#1] (#7) -- (P#5)
                \ifx&#8&\else
                    node[fill=none,#6,lab] {#8}
                \fi;
            }
            
            \tikzset{
                x={(1cm,0.0cm)},
                y={(0.0cm,1cm)},
                z={(0.0cm,0.5cm)}
            }

            \begin{scope}[shift={(0,0)}, scale=0.4]
                \coordinate (O) at (0,0,0);
                
                \tikzset{ray/.style={->, thick}, lab/.style={font=\small}}
                \Ray[thick]{1}{1}{1}{r1}{below}{O}{1}
                \Ray[thick]{1}{2}{1}{r2}{above left}{O}{2}
                \Ray[thick]{2}{1}{1}{r3}{below}{O}{3}
                \Ray[very thin, dashed, -]{2}{2}{1}{r4}{above left}{O}{4}
                \Ray[very thin, dashed, -]{2}{3}{1}{r5}{above left}{O}{5}
                \Ray[thick]{3}{1}{1}{r6}{right}{O}{6}
        
                \draw[thick] (Pr1) -- (Pr2) -- (Pr4) -- (Pr3) -- cycle;
                \draw[thick] (Pr2) -- (Pr4) -- (Pr5) -- cycle;
                \draw[thick] (Pr3) -- (Pr4) -- (Pr6) -- cycle;
                \draw[thick] (Pr4) -- (Pr5) -- (Pr6) -- cycle;
                
                \node[draw=none, fill=none, circle=none] at (1.6,2) {\textbf{\textcolor{white}{F}}};
                \node[draw=none, fill=none, circle=none] at (1.6,2) {F};
            \end{scope}
                
            \begin{scope}[shift={(\dx,0)}, scale=0.4]
                \coordinate (O) at (0,0,0);
                
                \tikzset{ray/.style={->, thick}, lab/.style={font=\small}}
                \Ray[thick]{1}{1}{1}{r1}{below}{O}{}
                \Ray[thick]{1}{2}{1}{r2}{above left}{O}{}
                \Ray[thick]{2}{1}{1}{r3}{below}{O}{}
                \Ray[very thin, dashed]{2}{2}{1}{r4}{above right}{O}{}
                \Ray[very thin, dashed]{2}{3}{1}{r5}{above left}{O}{}
                \Ray[thick]{3}{1}{1}{r6}{right}{O}{}
        
                \draw[thick] (Pr1) -- (Pr2) -- (Pr4) -- cycle;
                \draw[thick] (Pr1) -- (Pr3) -- (Pr4) -- cycle;
                \draw[thick] (Pr2) -- (Pr4) -- (Pr5) -- cycle;
                \draw[thick] (Pr3) -- (Pr4) -- (Pr6) -- cycle;
                \draw[thick] (Pr4) -- (Pr5) -- (Pr6) -- cycle;
            \end{scope}
                
            \begin{scope}[shift={(0,-\dy)}, scale=0.4]
                \coordinate (O) at (0,0,0);
                
                \tikzset{ray/.style={->, thick}, lab/.style={font=\small}}
                \Ray[thick]{1}{1}{1}{r1}{below}{O}{}
                \Ray[thick]{1}{2}{1}{r2}{above left}{O}{}
                \Ray[very thin, dashed]{2}{2}{1}{r4}{above left}{O}{}
                \Ray[very thin, dashed]{2}{3}{1}{r5}{above left}{O}{}
                \Ray[thick]{3}{1}{1}{r6}{right}{O}{}
        
                \draw[thick] (Pr1) -- (Pr2) -- (Pr4) -- (Pr6) -- cycle;
                \draw[thick] (Pr2) -- (Pr4) -- (Pr5) -- cycle;
                \draw[thick] (Pr4) -- (Pr5) -- (Pr6) -- cycle;
                
                \node[draw=none, fill=none, circle=none,
                text=black, double=white, double distance=1pt] at (1.6,2) {F};
            \end{scope}
                
            \begin{scope}[shift={(\dx,-\dy)}, scale=0.4]
                \coordinate (O) at (0,0,0);
                
                \tikzset{ray/.style={->, thick}, lab/.style={font=\small}}
                \Ray[thick]{1}{1}{1}{r1}{below}{O}{}
                \Ray[thick]{1}{2}{1}{r2}{above left}{O}{}
                \Ray[very thin, dashed]{2}{2}{1}{r4}{above right}{O}{}
                \Ray[very thin, dashed]{2}{3}{1}{r5}{above left}{O}{}
                \Ray[thick]{3}{1}{1}{r6}{right}{O}{}
        
                \draw[thick] (Pr1) -- (Pr2) -- (Pr4) -- cycle;
                \draw[thick] (Pr1) -- (Pr4) -- (Pr6) -- cycle;
                \draw[thick] (Pr2) -- (Pr4) -- (Pr5) -- cycle;
                \draw[thick] (Pr4) -- (Pr5) -- (Pr6) -- cycle;
            \end{scope}
        \end{tikzpicture}
        \caption[Acyclic vector configuration]{Analogous to \cref{fig:subdivision_triangulation_a} but for an ``acyclic'' vector configuration 
        {\par\centering
          \vspace{0.5em}
          $\mathbf{A} =
          \begin{bNiceMatrix}
                1 & 1 & 2 & 2 & 2 & 3\\
                1 & 2 & 1 & 2 & 3 & 1\\
                1 & 1 & 1 & 1 & 1 & 1\\
            \end{bNiceMatrix}.$
          \vspace{0.5em}
          \par} That is, plotted are polyhedral complexes subdividing $\ConeOp(\mathbf{A})$. Since $\mathbf{A}$ is acyclic (i.e., $\ConeOp(\mathbf{A})$ is pointed), these subdivisions can be seen as subdivisions of a cross section of $\ConeOp(\mathbf{A})$.}
        \label{fig:subdivision_triangulation_b}
    \end{subfigure}

    \caption{Subdivisions and triangulations of a point configuration and a vector configuration.}
    \label{fig:subdivision_triangulation}
\end{figure}

\begin{example}
    See \cref{fig:subdivision_triangulation_a} for examples of subdivisions of the point configuration
    \begin{equation}
        \mathbf{A} = \begin{bNiceMatrix}[first-row]
        1 & 2 & 3 & 4 & 5 & 6\\
        1 & 1 & 2 & 3 & 3 & 4\\
        1 & 3 & 3 & 1 & 2 & 3
    \end{bNiceMatrix}.
    \end{equation}
    The subdivisions on the left of \cref{fig:subdivision_triangulation_a} are not triangulations, characterized by containing a non-simplicial cell $F$. One can write, e.g., the top-left subdivision as
    \begin{equation}
        \mathscr{S} = \begin{bNiceMatrix}[first-row]
              & F &   &  \\
            1 & 1 & 3 & 4\\
            2 & 3 & 5 & 5\\
            3 & 4 & 6 & 6\\
            - & 5 & - & -
        \end{bNiceMatrix}
    \end{equation}
    where we have highlighted the column corresponding to the non-simplicial cell $F$. This cell, e.g., corresponds to a geometric region $F = \conv{\{\mathbf{A}_1, \mathbf{A}_3, \mathbf{A}_4, \mathbf{A}_5\}}$.
    
    The bottom row consists of ``non-fine'' subdivisions (to be discussed immediately), characterized by their absence of $\mathbf{A}_5$.
\end{example}

We stress a point here. As defined, a subdivision is a collection of subsets of $\mathbf{A}$. These subsets define geometric regions which decompose $\SuppOp(\mathbf{A})$. There can be multiple distinct subdivisions, however, which define the same decomposition into geometric regions, such as
\begin{align}
    \mathscr{T} = \begin{bmatrix}
        1 & 1 & 3 \\
        2 & 3 & 4 \\
        3 & 4 & 6
    \end{bmatrix}
    & \quad & \text{and} & \quad &
    \mathscr{S} = \begin{bmatrix}
        1 & 1 & 3 \\
        2 & 3 & 4 \\
        3 & 4 & 5 \\
        - & - & 6
    \end{bmatrix}
\end{align}
of $\mathbf{A}$ from \cref{fig:subdivision_triangulation_a}. The distinction is the inclusion of point $5$ in the relative interior of the simplex $\{3,4,6\}$. As defined, $\mathscr{T}\neq\mathscr{S}$ despite these two subdivisions defining the same decomposition of $\ConvOp(\mathbf{A})$. This distinction will prove crucial when studying the ``regularity'' of triangulations --- using language from \cref{sec:regular}, $\mathscr{T}$ will have a different secondary cone than $\mathscr{S}$. 
In connecting to toric geometry, however, we will often need to discuss the geometric decompositions of $\ConeOp(\mathbf{A})$, which we will subsequently call ``fans''. 

\begin{defn}
\label{def:fan}
    A (polyhedral) \textit{fan} is a collection $K = \{C_1,C_2,\dots\}$ of convex cones such that
    \begin{enumerate}
        \item $\RelintOp(C_i) \cap \RelintOp(C_j) = \varnothing$ for all $C_i, C_j\in K$ distinct and
        \item if $G$ is a face of $C\in K$, then $G\in K$.
    \end{enumerate}
\end{defn}

That is, \cref{def:subdivision} describes an abstract complex in contrast to \cref{def:fan} which describes a geometric complex. This distinction is often not made and will only matter for us in \cref{sec:frst_cy} and \cref{sec:translate}, so we will generally use the term ``subdivision'' of $\mathbf{A}_\mathrm{VC}$ and ``fan'' interchangeably unless otherwise specified.

We now define a few properties of triangulations (which apply to subdivisions, too). Consider a triangulation $\mathscr{T}(\mathbf{A})$ of a configuration $\mathbf{A}$ with label set $J$. As alluded to in \cref{fig:subdivision_triangulation}, one calls $\mathscr{T}(\mathbf{A})$ \textit{fine} if $\mathscr{T}(\mathbf{A})$ ``uses'' all elements $\mathbf{A}_j$. More formally,
\begin{equation}
    \mathscr{T}(\mathbf{A}) \text{ is fine} \iff J = \bigcup_{\sigma\in\mathscr{T}}\sigma
\end{equation}
Note, $\mathscr{T}(\mathbf{A})$ must still be a valid subdivision even if it is not fine, so it must subdivide $\supp{\mathbf{A}}$. Thus every extremal element $v$ of $\mathbf{A}$ (i.e., $v$ such that $\SuppOp(\mathbf{A}\setminus v)\neq \SuppOp(\mathbf{A})$) must appear in some simplex. Also note that fineness depends on the associated configuration. E.g., a triangulation $\mathscr{T}(\mathbf{A})$ can always be viewed as fine in the subconfiguration $\mathbf{A}'\subseteq\mathbf{A}$ defined as $\mathbf{A}' = \setbuilder{\mathbf{A}_j}{j\in\bigcup_{\sigma\in\mathscr{T}}\sigma}$.

Now, consider a point configuration (not a vector configuration) $\mathbf{A}_\mathrm{PC}$ containing the origin $\mathbf{0}\in\mathbf{A}_\mathrm{PC}$ (as always in our convention, with label $0$). Call a triangulation $\mathscr{T}(\mathbf{A}_\mathrm{PC})$ \textit{star} if the origin is a vertex of every simplex,
\begin{equation}
\label{eq:star_triang}
    \mathscr{T}(\mathbf{A}_\mathrm{PC}) \text{ is star} \iff 0 \in  \bigcap_{\text{maximal }\sigma\in\mathscr{T}}\sigma.
\end{equation}
A star triangulation $\mathscr{T}$ may be interpreted as/extended to a simplicial fan by replacing each simplex $\sigma$ with a simplicial $\cone{\sigma} = \setbuilder{sx}{x\in\sigma,s\geq0}$. This is the utility of the star property. One typically does not speak of ``star'' triangulations of vector configurations since, by definition, they already are simplicial fans and hence every simplex already contains $\mathbf{0}$.\footnote{There is some hairiness in what is meant here. What is explicitly meant is that $\sigma\in\mathscr{T}$ implies that $\mathbf{0}\in\SuppOp(\sigma)$. This does not mean that $0\in\sigma$ (i.e., that one treats $\mathbf{0}$ as a generator of $\sigma$). In fact, we say that $\sigma$ is simplicial if and only if $|\sigma|=\dim(\sigma)$ and hence, if $0\in\sigma$ (this means as a generator), then $\sigma$ is never simplicial. These are some of the issues with adding the origin $\mathbf{0}$ to a vector configuration.}

\subsection{Regularity and the Secondary Cone}
\label{sec:regular}

The final property of interest is that of ``regularity''. We define regularity, introduce the ``secondary cone'' of heights, and then demonstrate a construction (using ``circuits'') of this cone.

There are a multitude of ways to define regularity; the definition relevant for our purposes is that $\mathscr{T}$ is regular if it can be obtained via the following lifting procedure. Consider a configuration $\mathbf{A}$ and a vector $\omega\in\mathbb{R}^{|\mathbf{A}|}$. Define the following new configuration
\begin{equation}
    \mathbf{A}^\omega = \begin{bmatrix}
        \mathbf{A} \\
        \omega
    \end{bmatrix}.
\end{equation}
Call this the ``lifted configuration''. This is equivalent to embedding $\mathbf{A}$ in the subspace $\mathbb{R}^m$ of $\mathbb{R}^{m+1}$ given by the first $m$ coordinates and then ``lifting'' the point $j$ as $\mathbf{A}_j \to (\mathbf{A}_j, \omega_j)$. Consider the ``lower'' facets of $\supp{\mathbf{A}^\omega}$ --- those for which their defining (inwards-facing) normal $n$ has a positive final component in $\mathbb{R}^{m+1}$, $n_{m+1}>0$. These facets (along with their faces) define a polyhedral complex which, if projected down to the $\mathbb{R}^m$ subspace, defines a decomposition of $\supp{\mathbf{A}}$ --- a subdivision. Denote this subdivision as $\mathscr{S}(\mathbf{A},\omega)$ or, if it is a triangulation, $\mathscr{T}(\mathbf{A},\omega)$. See \cref{fig:lifting}. A subdivision/triangulation is said to be \textit{regular} if and only if such a ``height vector'' $\omega$ exists,
\begin{equation}
    \mathscr{S}(\mathbf{A}) \text{ is regular} \iff \exists\omega \text{ such that }\mathscr{S}=\mathscr{S}(\mathbf{A},\omega).
\end{equation}
For point configurations, this lifting construction works for any height vector $\omega\in\mathbb{R}^{|\mathbf{A}_\mathrm{PC}|}$ --- no matter what one chooses for $\omega$, there are lower faces of $\conv{\mathbf{A}_\mathrm{PC}^\omega}$ which define a subdivision $\mathscr{S}(\mathbf{A}_\mathrm{PC},\omega)$ of $\mathbf{A}_\mathrm{PC}$. More care is needed for vector configurations: some heights $\omega\in\mathbb{R}^{|\mathbf{A}_\mathrm{VC}|}$ do not define regular subdivisions\footnote{An aside, assuming the language of the rest of this section: a flip of a circuit can be understood as $\mathscr{T}(\mathbf{A},\omega)\to\mathscr{T}(\mathbf{A},-\omega)$. This has the interpretation of projecting down either the lower or the upper facets of $\mathbf{A}^\omega$. In \cref{ex:bad_heights}, we demonstrate that some heights $\omega$ define regular triangulations of circuits while $-\omega$ do not. I.e., the lower facets define a subdivision but the upper facets do not. This is equivalent to the circuit not being flippable. This is a geometric way to see why vector configurations can have non-flippable circuits while point configurations cannot.} of $\mathbf{A}_\mathrm{VC}$. See \cref{ex:bad_heights}. In general, one can show (\cite{De_Loera2010-ss}, Thm. 4.1.39) that height vectors $\omega$ generate regular subdivisions/triangulations of a vector configuration if and only if $\exists\psi\in\mathbb{R}^m$ 
such that $\omega + \mathbf{A}_\mathrm{VC}^T\psi \geq 0$, in which case the vector $\omega + \mathbf{A}_\mathrm{VC}^T\psi$ generates the same subdivision as $\omega$. One calls $\mathbf{A}_\mathrm{VC}^T\psi \geq 0$ a ``linear evaluation'' of $\mathbf{A}_\mathrm{VC}$ so one can restate this condition as: $\omega$ has to be non-negative, modulo a linear evaluation of $\mathbf{A}_\mathrm{VC}$. This implies both that non-negative heights $\omega\in\mathbb{R}_{\geq0}^{|\mathbf{A}_\mathrm{VC}|}$ always generate regular subdivisions and that every regular triangulation of a vector configuration $\mathscr{T}(\mathbf{A})$ can be generated by non-negative heights. 

\begin{figure}[t]
    \centering
    \resizebox{0.6\textwidth}{!}{
    \begin{tikzpicture}[scale=3.5]
        \draw[black, semithick] (0, 0) -- (1, 0) -- (0+0.6, 0.5) -- cycle;
        \draw[black, semithick] (1+0.6, 0.5) -- (1, 0) -- (0+0.6, 0.5) -- cycle;
    
        \fill[blue,opacity=0.1] (0, 0+1.1) -- (1, 0+0.2) -- (0+0.6, 0.5+0.3) -- cycle;
        \fill[blue,opacity=0.1] (1+0.6, 0.5+0.9) -- (1, 0+0.2) -- (0+0.6, 0.5+0.3) -- cycle;
        \draw[blue] (0, 0+1.1) -- (1, 0+0.2) -- (0+0.6, 0.5+0.3) -- cycle;
        \draw[blue] (1+0.6, 0.5+0.9) -- (1, 0+0.2) -- (0+0.6, 0.5+0.3) -- cycle;
    
        \foreach \x in {0,1} {
            \foreach \y in {0,0.5} {
                \pgfmathsetmacro\height{
                    ifthenelse(\x==0 && \y==0, 1.1,
                    ifthenelse(\x==1 && \y==0, 0.2,
                    ifthenelse(\x==1 && \y==0.5, 0.9, 
                    ifthenelse(\x==0 && \y==0.5, 0.3, ))))
                }
                \draw[thin, gray] (\x+1.176*\y,\y) -- (\x+1.176*\y,\y+\height);
            }
        }
    
        \node[fill=black, circle, inner sep=1.3pt] at (0,0) {};
        \node[fill=black, circle, inner sep=1.3pt] at (1,0) {};
        \node[fill=black, circle, inner sep=1.3pt] at (1+1.176*0.5,0.5) {};
        \node[fill=black, circle, inner sep=1.3pt] at (0+1.176*0.5,0.5) {};

        \node[below left] at (0,0) {$p_1$};
        \node[below right] at (1,0) {$p_3$};
        \node[below right] at (1+1.176*0.5,0.5) {$p_4$};
        \node[left] at (0+1.176*0.5,0.5) {$p_2$};
    
        \node[fill=blue, circle, inner sep=1.3pt] at (0,0+1.1) {};
        \node[fill=blue, circle, inner sep=1.3pt] at (1,0+0.2) {};
        \node[fill=blue, circle, inner sep=1.3pt] at (1+1.176*0.5,0.5+0.9) {};
        \node[fill=blue, circle, inner sep=1.3pt] at (0+1.176*0.5,0.5+0.3) {};

        \node[below left] at (0,0+1.1) {$\tilde{p}_1$};
        \node[right] at (1,0+0.2) {$\tilde{p}_3$};
        \node[below right] at (1+1.176*0.5,0.5+0.9) {$\tilde{p}_4$};
        \node[above] at (0+1.176*0.5,0.5+0.3) {$\tilde{p}_2$};

    \end{tikzpicture}
    }
    \caption{Diagram of the ``lifting'' procedure defining regular triangulations. The points $p_1$, $p_2$, $p_3$, and $p_4$ are embedded into $\mathbb{R}^3$ and then lifted by heights $\omega_1=1.1$, $\omega_2=0.3$, $\omega_3=0.2$, and $\omega_4=0.9$. The convex hull of the lifted point configuration is a $3$-simplex whose lower faces are plotted in blue. Projecting out the lifted coordinate generates the regular triangulation plotted in black. Figure modified from \cite{macfadden2023efficient}.}
    \label{fig:lifting}
\end{figure}
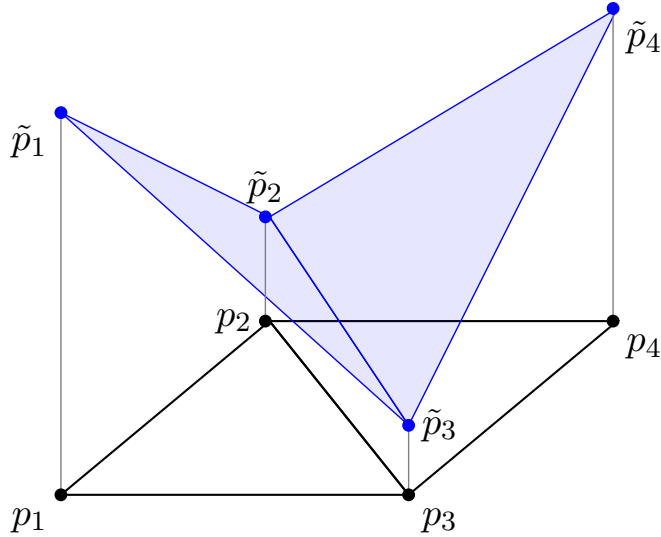

\begin{example}
    To see why some height vectors fail to define regular subdivisions of a vector configuration, consider $\mathbf{A} = \begin{bmatrix} -1 & 1\end{bmatrix}$ and the heights $\omega=(-1,-1)$. Lifting $\mathbf{A}$ by $\omega$, one achieves
    \begin{equation}
        \mathbf{A}^{(-1,-1)} = \begin{bmatrix}
        -1 & 1\\ -1 & -1
    \end{bmatrix}.
    \end{equation}
    For this lifting to define a regular triangulation/subdivision, the configuration $\mathbf{A}^{(-1,-1)}$ must have lower facets which, when projected down to their first component, span $\supp{\mathbf{A}} = \mathbb{R}$. This is not true, though: $\cone{\mathbf{A}^{(-1,-1)}}$ has two facets $F_1=\cone{\{(-1,-1)\}}$ and $F_2 = \cone{\{(1,-1)\}}$, neither of which are lower (they have inwards-facing normals $(1,-1)$ and $(-1,-1)$ respectively)
    . Linear evaluations of $\mathbf{A}$ do not save this height vector: such evaluations are of the form $\mathbf{A^T}s$ for some scalar $s$. Since $\omega+\mathbf{A}^Ts = (-1-s, -1+s)$ always has a negative component, the construction fails.
    
    If one were to instead lift $\mathbf{A}$ with $\omega=(1, 1)$, then
    \begin{equation}
        \mathbf{A}^{(1,1)} = \begin{bmatrix}
        -1 & 1\\ 1 & 1
    \end{bmatrix}.
    \end{equation}
    This is just a mirror of the previous lifted configuration, so both facets of this triangulation are lower, defining the only (regular) triangulation $\mathscr{T}=(\{1\},\{2\})$ of $\mathbf{A}$.
    \label{ex:bad_heights}
\end{example}

It is not difficult to show that if both $\omega$ and $\omega'$ generate the same regular triangulation $\mathscr{T}$, then both $\omega+\omega'$ and $s\omega$ also generate $\mathscr{T}$ for any choice of $s>0$ (\cite{De_Loera2010-ss}, Prop. 5.2.6, Cor. 5.2.8). This is the structure of a cone, so the collection of $\omega$ generating $\mathscr{T}$ defines (the relative interior of) a convex cone called the \textit{secondary cone} $\mathbf{C}(\mathbf{A},\mathscr{T})$,
\begin{equation}
\label{eq:secondary_cone_1}
    \mathbf{C}(\mathbf{A},\mathscr{T}) = \mathrm{cl}\left(\setbuilder{\omega\in\mathbb{R}^{|\mathbf{A}|}}{\mathscr{T}=\mathscr{T}(\mathbf{A},\omega)}\right)
\end{equation}
where $\mathrm{cl}$ indicates the closure. The $\omega\in\partial\,\mathbf{C}(\mathbf{A},\mathscr{T})$ do not define $\mathscr{T}$, but instead define ``coarser'' polyhedral subdivisions $\mathscr{S}$ which can be ``refined'' to $\mathscr{T}$ (\cite{De_Loera2010-ss}, Prop. 5.2.9; see \cref{fig:subdivision_triangulation,fig:diagonal_flip} for examples of refinements). See \cref{ex:simple_secondary_cone} for the construction of the secondary cone of each triangulation of a simple point configuration.

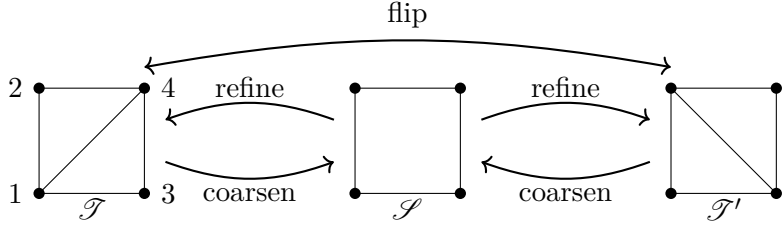
\begin{figure}
    \centering
    \begin{tikzpicture}[scale=1.4, every node/.style={circle, fill=black, inner sep=1.5pt}]
        \def\sx{3} 

        \newcommand{\DrawSquare}[3][]{%
            \coordinate (1) at (#2,#3);
            \coordinate (2) at (#2,#3+1);
            \coordinate (3) at (#2+1,#3);
            \coordinate (4) at (#2+1,#3+1);
            
            \draw[] (1) -- (2) -- (4) -- (3) -- cycle;
            
            \IfSubStr{#1}{diagA}{\draw[] (1) -- (4);}{}
            \IfSubStr{#1}{diagB}{\draw[] (2) -- (3);}{}
            
            \IfSubStr{#1}{labels}{%
            \node[label=left:{$1$}] at (1) {};
            \node[label=left:{$2$}] at (2) {};
            \node[label=right:{$3$}] at (3) {};
            \node[label=right:{$4$}] at (4) {};
            }{
            \node at (1) {};
            \node at (2) {};
            \node at (3) {};
            \node at (4) {};
            }
        }
    
        \DrawSquare[diagA, labels]{0}{0}
        \DrawSquare[]{\sx}{0}
        \DrawSquare[diagB]{2*\sx}{0}

        \node[fill=none] at (0.5,-0.15) {$\mathscr{T}$};
        \node[fill=none] at (1*\sx+0.5,-0.15) {$\mathscr{S}$};
        \node[fill=none] at (2*\sx+0.5,-0.15) {$\mathscr{T}'$};

        \tikzset{ray/.style={->, thick}, lab/.style={font=\small}}
        \draw[ray, <->, bend left=10]
            (1,1.2) to node[fill=none, above=-2pt, midway] {flip}
            (2*\sx,1.2);

        \draw[ray, <-, bend left=20]
            (1.2,0.7) to node[fill=none, above=-10pt, midway] {refine}
            (1*\sx-0.2,0.7);
        \draw[ray, ->, bend right=20]
            (1.2,0.3) to node[fill=none, below=-14pt, midway] {coarsen}
            (1*\sx-0.2,0.3);

        \draw[ray, ->, bend left=20]
            (1*\sx+1.2,0.7) to node[fill=none, above=-10pt, midway] {refine}
            (2*\sx-0.2,0.7);
        \draw[ray, <-, bend right=20]
            (1*\sx+1.2,0.3) to node[fill=none, below=-14pt, midway] {coarsen}
            (2*\sx-0.2,0.3);
    \end{tikzpicture}
    \caption{The two triangulations $\mathscr{T}$ and $\mathscr{T}'$ of the square $\mathbf{A}_\text{PC} = \{(1,1),(1,2),(2,1),(2,2)\}$. Flipping $\mathscr{T}\leftrightarrow\mathscr{T}'$ can be thought as coarsening a triangulation into the subdivision $\mathscr{S}$ and then refining this subdivision in `the other way'.}
    \label{fig:diagonal_flip}
\end{figure}

\begin{example}
    Consider the point configuration $\mathbf{A} = \{(1,1),(1,2),(2,1),(2,2)\}$ which has labels $J=\{1,2,3,4\}$. This configuration has two triangulations, $\mathscr{T} = (\{1,2,3\},\{2,3,4\})$ and $\mathscr{T}'=(\{1,2,4\},\{1,3,4\})$, both of which are regular. See \cref{fig:lifting,fig:diagonal_flip}. Observe that $\mathscr{T}$ is distinguished from $\mathscr{T}'$ by whether the line segment $\{2,3\}$ or $\{1,4\}$ appears. These segments intersect at $\frac{1}{2}\mathbf{A}_1 + \frac{1}{2}\mathbf{A}_4 = \frac{1}{2}\mathbf{A}_2+\frac{1}{2}\mathbf{A}_3$ and this point of intersection can be used to constrain the heights: since a line segment appears in a triangulation if and only if it is a lower face of $\mathbf{A}^\omega$, one observes
    \begin{align}
        \{2,3\} \in \mathscr{T} &\implies \frac{1}{2}\omega_2 + \frac{1}{2}\omega_3 < \frac{1}{2}\omega_1 + \frac{1}{2}\omega_4 \text{ and }\label{eq:simple_secondary_cone}\\
        \{1,4\} \in \mathscr{T}' &\implies \frac{1}{2}\omega_2 + \frac{1}{2}\omega_3 > \frac{1}{2}\omega_1 + \frac{1}{2}\omega_4.
    \end{align}
    The only other case is when $\omega_2+\omega_3 = \omega_1 + \omega_4$, for which all points of $\mathbf{A}^\omega$ are coplanar, defining the regular subdivision $\mathscr{S}=(\{1,2,3,4\})$. These are the only constraints, so the secondary cones of these triangulations/subdivisions are
    \begin{align}
        \mathbf{C}(\mathbf{A},\mathscr{T}) &= \{\omega\in\mathbb{R}^4: \begin{bmatrix} 1 & -1 & -1 & 1\end{bmatrix}\cdot\omega\geq0\},\\
        \mathbf{C}(\mathbf{A},\mathscr{T}') &= \{\omega\in\mathbb{R}^4: \begin{bmatrix} 1 & -1 & -1 & 1\end{bmatrix}\cdot\omega\leq0\}, \text{ and}\\
        \mathbf{C}(\mathbf{A},\mathscr{S}) &= \{\omega\in\mathbb{R}^4: \begin{bmatrix} 1 & -1 & -1 & 1\end{bmatrix}\cdot\omega=0\}.
    \end{align}
    These constraints are best thought of as being ``local'': any number of points $\mathbf{A}'\subset\mathbb{R}^2$ may be added $\mathbf{A}\to\mathbf{A}\cup\mathbf{A}'$ and, as long as $(\{1,2,3\}, \{2,3,4\})\subset\mathscr{T}(\mathbf{A}\cup\mathbf{A}')$, then the same constraint \cref{eq:simple_secondary_cone} holds. What matters is the local structure around the points $\mathbf{A}$.
    \label{ex:simple_secondary_cone}
\end{example}

Recall that cones can be represented either as the conical hull of some generators (a ``V-representation'') or as the intersection of some half-spaces (an ``H-representation''). While \cref{eq:secondary_cone_1} formally defines the secondary cone, it is not as useful as such a V- or H-representation. We will demonstrate how to compute an H-representation of $\mathbf{C}(\mathbf{A},\mathscr{T})$ from the simplices of $\mathscr{T}$, akin to \cref{ex:simple_secondary_cone}. To do this in general, we will want some machinery (the notion of ``flippable circuits'') that we develop in the next couple of sections.

\subsubsection{Dependencies and Circuits}

\label{sec:dep_and_circ}

The notion of dependency relevant to point configurations is encoded by a pair of distinct $\mu,\mu'\in\mathbb{R}_{\geq0}^{|\mathbf{A}_\mathrm{PC}|}$ which satisfy both $\mathbf{A}_\mathrm{PC}\,\mu = \mathbf{A}_\mathrm{PC}\,\mu'$ and $\sum_{j\in J}\mu_j=\sum_{j\in J}\mu'_j=1$. I.e., there is a point $x\in\conv{\mathbf{A}_\mathrm{PC}}$ that can be written as two distinct convex combinations of the points in $\mathbf{A}_\mathrm{PC}$. This is equivalent to an affine dependency of $\mathbf{A}_\mathrm{PC}$,
\begin{align}
    \mathbf{A}_\mathrm{PC}(\mu-\mu') &= 0 \label{eq:pc_dependencies_1}\\
    \sum_{j\in J}(\mu-\mu')_j &= 0. \label{eq:pc_dependencies_2}
\end{align}
Similarly, a dependency of a vector configuration can be encoded by distinct $\mu,\mu'\in\mathbb{R}_{\geq0}^{|\mathbf{A}_\mathrm{VC}|}$ which satisfy $\mathbf{A}_\mathrm{VC}\,\mu = \mathbf{A}_\mathrm{VC}\,\mu'$. There are no constraints on $\sum_{j\in J} \mu_j$ or $\sum_{j\in J} \mu'_j$. I.e., there is a vector $r\in\cone{\mathbf{A}_\mathrm{VC}}$ that can be written as two distinct conical combinations of elements in $\mathbf{A}_\mathrm{VC}$. This is equivalent to a linear dependency of $\mathbf{A}_\mathrm{VC}$,
\begin{equation}
\label{eq:vc_dependencies}
    \mathbf{A}_\mathrm{VC}(\mu-\mu') = 0.
\end{equation}
Note, for both point and vector configurations, one can either specify a dependency as a pair of non-negative vectors $\mu,\mu'$ or as a single mixed-sign vector 
$\lambda=\mu-\mu'$. To map a mixed-sign vector $\lambda$ to a pair of non-negative vectors, consider the sets (analogous to a Radon partition)
\begin{align}
    J_+ &= \setbuilder{j\in J}{\lambda_j>0},\\
    J_- &= \setbuilder{j\in J}{\lambda_j<0}, \text{ and}\\
    J_0 &= \setbuilder{j\in J}{\lambda_j=0}.
\end{align}
Then, the pair of vectors $\lambda_\pm$, defined as
\begin{equation}
    (\lambda_\pm)_j = \begin{cases}
        \lambda_j & j\in J_\pm,\\
        0 & \text{otherwise}.
    \end{cases}
\end{equation}
equivalently define the dependency. The mixed-sign vectors will naturally appear as normals defining the walls of secondary cones, so we prefer the mixed-sign vector characterization.\footnote{Note that for $\lambda=\mu-\mu'$, the vectors $\lambda_+$, $\lambda_-$ may differ from $\mu$, $\mu'$. This is not problematic since the mixed-sign vector (mod scale) is what is used in our constructions.}

As can be seen in \cref{eq:vc_dependencies}, the space of dependencies of a vector configuration is simply $\mathrm{null}(\mathbf{A}_\mathrm{VC})$. To identify the space of dependencies of a point configuration $\mathbf{A}'_\mathrm{PC}$, observe that conditions \cref{eq:pc_dependencies_1,eq:pc_dependencies_2} on $\mathbf{A}'_\mathrm{PC}$ are equivalent to \cref{eq:vc_dependencies} on ${\mathbf{A}'}_\mathrm{VC}^\mathbf{1}$. I.e., if one lifts $\mathbf{A}'_\mathrm{PC}$ by $\omega=\mathbf{1}$ and interprets it as a vector configuration, then this lifted configuration ${\mathbf{A}'}_\mathrm{VC}^\mathbf{1}$ has the same dependencies as the original one $\mathbf{A}'_\mathrm{PC}$. Thus the space of dependencies of $\mathbf{A}'_\mathrm{PC}$ is simply $\mathrm{null}({\mathbf{A}'}_\mathrm{VC}^\mathbf{1})$. One calls ${\mathbf{A}'}_\mathrm{VC}^\mathbf{1}$ the \textit{homogenization}\footnote{Formally, $\mathbf{A}^\omega$ is a homogenization if $\mathbf{A}^\omega$ is contained in an affine hyperplane of $\mathbb{R}^{m+1}$.} of $\mathbf{A}'_\mathrm{PC}$. The point configuration $\mathbf{A}'_\mathrm{PC}$ and its homogenization are effectively identical: for example, any subdivision $\mathscr{S}(\mathbf{A'_\mathrm{PC})}$ defines simplices --- sets of labels --- which equivalently define a subdivision $\mathscr{S}({\mathbf{A}'}_\mathrm{VC}^\mathbf{1})$ and vice-versa. It suffices, then, to just study dependencies/circuits of vector configurations (i.e., linear relations). We will do so.

We call a vector configuration \textit{independent} or a \textit{simplex} if it is linearly independent, otherwise it is \textit{dependent}. If $\mathbf{A}$ (with label set $J$) is dependent but every proper subconfiguration $\mathbf{A}'\subset\mathbf{A}$ is independent, then $\mathbf{A}$ is called \textit{minimally dependent}. Observe the following facts about minimally dependent configurations:
\begin{enumerate}
    \item (\cite{De_Loera2010-ss}, Lemma 4.1.7) they satisfy $\mathrm{dim}(\mathrm{null}(\mathbf{A})) = 1$ and hence have an effectively unique (up to scaling) linear dependency $\mathbf{A}\lambda=0$ with $\lambda\neq\mathbf{0}$,
    \item $\lambda_j\neq0$ for all $j\in J$ (i.e., $J_0=\varnothing$) since, otherwise, $\mathbf{A}\setminus\mathbf{A}_j$ would be a dependent proper subconfiguration of $\mathbf{A}$, and
    \item (\cite{De_Loera2010-ss}, Lemma 2.4.2) any minimally dependent configuration $\mathbf{A}$ has $1$ or $2$ triangulations, defined as
    \begin{align}
        \label{eq:circuit_triang_plus}
        \mathscr{T}_+(\mathbf{A}) &= \setbuilder{\mathbf{A}\setminus j}{j\in J_+}, \\
        \label{eq:circuit_triang_minus}
        \mathscr{T}_-(\mathbf{A}) &= \setbuilder{\mathbf{A}\setminus j}{j\in J_-}.
    \end{align}
    The only case in which $\mathbf{A}$ has a single triangulation is when $J_+ = \varnothing$ or $J_-=\varnothing$, for which the definition of either $\mathscr{T}_+(\mathbf{A})$ or $\mathscr{T}_-(\mathbf{A})$ degenerates.
\end{enumerate}
In connection with oriented matroids, one calls the pair of label sets $(J_+,J_-)$ a \textit{(signed) circuit} with \textit{signature} $(|J_+|,|J_-|)$. Typically, when unambiguous, we will list the signature as $(n,m)$ for $n=\max(|J_-|,|J_+|)$ and $m=\min(|J_-|,|J_+|)$. Since $\mathbf{A}$ equivalently defines $(J_+,J_-)$, we will often abuse notation and call $\mathbf{A}$ the circuit. In \cite{De_Loera2010-ss}, the symbol $Z$ is also used to denote the circuit.

The discrete move $\mathscr{T}_+(\mathbf{A}) \leftrightarrow \mathscr{T}_-(\mathbf{A})$ is a/the \textit{flip} of the circuit $\mathbf{A}$. We will provide a general definition of ``flip'' later, but this first encounter shows some of the key properties. Namely:
\begin{enumerate}
    \item flips are encoded by circuits (see \cref{eq:circuit_triang_plus,eq:circuit_triang_minus}), thus relating precisely $2$ triangulations and
    \item flips can be thought of as coarsening a triangulation (to a `corank-1' subdivision) and then refining it in ``the other way'' (see \cref{fig:diagonal_flip}).
\end{enumerate}
In general, flips should be thought of as minimal transformations between triangulations.

Flips in minimally dependent configurations are nice: the two endpoints, $\mathscr{T}_+(\mathbf{A})$ and $\mathscr{T}_-(\mathbf{A})$, are both regular with secondary cones
\begin{align}
    \label{eq:secondary_cone_circuit_a}
    \mathbf{C}(\mathbf{A},\mathscr{T}_+) &= \setbuilder{\omega\in\mathbb{R}^{|\mathbf{A}|}}{\lambda\cdot \omega \geq 0}\\
    \label{eq:secondary_cone_circuit_b}
    \mathbf{C}(\mathbf{A},\mathscr{T}_-) &= \setbuilder{\omega\in\mathbb{R}^{|\mathbf{A}|}}{\lambda\cdot \omega \leq 0}.
\end{align}
The coarse intermediate subdivision is likewise regular, generated by any $\omega$ satisfying $\lambda\cdot\omega=0$, such as $\omega=\mathbf{0}$. Let $\omega_\pm$ be heights such that $\mathscr{T}_\pm = \mathscr{T}(\mathbf{A},\omega_\pm)$. For example, if $\omega_+$ is any point in $\relint{\mathbf{C}(\mathbf{A},\mathscr{T}_+)}$, then $\omega_-=-\omega_+$ generates\footnote{This has a cute interpretation, from \cite{De_Loera2010-ss}: lifting by heights $-\omega_+$ can be thought of as projecting down the \textit{upper} facets of $\mathbf{A}^\omega$. Thus, the flip in this case can be thought of as picking either the lower or upper facet of the lifted polyhedron. This geometric picture makes it clear why vector configurations can have non-flippable circuits while point configurations cannot: polytopes always have upper and lower facets but this is not true for cones.} $\mathscr{T}_-$. Consider lifting $\mathbf{A}$ by heights $(1-t)\omega_+ + t\omega_-$ as $t$ increases from $0$ to $1$. The resultant triangulation will transform, as $t$ increases, from $\mathscr{T}_+(\mathbf{A}) = \mathscr{T}(\mathbf{A}, \omega_+)$ to $\mathscr{T}_-(\mathbf{A}) = \mathscr{T}(\mathbf{A}, \omega_-)$. This height perspective demonstrates a $0<\tau<1$ such that $\lambda\cdot\omega(\tau)=0$, in which case $\omega(\tau)$ generates the coarsening. 
See, e.g., \cref{ex:simple_secondary_cone,fig:diagonal_flip}. This height perspective applies whenever both endpoints of a flip, as well as the intermediate coarsening, are all regular. This `height homotopy' \cite{De_Loera2010-ss} picture will be frequently used when connecting to toric geometry.

While useful, this height perspective is limited: irregular triangulations can be flipped in which case there is no height interpretation. Similarly, a flip between two regular triangulations may pass through an irregular coarsening (\cite{De_Loera2010-ss}, Ex. 5.3.4), in which case the height picture also fails. It is thus best to be cautious when adopting the height perspective of flips.

\subsubsection{Common Circuit Signatures}

\label{sec:common_sig}
What circuit signatures can arise for a $d$-dimensional vector configuration $\mathbf{A}_\mathrm{VC}$? First, observe that any collection of $d+1$ vectors is dependent, so signatures $(n,m)$ must satisfy $n+m\leq d+1$. We now split discussion by the smallest element in the signature (in our ordering, $m\leq n$):
\begin{enumerate}
    \item A circuit with $m=|J_-|=0$ cannot define a flip: it has one triangulation since the construction \cref{eq:circuit_triang_minus} degenerates. We never encounter such circuits for $n<2$ since that would require $\mathbf{A}_\mathrm{VC}$ to be semi-pathological (e.g., signature $(1,0)$ requires $\mathbf{0}\in\mathbf{A}_\mathrm{VC}$).
    \item For $m=|J_-|=1$, the flip $\mathscr{T}_+ \leftrightarrow \mathscr{T}_-$ consists of deleting/inserting the unique $j\in J_-$. I.e., $\mathscr{T}_+$ uses all rays $J_+\cup J_-$ while $\mathscr{T}_-$ only uses rays $J_+$.
    \item Circuits with $m>1$ do not change which rays are used in the triangulation.
\end{enumerate}

\subsubsection{Flippable Circuits, Star, and Link}
\label{subsubsec:flippable_secondary_cone}
We now know how to flip a circuit and we saw hints (\cref{eq:secondary_cone_circuit_a,eq:secondary_cone_circuit_b}) of how it relates to the secondary cone. Once we generalize flips to triangulations of non-minimally dependent configurations, we will effectively have the construction of the secondary cone in hand.

Recall that flips can be thought of as minimal transformations on triangulations. We begin with a simple case: let $\mathbf{A}$ be a vector configuration and let $\mathbf{A}'\subseteq\mathbf{A}$ be a circuit (i.e., a minimally dependent subconfiguration) with $\dim(\mathbf{A}') = \dim(\mathbf{A})$ and with $2$ triangulations. Then, if $\mathscr{T}_+(\mathbf{A}')\subseteq\mathscr{T}(\mathbf{A})$, one can transform
\begin{equation}
\label{eq:flip_solid}
    \mathscr{T}(\mathbf{A}) \to \left( \mathscr{T}(\mathbf{A})\setminus \mathscr{T}_+(\mathbf{A}')\right) \;\cup \; \mathscr{T}_-(\mathbf{A}').
\end{equation}
In this case, the circuit defined by $\mathbf{A}'$ is said to be \textit{flippable}. A complication arises, however, if $\dim(\mathbf{A}')<\dim(\mathbf{A})$. In this case, the maximal simplices of $\mathscr{T}_+(\mathbf{A}')$ will not be maximal in $\mathscr{T}(\mathbf{A})$ and \cref{eq:flip_solid} will no longer work: even if $\mathscr{T}_+(\mathbf{A}')\subseteq\mathscr{T}(\mathbf{A})$, the circuit $\mathbf{A}'$ might not be flippable --- see \cref{ex:not_embedded}.

\begin{example}
    Take the point configuration
    \begin{equation}
        \mathbf{A} = 
        \begin{bNiceMatrix}[first-row]
            1 & 2 & 3 & 4 & 5\\
            1 & 2 & 3 & 4 & 5\\
            1 & 2 & 1 & 2 & 1
        \end{bNiceMatrix}
        .
    \end{equation}
    Observe that $\mathbf{A}$ contains the circuit $\mathbf{A}'=\{\mathbf{A}_1,\mathbf{A}_3,\mathbf{A}_5\}$ with dependency 
    \begin{equation}
        \lambda'=\begin{bmatrix} 1 & -2 & 1\end{bmatrix}.
    \end{equation}
    This circuit has two triangulations $\mathscr{T}_+(\mathbf{A}')=(\{1,3\},\{3,5\})$ and $\mathscr{T}_-(\mathbf{A}')=(\{1,5\})$. A flip $\mathscr{T}_+(\mathbf{A}')\to\mathscr{T}_-(\mathbf{A}')$ corresponds to deleting the point $\mathbf{A}_3$.
    
    Now, consider $\mathscr{T}(\mathbf{A}) = (\{1,2,3\}, \{2,3,4\}, \{3,4,5\})$. Observe:
    \begin{enumerate}
        \item $\mathscr{T}_+(\mathbf{A}')\subset\mathscr{T}(\mathbf{A})$ and
        \item both $\mathscr{T}(\mathbf{A})$ and $\mathbf{A}'$ are symmetric under the map $g$ mapping both $1\leftrightarrow5$ and $2\leftrightarrow4$.
    \end{enumerate}
    While the first property suggests that one might be able to flip $\mathbf{A}'$ in $\mathscr{T}(\mathbf{A})$, the second property will indicate it is not. We will \textit{sketch} an argument here.

    Assume that the flip of the circuit, $\mathscr{T}_+(\mathbf{A}')\to\mathscr{T}_-(\mathbf{A}')$, corresponds to a flip in the entire triangulation $\mathscr{T}(\mathbf{A})$,
    \begin{equation}
        \mathscr{T}(\mathbf{A})\to\mathrm{flip}(\mathscr{T}(\mathbf{A})),
    \end{equation}
    where we let $\mathrm{flip}(\mathscr{T}(\mathbf{A}))$ be the post-flip triangulation. Since both $\mathscr{T}(\mathbf{A})$ and $\mathbf{A}'$ (the circuit) are symmetric under $g$, the post-flip triangulation $\mathrm{flip}(\mathscr{T})$ should also be symmetric. I.e., $\mathrm{flip}(\mathscr{T}(\mathbf{A})) = g(\mathrm{flip}(\mathscr{T}(\mathbf{A})))$. This is already a contradiction since $\mathbf{A}\setminus 3$ is quadrilateral, analogous to \cref{fig:diagonal_flip}, and hence it has two triangulations:
    \begin{align}
        \mathscr{T}_{25} &= (\{1,2,5\}, \{2,4,5\}) \text{ and}\\
        \mathscr{T}_{14} &= (\{1,2,4\}, \{1,4,5\}).
    \end{align}
    These triangulations only differ by whether the diagonal $\{1,4\}$ or $\{2,5\}$ is kept. Neither $\mathscr{T}_{25}$ nor $\mathscr{T}_{14}$ is symmetric under $g$, hence the contradiction.

    This failure should be understood as there being no unique way to delete $\mathbf{A}_3$ from $\mathscr{T}(\mathbf{A})$. Flips must be unique (i.e., relate exactly $2$ triangulations).
    \label{ex:not_embedded}
\end{example}

To check whether a lower-dimensional circuit is flippable (i.e., to generalize \cref{eq:flip_solid}), one needs a bit more technology. First, let $\sigma\in\mathscr{T}$. We stress that $\sigma$ need not be maximal. Define the \textit{star}\footnote{This has unfortunate notational overload with ``star triangulations'' as in \cref{eq:star_triang}.} of $\sigma$ as
\begin{equation}
    \text{st}_\mathscr{T}(\sigma) = \setbuilder{\sigma'\in\mathscr{\mathscr{T}}}{\sigma\cup\sigma'\in\mathscr{T}}
\end{equation}
In other words, $\text{st}_\mathscr{T}(\sigma)$ is the collection of all simplices $\sigma''\in\mathscr{T}$ (as well as their faces) for which $\sigma\subseteq\sigma''$. See \cref{fig:star_link}. Observe that one has to add the proper faces of each such $\sigma''$, even if they do not intersect $\sigma$, in order for $\text{st}_\mathscr{T}(\sigma)$ to be a well defined complex (see \cref{def:subdivision}). Collect these simplices not intersecting $\sigma$ as the \textit{link},
\begin{equation}
    \text{link}_\mathscr{T}(\sigma) = \setbuilder{\sigma'\in\text{st}_\mathscr{T}(\sigma)}{\sigma\cap\sigma'=\varnothing}.
\end{equation}
See \cref{fig:star_link} again.

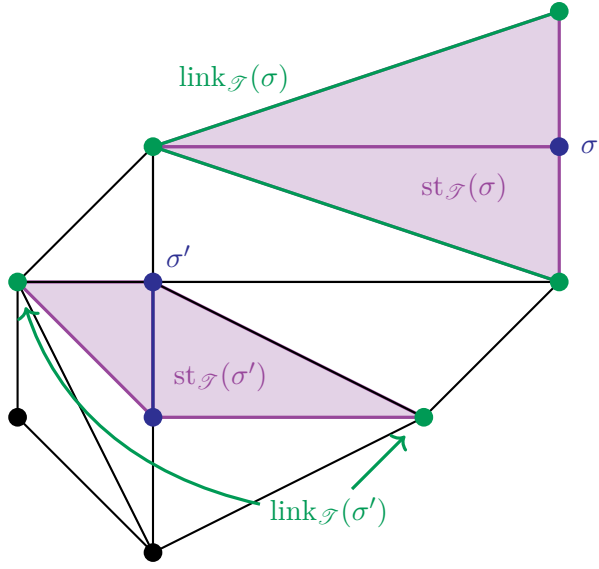
\begin{figure}
    \centering
    \begin{tikzpicture}[scale=1.8, every node/.style={circle, fill=black, inner sep=2.5pt}]
        \coordinate (v1) at (1,2);
        \coordinate (v2) at (1,3);
        \coordinate (v3) at (2,1);
        \coordinate (v4) at (2,2);
        \coordinate (v5) at (2,3);
        \coordinate (v6) at (2,4);
        \coordinate (v7) at (4,2);
        \coordinate (v8) at (5,3);
        \coordinate (v9) at (5,4);
        \coordinate (v10) at (5,5);

        \draw[thick] (v1) -- (v2) -- (v3) -- cycle;
        \draw (v2) -- (v3) -- (v4) -- cycle;
        \filldraw[very thick, fill=Purple!20, draw=Purple] (v2) -- (v4) -- (v5) -- cycle;
        \draw[thick] (v2) -- (v5) -- (v6) -- cycle;
        \draw[thick] (v3) -- (v4) -- (v7) -- cycle;
        \filldraw[very thick, fill=Purple!20, draw=Purple] (v4) -- (v5) -- (v7) -- cycle;
        \draw[thick] (v5) -- (v6) -- (v8) -- cycle;
        \draw[thick] (v5) -- (v7) -- (v8) -- cycle;
        \filldraw[very thick, fill=Purple!20, draw=Purple] (v6) -- (v8) -- (v9) -- cycle;
        \filldraw[very thick, fill=Purple!20, draw=Purple] (v6) -- (v9) -- (v10) -- cycle;

        \node[] at (v1) {};
        \node[fill=ForestGreen] at (v2) {};
        \node[] at (v3) {};
        \node[fill=Blue] at (v4) {};
        \node[Blue, label=above right:\textcolor{Blue}{$\sigma'$}] at (v5) {};
        \node[fill=ForestGreen] at (v6) {};
        \node[fill=ForestGreen] at (v7) {};
        \node[fill=ForestGreen] at (v8) {};
        \node[Blue, label=right:\textcolor{Blue}{$\sigma$}] at (v9) {};
        \node[fill=ForestGreen] at (v10) {};
        
        \draw[very thick, draw=Blue] (v4) -- (v5);

        \draw[very thick, draw=ForestGreen] (v6) -- (v8);
        \draw[very thick, draw=ForestGreen] (v6) -- (v10);

        \node[fill=none, Purple] at (4.3, 3.7) {$\text{st}_\mathscr{T}(\sigma)$};
        \node[fill=none, Purple] at (2.5, 2.3) {$\text{st}_\mathscr{T}(\sigma')$};
        
        \node[fill=none, ForestGreen] (txtA) at (3.3, 1.3) {$\text{link}_\mathscr{T}(\sigma')$};
        \draw[->, ForestGreen, very thick, shorten >=10pt, shorten <=0pt] (txtA) to[bend left=30] (v2);
        \draw[->, ForestGreen, very thick, shorten >=10pt, shorten <=-14pt] (txtA) -- (v7);

        \node[fill=none, ForestGreen] at (2.6, 4.5) {$\text{link}_\mathscr{T}(\sigma)$};
        
    \end{tikzpicture}
    \caption{A triangulation $\mathscr{T}$ of a point configuration. Two simplices $\sigma$ and $\sigma'$ of dimensions $0$ and $1$ respectively are plotted in blue. In green, the $\text{link}_\mathscr{T}(\sigma)$ and $\text{link}_\mathscr{T}(\sigma')$ are plotted. In purple, $\text{st}_\mathscr{T}(\sigma)$ and $\text{st}_\mathscr{T}(\sigma')$ are plotted. Note the link is a subcomplex of the star, so green colored regions also correspond to the star.}
    \label{fig:star_link}
\end{figure}

In this language, a circuit $\mathbf{A}'\subseteq\mathbf{A}$ can be flipped in $\mathscr{T}(\mathbf{A})$ if and only if (\cite{De_Loera2010-ss} Thm. 4.4.1)\footnote{We use the freedom in setting the sign of the dependency $\lambda$ to enforce that $\mathbf{A}'$ is originally triangulated as $\mathscr{T}_+(\mathbf{A})$.}
\begin{enumerate}
    \item $|J_-| > 0$,
    \item $\mathscr{T}_+(\mathbf{A}')\subseteq\mathscr{T}(\mathbf{A})$, and
    \item $L=\mathrm{link}_{\mathscr{T}(\mathbf{A})}(\sigma)$ is the same for all maximal $\sigma\in\mathscr{T}_+(\mathbf{A}')$.
\end{enumerate}
The first condition ensures that the circuit $\mathbf{A}'$ has $2$ triangulations while the latter two conditions ensure that $\mathscr{T}_+(\mathbf{A}')$ is properly embedded into $\mathscr{T}(\mathbf{A})$:
\begin{equation}
\label{eq:embedded}
    \setbuilder{\sigma\cup\sigma'}{\sigma\in\mathscr{T}_+(\mathbf{A}'), \;\sigma'\in L} \subseteq\mathscr{T}(\mathbf{A}).
\end{equation}
In lieu of a better name, we say that a circuit $\mathbf{A}'$ satisfying \cref{eq:embedded} is \textit{embedded} into $\mathscr{T}(\mathbf{A})$, regardless of the size $|J_-|$ (even if $|J_-|=0$). If $\mathbf{A}'$ is embedded with $|J_-|>0$, it is \textit{flippable}. With this, the appropriate generalization of \cref{eq:flip_solid} is
\begin{equation}
\label{eq:flip}
    \mathscr{T}(\mathbf{A}) \to \bigg( \mathscr{T}(\mathbf{A})\setminus \underbrace{ \setbuilder{\sigma\cup\sigma'}{\sigma\in\mathscr{T}_+(\mathbf{A}'),\sigma'\in L} }_{\text{embedding of }\mathscr{T}_+(\mathbf{A}') \text{ into } \mathscr{T}(\mathbf{A})}\bigg) \;\cup \; \underbrace{ \setbuilder{\sigma\cup\sigma'}{\sigma\in\mathscr{T}_-(\mathbf{A}'),\sigma'\in L} }_{\text{embedding of }\mathscr{T}_-(\mathbf{A}') \text{ into } \mathscr{T}(\mathbf{A})}.
\end{equation}
To connect with \cref{eq:flip_solid}, observe that \cref{eq:flip_solid} concerned flips of circuits with $\dim(\mathbf{A}') = \dim(\mathbf{A})$. In this case, $\text{st}_\mathscr{T}(\sigma) = \{\sigma\}$ and $\text{link}_\mathscr{T}(\sigma) = \varnothing$, showing that \cref{eq:flip} indeed reduces to \cref{eq:flip_solid}. We stress that \cref{eq:flip} only defines a flip if $\mathscr{T}_+(\mathbf{A}')$ has a constant link in $\mathscr{T}(\mathbf{A})$. This constant-link requirement is why the circuit in \cref{ex:not_embedded} failed to be flippable.

\subsubsection{General Construction of the Secondary Cone}
Finally, with an understanding of embedded/flippable circuits, we are able to provide an H-representation of the secondary cone $\mathbf{C}(\mathbf{A}_{\mathrm{VC}},\mathscr{T})$ of a regular triangulation $\mathscr{T}$. Recall from \cref{sec:dep_and_circ} the height perspective of flips: if the secondary cone of another triangulation $\mathscr{T}'$ shares a facet with $\mathbf{C}(\mathbf{A}_{\mathrm{VC}},\mathscr{T})$, then a height homotopy demonstrates a flip between $\mathscr{T}$ and $\mathscr{T}'$ through a regular subdivision. We later identified the intermediate subdivision with a circuit and, in \cref{ex:simple_secondary_cone}, saw that the dependency vector of the circuit defines the linear constraint on the heights to define the subdivision/facet. A similar story holds if a circuit is embedded but not flippable. This shows that $\mathbf{C}(\mathbf{A}_{\mathrm{VC}},\mathscr{T})$ is bounded by linear hyperplanes whose normals are given by the dependencies of the embedded circuits\footnote{Again, we set the orientation of the circuit to enforce that $\mathscr{T}_+(\mathbf{A}')\subseteq \mathscr{T}(\mathbf{A})$, not $\mathscr{T}_-(\mathbf{A}')\subseteq \mathscr{T}(\mathbf{A})$.},
\begin{equation}
\label{eq:secondary_cone_2}
    \mathbf{C}(\mathbf{A}_{\mathrm{VC}},\mathscr{T}) = \setbuilder{\omega}{\forall \text{ embedded circuits } \mathbf{A}' \text{ with dependency }\lambda; \;\lambda\cdot\omega\geq0}.
\end{equation}
Here, we implicitly let $\lambda_i=0$ for points $\mathbf{A}_i\notin \mathbf{A}'$. In other words, \cref{eq:secondary_cone_2} collects all embedded circuits, orients them such that $\mathscr{T}_+\subseteq\mathscr{T}$, and imposes the dependency $\lambda$ of each such circuit as an inwards-facing normal. A more common presentation is to collect these normals $\lambda$ as rows of a matrix $H$, so that $\mathbf{C}(\mathbf{A}_{\mathrm{VC}},\mathscr{T}) = \setbuilder{\omega}{H\omega\geq0}$. The definition given by \cref{eq:secondary_cone_2} is how one practically writes \cref{eq:secondary_cone_1}. This is the H-representation of the secondary cone.

Attempting to apply the construction \cref{eq:secondary_cone_2} to an irregular triangulation $\mathscr{T}$ results in either
\begin{enumerate}
    \item a non-solid cone for which height vectors $\omega$ in the relative interior define a non-triangulation subdivision which refines to $\mathscr{T}$ or
    \item a solid cone which contains no valid height vectors in its relative interior. 
\end{enumerate}
While the first case is easy to check using linear programming, the second case is a bit more annoying. Instead, if one wants to check the regularity of some triangulation $\mathscr{T}$, it is recommended to construct the cone of Theorem 2.3.20 in \cite{De_Loera2010-ss} and check if this is solid. This cone coincides with \cref{eq:secondary_cone_2} when $\mathscr{T}$ is regular but is always non-solid if $\mathscr{T}$ is irregular.

Now, consider the collection of secondary cones for a configuration $\mathbf{A}$. It is not hard to show that, if $\omega$ and $\omega'$ both define regular subdivisions, then so does both $\omega+\omega'$ and $s\omega$ for all $s\geq0$ (\cite{De_Loera2010-ss}, Thm. 5.2.11). Thus, the heights defining regular subdivisions form a convex cone. For point configurations, this is a trivial cone $\mathbb{R}^{|\mathbf{A}_\mathrm{PC}|}$; for vector configurations, this cone can be (but is not always) smaller. In either case, this cone has a natural decomposition: since $\omega\in\mathrm{relint}(\mathbf{C}(\mathbf{A},\mathscr{T}))$ uniquely defines a regular subdivision/triangulation of $\mathbf{A}$, the cone of valid heights is subdivided into a polyhedral fan of secondary cones. The external walls of this fan correspond to circuits for which $|J_-|=0$. This fan is called the \textit{secondary fan} $\Sigma$ of $\mathbf{A}$ \cite{Gelfand:1994} and will be discussed in more detail in \cref{sec:sec_fan}.

\subsection{Point vs. Vector Configuration}

\label{sec:pc_vs_vc}

As we saw when computing dependencies, every point configuration $\mathbf{A}_\mathrm{PC}$ can equivalently be thought of as a vector configuration $\mathbf{A}_\mathrm{VC}^\omega$ via homogenization. These objects are effectively the same. Previously we just used $\omega=\mathbf{1}$ but all that matters is that all points in $\mathbf{A}_\mathrm{VC}^\omega$ lay on an affine (not linear) hyperplane of $\mathbb{R}^{|\mathbf{A}_\mathrm{PC}|+1}$. In the other direction, any \textit{acyclic} vector configuration, one for which $\exists \psi\in\mathbb{R}^{\text{ambient dim}(\mathbf{A})}$ such that $\mathbf{A}_\mathrm{VC}^T \psi > 0$, can be equivalently thought of as a point configuration\footnote{See, for example, that any $\omega$ is a valid height vector in this case. This is because $\omega$ plus a sufficient scaling of $\mathbf{A}_\mathrm{VC}^T \psi>0$ will always be non-negative.}.

The novelty of vector configurations arises, in this work, for \textit{totally cyclic} vector configurations, those for which $\cone{\mathbf{A}} = \mathbb{R}^{\dim(\mathbf{A)}}$. Observe, if a vector configuration is totally cyclic, then it is not acyclic. A simple example of such novelty is that totally cyclic vector configurations may have non-flippable circuits, those for which either $J_+=\varnothing$ or $J_-=\varnothing$. Recall \cref{ex:bad_heights} which concerned a totally cyclic vector configuration, $\mathbf{A} = \begin{bmatrix} -1 & 1\end{bmatrix}$. Since $\lambda=(1,1)$ defines a dependency of this configuration, $J_- = \varnothing$ and hence the subdivision $\mathscr{S}(\mathbf{A})=(\{1,2\})$ only has one refinement, $\mathscr{T}(\mathbf{A})=(\{1\},\{2\})$. This circuit cannot be flipped --- it defines an impassible wall in the secondary fan. Such circuits are trivially impossible for point configurations: point configurations concern affine dependencies and hence $\sum_i\lambda_i=0$ which, if $\lambda\neq\mathbf{0}$, requires $\min(\lambda)<0<\max(\lambda)$.

Even for totally cyclic vector configurations, one can still sometimes draw a correspondence to a point configuration. For example, let $\mathscr{T}(\mathbf{A}_\mathrm{PC})$ be a star triangulation of a reflexive polytope $P=\conv{\mathbf{A}_\mathrm{PC}}$. Hence, $\mathscr{T}(\mathbf{A}_\mathrm{PC})$ can be mapped to a simplicial fan by extending the simplices $\conv{\sigma}$ to $\cone{\sigma}$. Thus, if we reinterpret $\mathbf{A}_\mathrm{PC}$ as a vector configuration\footnote{We stress, the reinterpretation of $\mathbf{A}_\mathrm{PC}$ as a vector configuration is \textit{not} the same as homogenization --- it is just treating every nonzero point like a vector.}, then the simplices of $\mathscr{T}(\mathbf{A}_\mathrm{PC})$ still define a triangulation, now of $\mathbf{A}_\mathrm{VC}$. Interesting behavior occurs, however, when $\mathscr{T}(\mathbf{A}_\mathrm{PC})$ is not star --- see \cref{ex:non_star}.

\begin{example}
\label{ex:non_star}
    Take the point configuration
    \begin{equation}
        \mathbf{A}_\mathrm{PC} = \begin{bNiceMatrix}[first-row]
        0 &  1 &  2 &  3 & 4 & 5\\
        0 & -1 & -1 &  1 & 1 & 0\\
        0 & -1 &  1 & -1 & 1 & 2
    \end{bNiceMatrix},
    \end{equation}
    with associated vector configuration (obtained by dropping the origin),
    \begin{equation}
    \label{eq:corresponding_vc}
        \mathbf{A}_\mathrm{VC} = \begin{bNiceMatrix}[first-row]
        1 &  2 &  3 & 4 & 5\\
        -1 & -1 &  1 & 1 & 0\\
        -1 &  1 & -1 & 1 & 2
    \end{bNiceMatrix}.
    \end{equation}
    both plotted in \cref{fig:starpc=vc}. Note: $\mathbf{A}_\mathrm{VC}$ is \textit{not} the homogenization of $\mathbf{A}_\mathrm{PC}$.
    
    There is a unique fine, star triangulation of $\mathbf{A}_\mathrm{PC}$,
    \begin{equation}
    \label{eq:starpc}
        \mathscr{T}(\mathbf{A}_\mathrm{PC}) = \begin{bmatrix}
            0 & 0 & 0 & 0 & 0\\
            1 & 1 & 2 & 3 & 4\\
            2 & 3 & 5 & 4 & 5
        \end{bmatrix},
    \end{equation}
    plotted in the top-left of \cref{fig:starpc=vc}. Since this triangulation is star, each simplex uniquely defines a simplicial cone with rays given by the non-origin points. This corresponds to the following triangulation of $\mathbf{A}_\mathrm{VC}$ (top-right of \cref{fig:starpc=vc}),
    \begin{equation}
    \label{eq:star_vc}
        \mathscr{T}(\mathbf{A}_\mathrm{VC}) = \begin{bmatrix}
            1 & 1 & 2 & 3 & 4\\
            2 & 3 & 5 & 4 & 5
        \end{bmatrix}.
    \end{equation}
    This correspondence can also be made through heights since $\mathscr{T}(\mathbf{A}_\mathrm{PC})$ is regular. Observe that every cone in $\supp{\mathbf{A}_\mathrm{VC}^\omega}$ trivially includes the origin, which can be thought of as having $0$-height. Hence, a non-negative height vector $\omega$ for $\mathscr{T}(\mathbf{A}_\mathrm{VC})$, such as $\omega=(1,1,1,1,1)$, can be interpreted as a height vector for $\mathscr{T}(\mathbf{A}_\mathrm{PC})$ upon setting $\omega_0=0$ (i.e., map $(1,1,1,1,1)\to(0,1,1,1,1,1)$). This map can be reversed: if $\omega'$ is any height generating $\mathscr{T}(\mathbf{A}_\mathrm{PC})$ with $\omega'_0\leq\omega'_j$ for all $j$ (this can always be satisfied for a regular, star $\mathscr{T}(\mathbf{A}_\mathrm{PC})$), then $\omega'_j\to\omega'_j-\omega'_0$ will be non-negative and generate $\mathscr{T}(\mathbf{A}_\mathrm{VC})$.

    Now, observe what changes if $\mathscr{T}(\mathbf{A}_\mathrm{PC})$ is non-star, such as
    \begin{equation}
    \label{eq:nonstarpc}
        \mathscr{T}(\mathbf{A}_\mathrm{PC}) = \begin{bmatrix}
            0 & 0 & 0 & 0 & 2\\
            1 & 1 & 2 & 3 & 4\\
            2 & 3 & 4 & 4 & 5
        \end{bmatrix}.
    \end{equation}
    This triangulation contains a non-star simplex $\{2,4,5\}$ which cannot obviously be encoded into any triangulation of $\mathbf{A}_\mathrm{VC}$. Only the subcomplex of simplices containing the origin, $\text{st}_{\mathscr{T}(\mathbf{A}_\mathrm{PC})}(0)$, has a natural correspondence to $\mathbf{A}_\mathrm{VC}$. Nonetheless, we can utilize the height map discussed in the last paragraph. For $\epsilon>0$, the heights $\omega=(0,1,1,1,1,2+\epsilon)$ generate $\mathscr{T}(\mathbf{A}_\mathrm{PC})$. Utilizing the aforementioned map $\omega_j\to\omega_j-\omega_0$, this corresponds to heights $\omega=(1,1,1,1,2+\epsilon)$ for the vector configuration. Lifting $\mathbf{A}_\mathrm{VC}$ by these heights generates the non-fine triangulation
    \begin{equation}
    \label{eq:multi_vc}
        \mathscr{T}(\mathbf{A}_\mathrm{VC}) = \begin{bmatrix}
            1 & 1 & 2 & 3\\
            2 & 3 & 4 & 4
        \end{bmatrix}
    \end{equation}
    which corresponds to extending all star simplices of \cref{eq:nonstarpc} to simplicial cones, dropping all non-star simplices. We have deleted vector $5$ from $\mathbf{A}_\mathrm{VC}$.
\end{example}

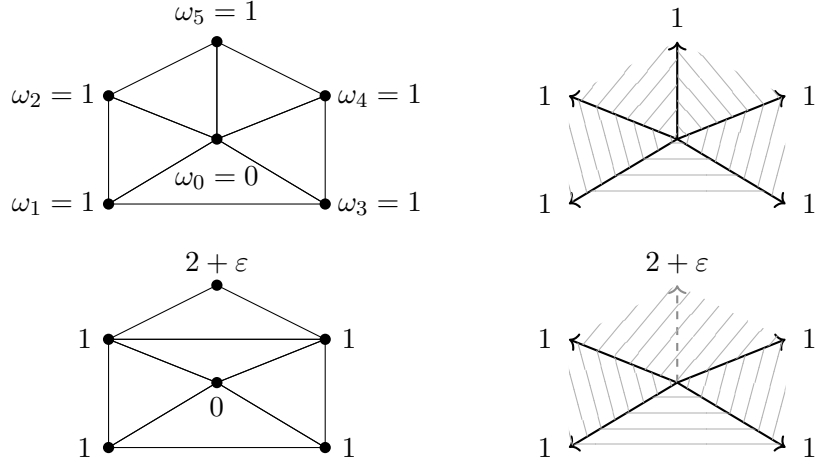
\begin{figure}
    \centering
    \begin{tikzpicture}[scale=1.8, every node/.style={circle, fill=black, inner sep=1.5pt}]
    
        \begin{scope}[shift={(0,0)}, scale=0.4]
            \coordinate (v1) at (1,1);
            \coordinate (v2) at (1,3);
            \coordinate (v3) at (3,2.2);
            \coordinate (v4) at (3,4);
            \coordinate (v5) at (5,1);
            \coordinate (v6) at (5,3);
    
            \draw (v1) -- (v2) -- (v3) -- cycle;
            \draw (v1) -- (v3) -- (v5) -- cycle;
            \draw (v2) -- (v3) -- (v4) -- cycle;
            \draw (v3) -- (v4) -- (v6) -- cycle;
            \draw (v3) -- (v5) -- (v6) -- cycle;
    
            \node[label=left:{$\omega_1=1$}] at (v1) {};
            \node[label=left:{$\omega_2=1$}] at (v2) {};
            \node[label={[yshift=6pt]below:{$\omega_0=0$}}] at (v3) {};
            \node[label={[yshift=-10pt]above:{$\omega_5=1$}}] at (v4) {};
            \node[label=right:{$\omega_3=1$}] at (v5) {};
            \node[label=right:{$\omega_4=1$}] at (v6) {};
        \end{scope}
            
        \begin{scope}[shift={(2*\dx,0)}, scale=0.4]
            \coordinate (v1) at (1,1);
            \coordinate (v2) at (1,3);
            \coordinate (v3) at (3,2.2);
            \coordinate (v4) at (3,4);
            \coordinate (v5) at (5,1);
            \coordinate (v6) at (5,3);
    
            \draw[thick, ->] (v3) -- (v1);
            \draw[thick, ->] (v3) -- (v2);
            \draw[thick, ->] (v3) -- (v4);
            \draw[thick, ->] (v3) -- (v5);
            \draw[thick, ->] (v3) -- (v6);

            \fill[pattern={Lines[angle=+105, distance=6pt]}, pattern color=gray!60] (v1) -- (v2) -- (v3) -- cycle;
            \fill[pattern={Lines[angle=+50, distance=6pt]}, pattern color=gray!60] (v2) -- (v3) -- (v4) -- cycle;
            \fill[pattern={Lines[angle=-50, distance=6pt]}, pattern color=gray!60] (v3) -- (v4) -- (v6) -- cycle;
            \fill[pattern={Lines[angle=-105, distance=6pt]}, pattern color=gray!60] (v3) -- (v5) -- (v6) -- cycle;
            \fill[pattern={Lines[angle=0, distance=6pt]}, pattern color=gray!60] (v1) -- (v3) -- (v5) -- cycle;
    
            \node[fill=none, label=left:1] at (v1) {};
            \node[fill=none, label=left:1] at (v2) {};
            \node[fill=none] at (v3) {};
            \node[fill=none, label=above:1] at (v4) {};
            \node[fill=none, label=right:1] at (v5) {};
            \node[fill=none, label=right:1] at (v6) {};
        \end{scope}
            
        \begin{scope}[shift={(0,-1.2*\dy)}, scale=0.4]
            \coordinate (v1) at (1,1);
            \coordinate (v2) at (1,3);
            \coordinate (v3) at (3,2.2);
            \coordinate (v4) at (3,4);
            \coordinate (v5) at (5,1);
            \coordinate (v6) at (5,3);
    
            \draw (v1) -- (v2) -- (v3) -- cycle;
            \draw (v1) -- (v3) -- (v5) -- cycle;
            \draw (v2) -- (v3) -- (v6) -- cycle;
            \draw (v2) -- (v4) -- (v6) -- cycle;
            \draw (v3) -- (v5) -- (v6) -- cycle;
    
            \node[label=left:1] at (v1) {};
            \node[label=left:1] at (v2) {};
            \node[label=below:0] at (v3) {};
            \node[label={[yshift=-3mm]above:$2+\epsilon$}] at (v4) {};
            \node[label=right:1] at (v5) {};
            \node[label=right:1] at (v6) {};
        \end{scope}
            
        \begin{scope}[shift={(2*\dx,-1.2*\dy)}, scale=0.4]
            \coordinate (v1) at (1,1);
            \coordinate (v2) at (1,3);
            \coordinate (v3) at (3,2.2);
            \coordinate (v4) at (3,4);
            \coordinate (v5) at (5,1);
            \coordinate (v6) at (5,3);
    
            \draw[thick, ->] (v3) -- (v1);
            \draw[thick, ->] (v3) -- (v2);
            \draw[thick, dashed, gray!90, ->] (v3) -- (v4);
            \draw[thick, ->] (v3) -- (v5);
            \draw[thick, ->] (v3) -- (v6);

            \fill[pattern={Lines[angle=+105, distance=6pt]}, pattern color=gray!60] (v1) -- (v2) -- (v3) -- cycle;
            \fill[pattern={Lines[angle=+50, distance=6pt]}, pattern color=gray!60] (v2) -- (v3) -- (v6) -- (v4) -- cycle;
            \fill[pattern={Lines[angle=-105, distance=6pt]}, pattern color=gray!60] (v3) -- (v5) -- (v6) -- cycle;
            \fill[pattern={Lines[angle=0, distance=6pt]}, pattern color=gray!60] (v1) -- (v3) -- (v5) -- cycle;
    
            \node[fill=none, label=left:1] at (v1) {};
            \node[fill=none, label=left:1] at (v2) {};
            \node[fill=none] at (v3) {};
            \node[fill=none, label={[yshift=-3mm]above:$2+\epsilon$}] at (v4) {};
            \node[fill=none, label=right:1] at (v5) {};
            \node[fill=none, label=right:1] at (v6) {};
        \end{scope}
    \end{tikzpicture}
    \caption{Left: regular triangulations $\mathscr{T}$ of a point configuration $\mathbf{A}$, labeled with heights $\omega_j$ which generate $\mathscr{T}$. Right: the associated vector configuration $\mathbf{A}\setminus0$ (\textit{not} the homogenization), lifted with corresponding heights.}
    \label{fig:starpc=vc}
\end{figure}

While \cref{ex:non_star} maps a non-star triangulation of a point configuration to a non-fine triangulation of a vector configuration, there are also cases where one can map a non-star triangulation of a point configuration to a fine triangulation of a vector configuration. These will be the ``vex triangulations'' that we study in \cref{sec:ref_poly}.

The lessons from \cref{ex:non_star} are that one \textit{can} map back and forth between triangulations of point configurations $\mathbf{A}_\mathrm{PC}$ (which contain the origin) and their associated vector configurations $\mathbf{A}_\mathrm{VC} = \mathbf{A}_\mathrm{PC} \setminus 0$, preserving the star simplices. We formalize these lessons below.

\begin{prop}
\label{thm:pc_to_vc}
    For any $\mathscr{T}(\mathbf{A}_\mathrm{PC},\omega)$ with $\omega_0\leq\omega_j$, lifting the analogous configuration $\mathbf{A}_\mathrm{VC}=\mathbf{A}_\mathrm{PC}\setminus0$ by $\omega_j-\omega_0$ generates
    \begin{equation}
        \mathscr{T}(\mathbf{A}_\mathrm{VC}, \omega_j-\omega_0) = \setbuilder{\sigma}{\sigma\cup\{0\}\in\mathscr{T}(\mathbf{A}_\mathrm{PC},\omega)}.
    \end{equation}
\end{prop}

\begin{prop}
\label{thm:vc_to_pc}
    For any $\mathscr{T}(\mathbf{A}_\mathrm{VC},\omega)$ with $\omega\geq0$, lifting the analogous configuration $\mathbf{A}\cup\{\mathbf{0}\}$ by $(0,\omega)$ generates a regular subdivision $\mathscr{S}(\mathbf{A}_\mathrm{PC})$ satisfying
    \begin{equation}
        \setbuilder{\sigma\cup\{0\}}{\sigma\in\mathscr{T}(\mathbf{A}_\mathrm{VC},\omega)} \subseteq \mathscr{S}(\mathbf{A}_\mathrm{PC},(0,\omega)).
    \end{equation}
    We stress: $\mathscr{S}(\mathbf{A}_\mathrm{PC})$ can be non-star.
\end{prop}

These two propositions have analogous proofs, so we only show \cref{thm:pc_to_vc} in detail. The result of these theorems is that one can map between regular triangulations of the associated point and vector configurations while preserving the star of $0$ in $\mathscr{T}(\mathbf{A}_\mathrm{PC})$. We will have more to say about these maps in \cref{subsubsec:upper}.

\begin{proof}[Proof of \cref{thm:pc_to_vc}]
    Let $\sigma\cup\{0\}\in\mathscr{T}(\mathbf{A}_\mathrm{PC},\omega)$. Call this a ``star simplex'' of $\mathscr{T}(\mathbf{A}_\mathrm{PC},\omega)$. One needs to show that $\sigma$ corresponds to a lower facet of $\mathbf{A}_\mathrm{VC}^{\omega_j-\omega_0}$.

    First we show that, upon lifting, $\ConeOp(\sigma)$ is \textit{contained in} a lower facet of $\mathbf{A}_\mathrm{VC}^{\omega_j-\omega_0}$. Suppose, for the sake of contradiction, that $\ConeOp(\sigma)$ was not contained in a lower facet of $\mathbf{A}_\mathrm{VC}^{\omega_j-\omega_0}$. Then, there exists a $p\in\ConeOp(\sigma)$ that, when lifted to $p^{\omega_j-\omega_0}$, is also not contained in a lower facet. Since $\omega_j-\omega_0$ is a valid, non-negative height vector, there is a minimal $h$ such that $(p,h)\in\ConeOp(\mathbf{A}_\mathrm{VC}^{\omega_j-\omega_0})$. Call such a point $p^*$. Clearly, then, $p^*$ is contained in a lower facet, $f$, of $\mathbf{A}_\mathrm{VC}^{\omega_j-\omega_0}$. Let $\sigma'$ be the subconfiguration corresponding to the rays of $\mathbf{A}_\mathrm{VC}^{\omega_j-\omega_0}$ defining $f$. By scaling $p$ such that $p\in\ConvOp(\sigma\cup\{0\})\cap\ConvOp(\sigma'\cup\{0\})$, one observes that $\sigma\cup\{0\}$ likewise is not a lower facet of $\mathbf{A}_\mathrm{PC}^\omega$, a contradiction.

    We now must show that $j\notin\sigma$ implies that $\mathbf{A}_j$ is lifted above the linear hyperplane spanned by $\sigma$. This will show that $\sigma$ is not just contained in a lower facet of $\mathbf{A}_\mathrm{VC}^{\omega_j-\omega_0}$, but that it \textit{is} a lower facet. Assume, for the sake of contradiction, that $\omega_j$ lifts $\mathbf{A}_j$ onto the hyperplane spanned by $\sigma$. In terms of $\mathscr{T}(\mathbf{A}_\mathrm{PC},\omega)$, this shows that $\sigma\cup\{0\}$ is strictly contained in a non-simplicial cell of $\mathscr{T}(\mathbf{A}_\mathrm{PC},\omega)$ which also contains $j$ (and maybe other points). This is a contradiction since $\sigma\cup\{0\}$ is a maximal cell of $\mathscr{T}(\mathbf{A}_\mathrm{PC},\omega)$.
\end{proof}

\Cref{ex:non_star} also touches upon an interesting and important note: every subdivision of a totally cyclic vector configuration, regardless of fineness, has the same support. Hence \cref{eq:multi_vc} equivalently defines a triangulation of $\mathbf{A}_\mathrm{VC}$ or of $\mathbf{A}_\mathrm{VC}\setminus5$. These vector configurations have different corresponding point configurations, either $\mathbf{A}_\mathrm{PC} \cup\{\mathbf{0\}}$ or $(\mathbf{A}_\mathrm{PC}\setminus5)\cup\{\mathbf{0\}}$, respectively. Only for the latter configuration does \cref{eq:multi_vc} correspond to an FRST. In this way, the most natural point configuration corresponding to some triangulation $\mathscr{T}(\mathbf{A}_\mathrm{VC})$ is that of the vectors \textit{used} in $\mathscr{T}(\mathbf{A}_\mathrm{VC})$ (and the origin).

Thus, by performing a deletion/insertion flip in $\mathbf{A}_\mathrm{VC}$ (specifically, of what would be a vertex of $\conv{\mathbf{A}_\mathrm{VC}}$), one can transform between FRSTs of different polytopes. In fact, \textit{any} collection FRSTs can trivially be connected in this way (just consider the union of their vector configurations). Such a transformation is particularly interesting when it occurs via a single flip between two reflexive $4$D polytopes (this occurs). As we will elaborate on in \cref{sec:sec_fan}, flips induce birational maps of toric varieties, which in turn can give rise to such maps between their anticanonical hypersurfaces. In this way, one may be able to show that these insertion/deletion flips constitute the building blocks of Reid's fantasy --- the conjecture that all Calabi--Yau threefolds are birational \cite{reid1987moduli}, or are continuously connected in the moduli space --- within the Kreuzer--Skarke database.

This concludes the review. All-in-all, we have the dictionary given by \cref{tab:pc_vc_dict} between point configurations and vector configurations.
\begin{table}[]
    \centering
    {
    \renewcommand{\arraystretch}{1.2}
    \begin{tabular}{|>{\centering\arraybackslash}m{4.25cm}|
                >{\centering\arraybackslash}m{4.25cm}|
                >{\centering\arraybackslash}m{4.25cm}|}
        \hline
        & Point Configuration & Vector Configuration \\
        \hline 
        \hline
        elements & points, $(\mathbf{A}_\mathrm{PC})_j$ & vectors, $(\mathbf{A}_\mathrm{VC})_j$\\
        support & convex hull, $\conv{\mathbf{A}_\mathrm{PC}}$ & positive hull, $\cone{\mathbf{A}_\mathrm{VC}}$\\
        subdivison & polytopal complex & polyhedral fan\\
        triangulation & simplicial complex & simplicial fan\\
        what defines a fan & star subdivision & any subdivision\\
        height-space & $\supp{\Sigma} = \mathbb{R}^{|\mathbf{A}|}$ & $\mathbb{R}_{\geq0}^{|\mathbf{A}|} \subseteq \supp{\Sigma} \subseteq\mathbb{R}^{|\mathbf{A}|}$\\
        correspondence of configs (homogenization) & \multicolumn{2}{c|}{$\mathbf{A}_\mathrm{PC} \simeq \mathbf{A}^\mathbf{1}_\mathrm{VC}$} \\
        correspondence of triangulations & \multicolumn{2}{c|}{$\mathrm{st}_{\mathscr{T}(\mathbf{A}_\mathrm{PC})}(0) \simeq \mathscr{T}(\mathbf{A}_\mathrm{VC})$} \\
        \hline
    \end{tabular}
    }
    \caption{Summary of the structure of point and vector configurations.}
    \label{tab:pc_vc_dict}
\end{table}

\section{Review of Toric Geometry}

\label{sec:toric}

In this section, we review toric geometry and set notation, following \cite{cls}. Complementary introductions from various perspectives can be found, e.g., in \cite{fulton1993introduction, Hori:2003ic, Closset:2009sv}. All results stated here are standard. We assume some familiarity with algebraic geometry.

\subsection{Toric Varieties}

A \textit{toric variety} is an irreducible variety $V$ containing a dense algebraic torus $T \cong (\mathbb{C}^*)^n$ such that the action of $T$ on itself extends algebraically to $V$.
To a torus $(\mathbb{C}^*)^n$ we can associate characters and one-parameter subgroups: group homomorphisms to and from $\mathbb{C}^*$, respectively. These are given by dual lattices of Laurent monomials. Explicitly, characters take the form
\begin{equation}
    (t_1, \dots, t_n) \mapsto \prod_{i=1}^n t_i^{a_i}
\end{equation}
for some vector $a = (a_1, \dots, a_n) \in \mathbb{Z}^n$. Likewise, one-parameter subgroups take the form $s \mapsto (s^{b_1}, \dots, s^{b_n})$ for some $b = (b_1, \dots, b_n) \in \mathbb{Z}^n$. In particular, then, the torus $T \subset V$ has a \textit{character lattice} $M \cong \mathbb{Z}^n$ and a \textit{one-parameter subgroup lattice} $N \cong \mathbb{Z}^n$. It is useful to note that $T \cong \mathrm{Hom}(M, \mathbb{C}^*)$ by $t \mapsto (\chi \mapsto \chi(t))$ for $\chi \in M$.

We will only consider toric varieties which arise from \textit{fans} $\Sigma$ contained in $N_\mathbb{R} := N \otimes_{\mathbb{Z}} \mathbb{R}$ (i.e., separated normal toric varieties \cite{normalHasFan}), by which we mean a finite collection of rational polyhedral cones in $N_\mathbb{R}$ closed under taking faces and such that pairwise intersections are faces of each cone.\footnote{We note that if the cones are not pointed, there are some subtleties. In \cite{cls} such fans are called generalized fans; these arise at the boundary of the toric effective cone/secondary fan, and we will return to them in \cref{sec:translate}, but they will not play a major role.} Note that subdivisions of vector configurations with pointed cones induce fans (by taking supports), but these objects are formally different because fans do not remember vectors and cells: they only consist of the geometric regions (i.e., the cones). The building blocks of toric varieties are \textit{affine toric varieties} which are in correspondence to cones in $N_\mathbb{R}$: fans encode general toric varieties by prescribing a set of affine toric varieties (given by the individual cones) and how they are glued together (namely, along the affine toric varieties associated to intersections of cones). However, we will not pursue this construction: rather, we will construct toric varieties as quotients.\footnote{We stress that toric varieties only depend on the cones in a fan, as compared to triangulations of vector configurations, which depend on cells (subsets of vectors) rather than merely supports of cells (conical hulls of cells, or cones). This is elaborated upon in \cref{sec:translate}.}

Fix an $n$-dimensional lattice $N$ and a fan $\Sigma \subset N_\mathbb{R}$. Let $\Sigma(n)$ denote the set of $n$-dimensional cones in $\Sigma$, and let $u_\rho$ denote the \textit{minimal generator} of the ray $\rho \in \Sigma(1)$, defined as the generator of $N \cap \rho$. The associated toric variety $V_\Sigma$ with dense torus $T$ can be constructed as a quotient, written as follows.
\begin{equation}
    \label{eq:quot}
    V_\Sigma = \frac{\mathbb{C}^{|\Sigma(1)|} \setminus Z(\Sigma)}{G}.
\end{equation}
Here, $Z(\Sigma) \subset \mathbb{C}^{|\Sigma(1)|}$ is the \textit{exceptional set}, which must be removed for a well-behaved (i.e., separated/Hausdorff) quotient, and $G$ is a group acting on $\mathbb{C}^{|\Sigma(1)|}$. This is a natural generalization of projective space, where one removes the origin and quotients by the torus $\mathbb{C}^*$ acting by overall scaling. The group $G$ and other coarse topological invariants of $V_\Sigma$ are determined by the rays $\Sigma(1)$ alone. For example, $V_\Sigma$ is compact if $\Sigma$ is complete (i.e., $\cup_{\sigma \in \Sigma} \sigma = N_\mathbb{R}$). In contrast, the exceptional set $Z(\Sigma)$ and the finer topological data of $V_\Sigma$ --- especially its intersection theory --- are determined by the higher-dimensional cones.

Let us now define $G$ in \cref{eq:quot}, though it will take a little work. Following \S5.1 in \cite{cls}, $\mathbb{C}^{|\Sigma(1)|}$ is acted upon naturally by the torus $(\mathbb{C}^*)^{|\Sigma(1)|}$: in particular, we have the following short exact sequence.
\begin{equation}
    0 \to G \stackrel{f}{\to} (\mathbb{C}^*)^{|\Sigma(1)|} \stackrel{g}{\to} T \to 0.
\end{equation}
This is the intuitive result that the torus $T$ of $V_\Sigma$ consists of those elements of $(\mathbb{C}^*)^{|\Sigma(1)|}$ that were not rendered trivial upon taking the quotient by $G$. We have not yet understood any of these maps, though: in this direction, we can re-write the sequence as
\begin{equation}
    0 \to \mathrm{Hom}(\mathrm{Cl}(V_\Sigma), \mathbb{C}^*) \stackrel{f}{\to} \mathrm{Hom}(\mathbb{Z}^{|\Sigma(1)|}, \mathbb{C}^*) 
    \stackrel{g}{\to} \mathrm{Hom}(M, \mathbb{C}^*) \to 0.
\end{equation}
Here, $\mathrm{Hom}$ denotes the group of homomorphisms between a pair of groups. At this point, $\mathrm{Cl}(V_\Sigma)$ is \textit{defined} by this sequence (i.e., as the group whose characters are given by $\ker g$), but it will end up being the divisor class group of $V_\Sigma$, hence the notation. This sequence results from applying $\mathrm{Hom}(-, \mathbb{C}^*)$ to the following sequence (\cite{cls}, Th. 4.1.3).
\begin{equation}
    \label{eq:SES}
    0 \to M \stackrel{\alpha}{\longrightarrow} \mathbb{Z}^{|\Sigma(1)|} \stackrel{\beta}{\to} \mathrm{Cl}(V_\Sigma) \to 0.
\end{equation}
This is one of the more important sequences in toric geometry. We will interpret it momentarily: for now, though, let us focus on defining $G$. The linear map $\alpha$ is determined by $|\Sigma(1)|$ vectors in $M^* = N$. We already have a natural set of such vectors, namely the minimal generators $u_\rho$, and these are indeed the correct choice: $\alpha$ acts as
\begin{equation}
    m \mapsto (\langle m, u_\rho \rangle)_{\rho \in \Sigma(1)}.
\end{equation}
From the sequence, we have that $\mathrm{Cl}(V_\Sigma) = \mathrm{coker}(\alpha) \cong \mathbb{Z}^{|\Sigma(1)|} / M$ and $G$ is formally defined as $\mathrm{Hom}(\mathrm{Cl}(V_\Sigma), \mathbb{C}^*)$, though this is rather abstract. 

It's worthwhile to unpack the definition of $G$, so let's take some time to do this. Let us understand how $\alpha$ and $\beta$ induce $f$ and $g$ in the original sequence. First, exploiting $T \cong \mathrm{Hom}(M, \mathbb{C}^*)$, we can write the map $g$ as 
\begin{equation}
    (t_\rho)_{\rho \in \Sigma(1)} \mapsto \Big( m \mapsto \prod_\rho t_\rho^{\langle m, u_\rho \rangle} \Big).
\end{equation}
This arises just from unraveling definitions, though it is notationally involved. We can then readily express $G$ as the kernel of $g$: that is, points in $\mathbb{C}^{|\Sigma(1)|}$ related by elements in $\ker g$ are identified in $V_\Sigma$. 
\begin{equation}
    \begin{aligned}
        G 
        &= \ker g \\
        &= \Big\{ (t_\rho)_{\rho \in \Sigma(1)} \in (\mathbb{C}^*)^{|\Sigma(1)|} \; \Big| \;  \prod_\rho t_\rho^{\langle m, u_\rho \rangle} = 1, \forall m \in M\Big \} \\
        &= \Big\langle (\lambda^{q_\rho})_{\rho \in \Sigma(1)} \in (\mathbb{C}^*)^{|\Sigma(1)|} \; \Big| \; \lambda \in \mathbb{C}^*, q_\rho \in \mathbb{Z}^{|\Sigma(1)|}, \sum_\rho q_\rho u_\rho = 0 \Big\rangle.
    \end{aligned}
    \label{eq:G}
\end{equation}
The last equality is a convenient expression of the generators of $G$ in terms of linear relations of the $u_\rho$ (direct substitution readily confirms that these elements belong to $\ker g$). In particular, it suffices to let the $q_\rho$ be a \textit{basis} of such linear relations. These impose \textit{scaling relations} on $\mathbb{C}^{|\Sigma(1)|}$, identifying $(x_\rho)_{\rho \in \Sigma(1)} \sim (\lambda^{q_\rho} x_\rho)_{\rho \in \Sigma(1)}$, generalizing the familiar single scaling relation of $\mathbb{P}^n$ (for which each $q_\rho = 1$). 

We can equivalently recover this final expression for $G$ by understanding $\beta$. Let us (non-canonically) choose group generators for the finitely generated Abelian group $\mathrm{Cl}(V_\Sigma)$, decomposing it as $\mathrm{Cl}(V_\Sigma) \cong \mathbb{Z}^k \times H$ for $H$ finite. Letting $N' = \mathbb{Z}\{u_\rho \; | \; \rho \in \Sigma(1)\}$, one can show that $H = N / N'$: we will assume $H = \{0\}$, as the free part of $\mathrm{Cl}(V_\Sigma)$ is most important for our purposes.\footnote{The case of $H \neq \{0\}$ is not much more difficult, just mildly cumbersome. All Weil divisors have a torsional component valued in $H$, the group $G$ also has the finite factor $H$, and as a consequence $V_\Sigma$ is non-simply connected with $\pi^1(V_\Sigma) = H$ (\cite{cls}, Th. 12.1.10). We are primarily interested in the class group because $\mathrm{Cl}(V_\Sigma)_\mathbb{R}$ houses the secondary fan: for this purpose, torsion is clearly irrelevant. Additionally, we comment that of the $\sim$400 million 4D reflexive polytopes, only 16 give rise to this finite factor \cite{Klemm:2004km, Batyrev:2005jc}.} Then the rows of the matrix representation $Q$ of $\beta : \mathbb{Z}^{|\Sigma(1)|} \to \mathbb{Z}^k$ furnish a basis of linear relations of the $u_\rho \in N$. Then $G \cong \mathrm{Hom}(\mathbb{Z}^k, \mathbb{C}^*) \cong (\mathbb{C}^*)^k$ and the induced map $f$ maps generators $\lambda e_i \in (\mathbb{C}^*)^k$ (for $e_i$ the $i$th standard basis vector) to $(\lambda^{Q_{i\rho}})_{\rho \in \Sigma(1)} \in (\mathbb{C}^*)^{|\Sigma(1)|}$, which are exactly the generators exhibited in \cref{eq:G}.

The upshot, then, is that both the free part of the group $G$ and the map $\beta$ are described by a basis of linear relations of the minimal generators $u_\rho$. It's worth seeing this in practice in \cref{ex:hirz_1}.

Now, with $G$ in hand, we can now turn our attention to the exceptional set $Z(\Sigma)$. To achieve this, we will define the \textit{irrelevant ideal} $B$, a square-free monomial ideal in the coordinate ring $S = \mathbb{C}[x_\rho \; | \; \rho \in \Sigma(1)]$ of the ``upstairs'' $\mathbb{C}^{|\Sigma(1)|}$. To a subset $\tau \subset \Sigma(1)$ we can associate a monomial $x^\tau = \prod_{\rho \in \tau} x_\rho \in S$, and we also define $\hat{\tau} := \Sigma(1) \setminus \tau$. Then we can write
\begin{equation}
    B = \left\langle x^{\hat{\tau}} \; | \; \tau \in \Sigma \right\rangle.
\end{equation}
In particular, it suffices to take $\tau \in \Sigma(n)$. The exceptional set is then 
\begin{equation}
    Z(\Sigma) = \mathbb{V}(B).
\end{equation} 
We can interpret this as follows. The affine open charts of $V_\Sigma$ are given by the affine toric varieties associated to the maximal cones $\sigma \in \Sigma(n)$, on which we impose $x_\rho \neq 0$ for $\rho \notin \sigma$. For this to indeed be an open cover, we cannot have points which do not belong to any chart. One can verify that $Z(\Sigma)$ is exactly the set of points in $\mathbb{C}^{|\Sigma(1)|}$ which do not fall in a chart.

While this is succinct, in the case that $\Sigma$ is simplicial we have an alternative definition for $Z(\Sigma)$ that elucidates the intersection theory of $V_\Sigma$. In such a case, we can define another square-free monomial ideal, the \textit{Stanley-Reisner ideal} $J \subset S$, given by
\begin{equation}
    J = \left\langle x^\tau \; \Big| \; \tau \subset \Sigma(1), \tau \notin \Sigma \right\rangle.
\end{equation}
in which case 
\begin{equation}
    \label{eq:primary_decomp}
    Z(\Sigma) = \bigcup_{\substack{\tau \subset \Sigma(1) \\ x^\tau \in J}} \mathbb{V}( \langle x_\rho \; | \; \rho \in \tau \rangle ).
\end{equation}
This is Alexander duality for monomial ideals: the minimal generators of $J$ encode the primary decomposition of $B$, and vice versa. Primary decomposition is convenient because it partitions the associated variety into irreducible components: this is the content of \cref{eq:primary_decomp}. Namely, the monomial generators $x^\tau$ of $J$ for $\tau = (\rho_1, \dots, \rho_m)$ encode the irreducible components $x_{\rho_1} = \dots = x_{\rho_m} = 0$ of $Z(\Sigma)$. This makes it especially clear what is going on: we exclude $x_{\rho_1} = \dots = x_{\rho_m} = 0$ from $V_\Sigma$ if $u_{\rho_1}, \dots, u_{\rho_m}$ do not span a cone in $\Sigma$.

Now let us spend some time elaborating on this construction. In the same way that the homogeneous coordinate ring $\mathbb{C}[x_0, \dots, x_n]$ is useful for studying $\mathbb{P}^n$ (e.g., for expressing its subvarieties), the \textit{homogeneous coordinate ring} or \textit{Cox ring} $S$ of $V_\Sigma$ will play a crucial role. The codimension-$m$ linear subspaces $\mathbb{V}(x_\rho \; | \; \rho \in \sigma) \subset \mathbb{C}^{|\Sigma(1)|}$ for $\sigma \in \Sigma(m)$ an $m$-cone in the fan descend to codimension-$m$ subvarieties of $V_\Sigma$. We call subvarieties of this type \textit{toric subvarieties}, as they are the prime torus-invariant subvarieties. We will also use the notation $D_\rho$ for $V_{\Sigma,\rho}$ with $\rho \in \Sigma(1)$ a one-cone (i.e., a ray). The $D_\rho$ are called \textit{prime toric divisors}. The discussion of the previous paragraph can be summarized as the statement that toric subvarieties intersect when all of their associated rays form a cone in the fan: in this way, the fan quite explicitly encodes the intersection structure of the toric variety.

Let us return to our claim that $\mathrm{Cl}(V_\Sigma)$ as defined above was the class group of Weil divisors on $V_\Sigma$. We recall that Weil divisors are formal integer linear combinations of irreducible complex codimension-one subvarieties in $V$, and the class group is an algebraic geometric analog to homology, given by the Weil divisors modulo principal divisors (those divisors given by a single meromorphic function). Indeed, if $V$ is smooth, $\mathrm{Cl}(V)$ is directly related to singular homology: $\mathrm{Cl}(V_\Sigma) \cong H^2(V_\Sigma,\mathbb{Z})$ (\cite{cls}, Th. 12.3.2). In particular, the (Poincar\'e duals of) divisor classes will all belong to the $(1,1)$ piece of the Hodge decomposition of de Rham cohomology.\footnote{There are technically some subtleties in thinking about differential forms on singular spaces, but we mostly care about differential forms on toric varieties for the purposes of restricting them to smooth hypersurfaces, so this won't impact us.} However, in general we only have $\mathrm{Pic}(V_\Sigma) \cong H^2(V_\Sigma,\mathbb{Z})$ for $\mathrm{Pic}(V_\Sigma) \subset \mathrm{Cl}(V_\Sigma)$ the Picard group of Cartier divisors classes on $V_\Sigma$. We will return to the Picard group later in the section. Sometimes we will use brackets to denote the class $[D]$ of a divisor $D$, while often for brevity we will let $D$ denote both a divisor and its class. 

Let us now interpret \cref{eq:SES}. Elements $a \in \mathbb{Z}^{|\Sigma(1)|}$ are interpreted as \textit{torus-invariant divisors}\footnote{We mention that it can be convenient in algebraic geometry to let divisors be rational or even real linear combinations of irreducible subvarieties: these are called $\mathbb{Q}$-divisors and $\mathbb{R}$-divisors, respectively. These won't be strictly necessary for us, but we refer to them when comparing with triangulation theory in \cref{tab:dict} for completeness. Torus-invariant divisors correspond to integral height vectors, while a general, real-valued height vector is naturally thought of as a torus-invariant $\mathbb{R}$-divisor.} $D_a = \sum_\rho a_\rho D_\rho$ given by $\mathbb{V}(x^a)$, where $x^a$ is shorthand for $\prod_\rho x_\rho^{a_\rho}$. The map $\alpha$ maps a character to the divisor associated to its unique extension to all of $V_\Sigma$, which is a principal torus-invariant divisor. \cref{eq:SES} is then the statement that on a toric variety, the class group can be constructed as torus-invariant divisors modulo principal torus-invariant divisors. In particular, the class group is generated by torus-invariant divisors. \cref{eq:SES} also enables us to compute the class of a torus-invariant divisor $D$ by applying $\beta$ (indeed, given a matrix representation $Q$, the class of the $i$th prime toric divisor is the $i$th column of $Q$).

The map $\beta$, as well as any choice of matrix representation $Q$ for it, are referred to as a \textit{class group grading}, as they provide a $\mathrm{Cl}(V_\Sigma)$-grading of the torus-invariant divisors and their associated monomials in $S$ (and thus, by linearity, all of $S$). We mention that $Q$ is also commonly referred to in supersymmetric/string theoretic applications as the \textit{gauged linear sigma model (GLSM) charge matrix}, as it encodes the $U(1)$ charges of fields in a supersymmetric GLSM with superpotential $W = 0$ whose vacuum manifold, modulo gauge transformations, is a toric variety.

Let us briefly exemplify the concepts introduced thus far in an example. 
\begin{example}
    Let us construct the $r$th Hirzebruch surface $\mathscr{H}_r$ as a toric variety. Let $N = \mathbb{Z}^2$, in which case we can write the minimal generators of the rays (one-cones) of the fan for $\mathscr{H}_r$ as the columns of the following matrix. 
    \begin{equation}
        \label{eq:hirz_rays}
        \mathbf{A}_r = \begin{pNiceMatrix}[first-row]
            1 & 2 & 3 & 4 \\
            0 & -r & 1 & -1 \\
            1 & -1 & 0 & 0
        \end{pNiceMatrix}.
    \end{equation}
    $\mathbf{A}_r^\top$ represents the map $\alpha$ in \cref{eq:SES}. In particular, $|\Sigma(1)| = 4$. We let $u_i$ denote the $i$th column of this matrix (i.e., the minimal generator of the $i$th ray in $\Sigma(1)$). The $u_i$ generate $N$, so we have no finite factor/torsion in $G$ or $\mathrm{Cl}(\mathscr{H}_r)$: in particular, $\mathrm{Cl}(\mathscr{H}_r) \cong \mathbb{Z}^2$. $\beta$ in \cref{eq:SES} is represented by a matrix $Q$ (i.e., a GLSM charge matrix), a choice of which fixes a basis for the class group. We choose the following.
    \begin{equation}
        \label{eq:hirz_charge}
        Q = \begin{pNiceMatrix}[first-row]
            1 & 2 & 3 & 4 \\
            1 & 1 & r & 0 \\
            0 & 0 & 1 & 1
        \end{pNiceMatrix}.
    \end{equation}
    This choice involved fixing a basis of linear relations of the $u_i$ (the rows of $Q$). In particular, the columns of $Q$ are the classes of the prime toric divisors $D_i$: i.e., we've chosen a basis for $\mathrm{Cl}(\mathcal{H}_r)$ consisting of $D_1 = D_2$ and $D_4$. The group $G \cong (\mathbb{C}^*)^2 \ni (\lambda, \tau)$ is the subgroup of $(\mathbb{C}^*)^4$ generated by $(\lambda^1, \lambda^1, \lambda^r, \lambda^0)$ and $(\tau^0, \tau^0, \tau^1, \tau^1)$, so performing the quotient of $\mathbb{C}^4$ imposes the scaling relations
    \begin{equation}
        (x_1, x_2, x_3, x_4) \sim (\lambda \, x_1, \lambda \, x_2, \lambda^r \tau \, x_3, \tau \, x_4).
    \end{equation}
    There is a unique fan which uses all of these one-cones. These induce the following irrelevant ideal $B$ and Stanley-Reisner ideal $J$.
    \begin{equation}
        \begin{aligned}
            B &= \langle x_1x_3, x_1x_4, x_2x_3, x_2x_4 \rangle , \\
            J &= \langle x_1x_2, x_3x_4 \rangle .
        \end{aligned}
    \end{equation}
    \label{ex:hirz_1}
\end{example}

We will be interested in the singularity structure of toric varieties and their hypersurfaces, so we note now that the singular locus $\mathrm{Sing} \, V_\Sigma$ of $V_\Sigma$ is a union of toric subvarieties as follows (\cite{cls}, Prop. 11.1.2).
\begin{equation}
    \mathrm{Sing} \, V_\Sigma = \bigcup_{\substack{\sigma \in \Sigma \\ \sigma \text{ not smooth}}} V_{\Sigma, \sigma} .
\end{equation}
We recall that a $k$-dimensional cone in $N_\mathbb{R}$ is \textit{smooth} if its minimal generators can be extended to a $\mathbb{Z}$-basis for $N$: in such a case, the affine toric variety associated to the cone is isomorphic to $\mathbb{C}^k$.\footnote{A cone fails to be smooth immediately if $\sigma$ has greater than $\dim N$ generators: this is much of our motivation to consider only simplicial fans.}
To a maximal cone $\sigma \in \Sigma(n)$ we can associate a polytope $\Pi_\sigma$ given by
\begin{equation}
    \label{eq:simplex_poly}
    \Pi_\sigma = \conv{\{u_\rho \; | \; \rho \in \sigma\} \cup \{\mathbf{0}\}} .
\end{equation}
If the only non-vertex lattice points in $\Pi_\sigma$ lie on the affine hyperplane containing the $u_\rho$, then the affine toric variety associated to $\sigma$ is said to have canonical singularities. The specific details of classifications of singularities will not be important to us, so we defer the reader to \S11.4 of \cite{cls} for further details in the toric context. What matters for us is that each of these non-vertex lattice points induces a canonical singularity (\cite{cls}, Prop. 11.4.17). These can be resolved by including each of these interior lattice points as minimal generators of one-cones in the fan: i.e., blowing up the singularities into new exceptional prime toric divisors. The result has only terminal singularities (\cite{cls}, Prop 11.4.12).

From this discussion we can conclude that to avoid canonical singularities, one is motivated to restrict to fans $\Sigma$ whose minimal generators $u_\rho$ satisfy $\{u_\rho\}_{\rho \in \Sigma(1)} = \partial \Delta^\circ \cap N$ --- that is, they constitute \textit{all} of the boundary points of a polytope --- such that $\Delta^\circ$ has a unique interior point (the origin). Such polytopes are called \textit{canonical polytopes} \cite{Kasprzyk:sca}: we will see in \cref{sec:ref_poly} why one may be motivated to restrict further to reflexive polytopes in the context of Calabi--Yau geometry, but we note that canonical polytopes are already useful in their own right for the construction of Calabi--Yau manifolds \cite{Kasprzyk:sca}.

The polytopes $\Pi_\sigma$ also allow us to conveniently introduce the intersection theory of $V_\Sigma$, which is well-defined for simplicial $\Sigma$. The intersection number associated to $n$ distinct prime toric divisors on a toric $n$-fold is
\begin{equation}
    D_1 \cdot D_2 \cdot {\dots} \cdot D_n = 
    \begin{cases}
        \frac{1}{\mathrm{vol} \, \Pi_\sigma} & \sigma = \cone{\rho_1, {\dots}, \rho_n} \in \Sigma, \\
        0 & \text{otherwise.}
    \end{cases}
\end{equation}
That is, the intersection number is $1$ if $u_1, \dots, u_n$ span a smooth cone $\sigma \in \Sigma$; if they span a non-smooth cone, the number is fractional, signaling that the affine toric variety associated to $\sigma$ was singular. It turns out that all other intersection numbers can be deduced from these, by exploiting the linear relations among the prime toric divisors.
For more details on this computation, see e.g., \cite{Demirtas:2022hqf}. 

Another important topological datum is the total Chern class: its Poincar\'e dual is the sum of all toric subvarieties $V_{\Sigma,\sigma}$ (\cite{cls}, Prop. 13.1.12).

\subsection{Line Bundles and Newton Polytopes}

\label{sec:line_bundle_newton}

A variety can be at most two out of the following three: toric, Calabi--Yau, and compact.\footnote{Calabi--Yau toric varieties must have minimal generators belonging to a common affine hyperplane, which is incompatible with the fan being complete. In the language of triangulation theory compact (Calabi--Yau) toric varieties arise from totally cyclic (acyclic) vector configurations. We note that non-compact toric Calabi--Yau varieties play an important role in geometric engineering and supersymmetric field theory: see, e.g., \cite{Katz:1996fh, Leung:1997tw, Intriligator:1997pq} and also \cite{Eckhard:2020jyr} for a review.} To construct compact Calabi--Yau varieties, then, we will study toric varieties and consider Calabi--Yau hypersurfaces inside of them. We therefore turn our attention to hypersurfaces in toric varieties, the nicest of which are given as zero loci of sections of line bundles. 

We recall that line bundles are in correspondence with Cartier divisor classes, or Weil divisors classes which can be locally represented as the zero locus of a single equation.\footnote{The first Chern class maps line bundles to their corresponding Cartier divisor classes, while to go the other way one can take the equations $f_i = 0$ defining a representative Cartier divisor on an open chart $\{U_i\}$ and construct transition functions $f_i/f_j$ between these charts, which define a line bundle.} We will often use the two interchangeably. For simplicial toric varieties, the Picard group of Cartier divisors is free (\cite{cls}, Prop. 4.2.5) and has finite index in the class group (\cite{cls}, Prop. 4.2.7), so all Weil divisors admit a multiple that is Cartier. Such Weil divisors are called $\mathbb{Q}$-Cartier, and varieties where all Weil divisors are $\mathbb{Q}$-Cartier are called $\mathbb{Q}$-factorial (i.e., simplicial toric varieties are $\mathbb{Q}$-factorial). We also have that the \textit{Picard number} is $\mathrm{rank} \, \mathrm{Pic}(V_\Sigma) = \mathrm{rank} \, \mathrm{Cl}(V_\Sigma) = |\Sigma(1)| - \dim V_\Sigma$ (assuming completeness). A torus-invariant divisor $D = \sum_\rho a_\rho D_\rho$ is in particular Cartier if for every maximal cone $\sigma \in \Sigma(n)$ spanned by minimal generators $u_{\rho_1}, \dots, u_{\rho_n}$, there exists $m_\sigma \in M$ such that $\langle m_\sigma, u_{\rho_i} \rangle = -a_{\rho_i}$, $1 \leq i \leq n$ (\cite{cls}, Th. 4.2.8). The characters $m_\sigma$ are referred to as the Cartier data of $D$, and they furnish the local equations of $D$ on the canonical open charts of $V_\Sigma$ given by the affine toric varieties associated to the maximal cones. The Cartier property lifts to divisor classes.

We now consider sections of line bundles on $V_\Sigma$.\footnote{While only Cartier divisors are associated to line bundles, Weil divisors more generally are associated to rank-one reflexive sheaves, and the discussion of this paragraph immediately applies to these sheaves as well.} Homogeneous polynomials in $S$ with a Cartier $\mathrm{Cl}(V_\Sigma)$-grading $[D]$ are global sections of the line bundle $\mathcal{O}(V_\Sigma, [D])$ associated to $[D]$. More generally, Laurent monomials $x^a$ ($a \in \mathbb{Z}^{|\Sigma(1)|}$) in the field of fractions $\mathrm{Frac}(S)$ --- whose divisors are the torus-invariant divisors $\sum_\rho a_\rho D_\rho$ --- exhaust the local sections of $\mathcal{O}(V_\Sigma, D)$. These can be parameterized by the kernel of the map $\beta : \mathbb{Z}^{|\Sigma(1)|} \to \mathrm{Cl}(V_\Sigma)$ from \cref{eq:SES}, which is just the $M$ lattice. Explicitly, fix a class $[D]$ represented by a torus-invariant divisor $D = \sum_\rho a_\rho D_\rho$: the Laurent monomials of degree $[D]$ are given as follows.
\begin{equation}
    x^{\langle m, D \rangle} := \prod_\rho x^{\langle m, u_\rho \rangle + a_\rho}.
\end{equation}
The global sections are exactly those which have non-negative exponents and thus belong to $S$ (see, e.g., Prop. 5.1.4 in \cite{cls}). These correspond to a polyhedron in $M$: for compact $V_\Sigma$, it is a polytope known as the \textit{Newton polytope} $\Delta_D \subset M$.
\begin{equation}
    \Delta_D = \{ m \in M \; | \; \langle m, u_\rho \rangle + a_\rho \geq 0 \}.
\end{equation}  
Choices of different $a$ with the same $\mathrm{Cl}(V_\Sigma)$-grading (torus-invariant divisors in the same class) have Newton polytopes related by translation in $M$. A line bundle/Cartier divisor is \textit{effective} if it admits a global section. The effective torus-invariant divisors belong to the positive orthant in $\mathbb{Z}^{|\Sigma(1)|}$, which is mapped by $\beta$ to the \textit{effective cone} $\mathrm{Eff}(V_\Sigma) \subset \mathrm{Cl}(V_\Sigma)$. Effective divisors are those whose Newton polytopes contain lattice points.

The convex geometry of the Newton polytope $\Delta_D$ encodes rich and diverse structure. Its lattice points encode global sections of $D$, while its faces (especially their intersection structure) encode the \textit{basepoint loci} of $D$, where global sections of $\mathcal{O}(V_\Sigma, D)$ always vanish. In particular, every ray $\rho \in \Sigma(1)$ arising in the fan imposes an affine linear constraint on $\Delta_D$ and thus induces a face 
\begin{equation}
    \label{eq:poly_faces}
    \Theta_{D,\rho} := \setbuilder{m \in M}{\langle m, u_\rho \rangle = -a_\rho} \cap \Delta_D.
\end{equation}
This face need not have codimension one: indeed, it may even be empty. The lattice points on this face correspond to global sections that do not depend on $x_\rho$. If every global section of $D$ depends on at least one of a set of homogeneous coordinates $x_{\rho_1}, \dots, x_{\rho_k}$, then $D_{\rho_1} \cap \dots \cap D_{\rho_k}$ is a basepoint locus of $D$ if it doesn't fall in the exceptional set $Z(\Sigma)$. This corresponds to the face $\Theta_{D,\rho_1} \cap \dots \cap \Theta_{D,\rho_k}$ being empty. 

Of particular interest to us is the \textit{anticanonical class} 
\begin{equation}
    \overline{K} := -K = \sum_\rho [D_\rho] = c_1(V_\Sigma).
\end{equation}
($K$ being the canonical divisor) because divisors in this class are Calabi--Yau by adjunction. Its Newton polytope $\Delta_{\overline{K}} \subset M$ agrees with the more widespread notion of the \textit{polar dual} $(\Delta^\circ)^\circ = \Delta$ of a polytope $\Delta^\circ \subset N$,\footnote{This motivates the notation: $^\circ$ is the polar dual operation.} the convex hull of the $u_\rho$. In particular, unlike most divisors, the Newton polytope of the anticanonical divisor is uniquely determined by the vertices of $\Delta^\circ$. That is, it is agnostic to the presence of other rays whose minimal generators are lattice points of $\Delta^\circ$. The pair $\Delta^\circ, \Delta$ is especially important when both are \textit{lattice polytopes} (i.e., their vertices are integral), in which case we call them both \textit{reflexive polytopes}.

Reflexive polytopes are canonical polytopes, and are especially amenable to constructing smooth anticanonical hypersurfaces, as we will review when we discuss Batyrev's construction in \cref{sec:batyrev}. For now, we will just introduce one concept for reflexive polytopes that we will exploit several times. Namely, there is an inclusion-reversing duality of faces for pairs of reflexive polytopes, relating $k$-dimensional faces of one to $(n-k-1)$-faces of the other. This duality is convenient for at least the following reason. Given a subset of one-cones $\rho_1, \dots, \rho_k$ or minimal generators $u_1, \dots, u_k$, define their \textit{minface} $\Theta^\circ \subset \Delta^\circ$ to be the smallest face of $\Delta^\circ$ containing each minimal generator \cite{Braun:2017nhi}.\footnote{We comment that in triangulations, this is known as the \textit{carrier} of a set of points \cite{De_Loera2010-ss}.} If one restricts to the toric subvariety $x_1 = \dots = x_k = 0$ associated to the subset, the only sections of the anticanonical bundle which do not vanish correspond to the points on the dual face $\Theta \subset \Delta$.

\subsection{Normal Fans and K\"ahler Cones}

\label{sec:toric_bir}

We conclude this review of toric geometry by initiating a discussion of their birational geometry, though this will continue into the next section as we connect with triangulation theory to discuss the secondary fan. In the strictest sense, all toric varieties of the same dimension are birational, because they each contain equidimensional dense tori. However, following \S 15.3 in \cite{cls}, we take a birational map to be a torus-equivariant map which is an isomorphism in codimension 1 (for non-toric varieties --- such as compact Calabi--Yau varieties --- we also take birational maps to be isomorphisms in codimension one). This means that toric varieties with exactly the same one-cones but different higher-dimensional cones are birational. Thus, to study the birational equivalence class of $V_\Sigma$, we are interested in the set of toric varieties whose fans have exactly the one-cones $\Sigma(1)$. Ultimately, the object we will employ to study these toric varieties --- the secondary fan, the topic of \cref{sec:sec_fan} --- will describe all fans $\Sigma'$ whose one-cones are \textit{contained in} $\Sigma(1)$. If $\Sigma'(1) = \Sigma(1)$, we say $\Sigma'$ is a \textit{fine} fan (with respect to $\Sigma(1)$). When a fixed $\Sigma(1)$ is understood, we will also refer to the associated toric varieties as being fine or not. The toric varieties associated to fine fans will form a birational class. We comment that for fixed $\Sigma(1)$, non-fine toric varieties are related to the fine toric varieties by blowups, which are not isomorphisms in codimension one, so we will not call them birational maps, though they do define isomorphisms between dense open subsets.

We begin by introducing a map from effective divisors $D$ on $V_\Sigma$ to fans $\Sigma'$ with one-cones contained in $\Sigma(1)$. This is achieved by mapping $D$ to the \textit{normal fan} $\Sigma_D$ of its Newton polytope $\Delta_D$. We can construct $\Sigma_D$ as follows. The maximal cones of $\Sigma_D$ are in bijection with the vertices of $\Delta_D$. In particular, to each (generically non-integer) vertex $m \in M \otimes \mathbb{Q}$ we associate the maximal cone given by the positive hull of the rays $\rho$ such that $m$ saturates the hyperplane inequality $\langle m, u_\rho \rangle + a_\rho \geq 0$. Because the normal fan is invariant under polytope translation, $\Sigma_D$ depends only on $[D]$. Importantly, from our earlier discussion of basepoint loci, we can conclude that the basepoint loci of $D$ are contained in $Z(\Sigma_D)$: if an intersection of $\Theta_{D,\rho_i}$ is empty, it contains no vertex. In particular, $D$ is ample and basepoint-free on the toric variety $V_{\Sigma_D}$ (\cite{cls}, Prop. 6.1.10). Consequently, normal fans include only a subset of the one-cones in $\Sigma(1)$ because if $D_\rho$ belongs to the basepoint locus of $D$, $\rho$ will not be included in $\Sigma_D$. 

A fan is \textit{regular} if it can be realized as the normal fan of some polytope \cite{fulton1993introduction}, otherwise it is \textit{irregular}. All projective toric varieties arise from regular fans (\cite{cls}, Prop. 7.2.9). We are primarily interested in K\"ahler hypersurfaces $X$, which in particular are projective, so we will restrict our discussion to regular fans.

The torus-invariant divisor $D$ is Cartier on the toric variety $V_D$ associated to the normal fan $\Sigma_D$ of $\Delta_D$ if and only if $\Delta_D$ is lattice, as for this variety the Cartier data are exactly the vertices of $\Delta_D$. A variety is \textit{Fano} if its anticanonical divisor is ample: for a toric variety, this means taking the fan $\Sigma_{\overline{K}}$. A variety is \textit{Gorenstein} if the anticanonical divisor is Cartier: thus, a toric variety is Gorenstein Fano if and only if it arises as the normal fan of a reflexive polytope, thought of as the Newton polytope of $\overline{K}$. An important property of the normal fan of $\Delta_{\overline{K}}$ of any fan $\Sigma$ is that its only one-cones are those whose minimal generators were \textit{vertices} of the convex hull of the minimal generators of $\Sigma$. In the event that these generators were the lattice points on the boundary of a polytope $\Delta^\circ$, this fan is referred to as the \textit{central fan} of $\Delta^\circ$. The one-cones of this fan are generated by the vertices of $\Delta^\circ$ and more generally its $k$-cones are the cones over the $(k-1)$-faces of $\Delta^\circ$.
This is the starting point for Batyrev's construction of smooth Calabi--Yau hypersurfaces \cite{Batyrev:1993oya}, which we will discuss in greater detail in \cref{sec:frst_cy}.

We are now prepared to understand how the normal fan construction relates to toric birational geometry. The map $D \mapsto \Sigma_D$ sends effective divisors on $V_\Sigma$ to regular fans whose one-cones are a subset of $\Sigma(1)$. It turns out that this map is piecewise constant on $\mathrm{Eff}(V_\Sigma)$. It is then natural to ask what the level sets are. In the following, we will only consider simplicial fans. For such fans, they are solid polyhedral cones, called \textit{secondary cones} (differing slightly from the secondary cones defined in \cref{sec:triang}, as we will clarify momentarily). For $\Sigma'$ satisfying $\Sigma'(1) \subseteq \Sigma(1)$, its secondary cone $\Gamma_{\Sigma'}$ with respect to $\Sigma(1)$ is (\cite{cls}, Prop. 15.2.1)
\begin{equation}
    \label{eq:sec_cone}
    \Gamma(V_{\Sigma'}) = \bigcap_{\sigma \in \Sigma'} \cone{D_\rho \; | \; \rho \notin \sigma} \subset \mathrm{Eff}(V_\Sigma) \subset \mathrm{Cl}(V_\Sigma).
\end{equation}
That is, the normal fan of the Newton polytope of any divisor class belonging to the relative interior of $\Gamma_\Sigma$ is $\Sigma$. In the next section, we will discuss how these secondary cones subdivide the effective cone $\mathrm{Eff}(V_\Sigma)$. Because these secondary cones characterize birationally related toric varieties, this is a geometric organization of toric birational geometry. We stress that irregular fans do not have a secondary cone, and moreover that we are only considering simplicial fans in this context.

The secondary cone $\Gamma_{\Sigma'}$ associated to a simplicial fan $\Sigma'$ has a convenient geometric interpretation (\cite{cls}, Prop. 15.1.3.).
\begin{equation}
    \label{eq:nef}
    \Gamma(V_{\Sigma'}) \cong \mathrm{Nef}(V_{\Sigma'}) \times \mathbb{R}_+^{|\Sigma(1)| - |\Sigma'(1)|}.
\end{equation}
Here, $\mathrm{Nef}(V_{\Sigma'})$ is the \textit{nef (numerically effective) cone} of $V_{\Sigma'}$, or the cone of divisors which have non-negative pairing with every curve class in the \textit{Mori cone} of effective curves.

That is, $\mathrm{Nef}(V_{\Sigma'})$ is dual to the Mori cone. Because this pairing is Poincar\'e dual to the evaluation of two-forms on curves, the interior of the nef cone is (up to Poincar\'e duality) the \textit{K\"ahler cone} of $2$-forms (in particular, $(1,1)$-forms) assigning strictly positive volume to all holomorphic cycles. We stress that for fine fans, \cref{eq:nef} is strict equality: secondary cones \textit{are} nef cones. As we noted, the term ``secondary cone'' means something similar but different in triangulation theory, to be clarified fully in the next section. To avoid confusion we will refer to fine secondary cones in the sense of \cref{eq:sec_cone} as nef/K\"ahler cones both in \cref{tab:dict} and for the remainder of the text. 

Let us conclude by exhibiting some of these recent concepts by continuing our study of the Hirzebruch surfaces of \cref{ex:hirz_1}. 
\begin{example}
    Recall the $r$th Hirzebruch surface from \cref{ex:hirz_1}. Let us restrict to $\mathscr{H}_r$ for $r \in \{1,2\}$, as in these cases the minimal generators of the fan $\Sigma_r$ for $\mathscr{H}_r$ are, in fact, the lattice points of a 2D reflexive polytope $\Delta_r^\circ$:
    \begin{equation}
        \Delta^\circ = \conv{\{u_1, u_2, u_3, u_4\}}.
    \end{equation}
    Thus, this is a nice toy model for our extended discussion of reflexive polytopes in later sections. In fact, $\Sigma_r$ is the central fan of $\Delta_r^\circ$: all of its cones are given by faces of $\Delta_r^\circ$. Equivalently, it is the normal fan of the polar dual polytope $\Delta_r$ (the Newton polytope of the anticanonical divisor $\overline{K}$). For example, let us fix $r = 1$.
    \begin{equation}
        \Delta_1 = \conv{\{(-1, -1), (-1, 2), (1, -1), (1, 0)\}}.
    \end{equation}
    Under the normal fan construction, the maximal cone generated by $u_3$ and $u_1$ in $\Sigma_1$ (for example) is associated to the vertex $(-1,-1)$ of $\Delta_1$, as it lies on the faces $\Theta_{\overline{K},u_3}$ and $\Theta_{\overline{K},u_1}$. Repeating this for the other vertices recovers the whole fan $\Sigma_1$.

    Thus, $\mathscr{H}_1$ is a Gorenstein Fano toric variety (as is $\mathscr{H}_2$). It is worth stressing that for higher-dimensional reflexive polytopes, it is rare for the central fan to be simplicial and smooth. We will return to this point in \cref{sec:frst_cy}.

    We can construct the secondary cone of $\mathscr{H}_r$ as well. Because its fan is fine, this will also be its nef/K\"ahler cone. Applying \Cref{eq:sec_cone}, and recalling how we expressed our prime toric divisors in a basis in \cref{ex:hirz_1}, we find
    \begin{equation}
        \begin{aligned}
            \Gamma_{\Sigma} 
            &= \cone{D_2, D_4} \cap \cone{D_1, D_4} \cap \cone{D_1, D_3} \cap \cone{D_2, D_3} \\
            &= \cone{D_1, D_3}.
        \end{aligned}
    \end{equation}
    As a nice sanity check, we see the anticanonical class $\overline{K} = D_1 + D_2 + D_3 + D_4$, which is $(2+r,2)$ in our basis from \cref{ex:hirz_1}, belongs to the relative interior of the nef cone $\cone{(1,0),(r,1)}$ --- meaning $\overline{K}$ is ample and the toric variety is Fano --- if and only if $r \in \{0,1\}$. This agrees with our earlier comment about when the minimal generators of the fan $\Sigma_r$ form a reflexive polytope. 
    \label{ex:hirz_2}
\end{example}

\section{Review of the Secondary Fan}
\label{sec:sec_fan}

We continue our review of the triangulation/fan/toric variety classification problem stated in \cref{sec:intro}, but at this point the separate paths of triangulation theory and toric geometry intersect. We emphasize that the content of this section is still a review of known results, though we hope our exposition is illuminating. First we will discuss the secondary fan and its role in triangulation theory and toric geometry. Then, we will introduce an algorithm for constructing all fine regular simplicial fans of a vector configuration.

\subsection{Theory}

In \cref{sec:triang} and \cref{sec:toric}, we converged on the notion of the \textit{secondary cone} from two directions. Fix a vector configuration $\mathbf{A}$ in $\mathbb{R}^n$ with primitive integral vectors, thinking of $\mathbb{R}^n$ as $\mathbb{Z}^n \otimes \mathbb{R}$, with $\mathbb{Z}^n$ being the lattice $N$ upon fixing a basis. Let $\mathscr{T}$ be some fine regular triangulation of $\mathbf{A}$. By taking supports, this induces a simplicial fan $\Sigma$. In particular, $\Sigma(1)$ is the set of rays generated by the vectors in $\mathbf{A}$, so the minimal generators of $\Sigma$ are just the vectors in $\mathbf{A}$. More generally, triangulations of $\mathbf{A}$ and simplicial fans constructed from one-cones $\Sigma(1)$ are in one-to-one correspondence, and thus can be identified. We note, as discussed in \cref{sec:triang}, that this bijection does not continue to hold if we relax to subdivisions and to fans, because several subdivisions of $\mathbf{A}$ can in general induce the same fan: we will say more about this in \cref{sec:translate}. 

In \cref{sec:regular}, we introduced the lifting construction, which mapped a height vector $\omega \in \mathbb{R}^{|\mathbf{A}|}$ to a triangulation (or perhaps just a subdivision). On the other hand, in \cref{sec:toric_bir} we used normal fans of Newton polytopes to map divisor classes $D \in \mathrm{Eff}(V_\Sigma) \subset \mathrm{Cl}(V_\Sigma)$ to fans $\Sigma_D$. In both setups, we identified that the map from vectors were piecewise constant on polyhedral cones known as secondary cones. In particular, the regular subdivision lifting by $\{a_1, \dots, a_{|\mathbf{A}|}\}$ induces a fan which is equal to the normal fan of the Newton polytope associated to the torus-invariant divisor $\sum_i a_i D_i$.

We have also glimpsed how distinct secondary cones for distinct triangulations/fans sew together: in \cref{sec:regular}, we saw how triangulations with secondary cones that intersect at a facet are related by flips. In general, all secondary cones of all triangulations/fans arising from $\mathbf{A}/\Sigma(1)$ organize into a polyhedral fan: namely, the \textit{secondary fan} \cite{Gelfand:1994}. This is an object of central importance in triangulation theory and toric geometry alike, and is the subject of this section. In particular, the intersections between secondary cones in the secondary fan encode how the associated regular triangulations/fans are combinatorially related (in triangulation theory) and how the induced toric varieties are related by maps called extremal contractions (in toric geometry). In particular, many such contractions are birational maps in the sense of \cref{sec:toric_bir}, so the secondary fan encodes a toric variety's birational geometry.\footnote{One way of understanding this is that toric varieties are quotients in the sense of geometric invariant theory (GIT). It it well understood in GIT that the well-behaved quotients of a space (such as $\mathbb{C}^{|\Sigma(1)|}$) by a group are not unique: one must additionally choose a semistable locus and excise its complement (i.e., $Z(\Sigma)$). This choice is parameterized in GIT by an (effective) line bundle, which we recall corresponds to a Cartier divisor in the Picard group (and for toric varieties, this is a finite index subgroup of the class group). The resulting decomposition of the effective cone into chambers associated to distinct GIT quotients produces the \textit{GIT fan}, of which the secondary fan is a special case. Indeed, the theory of toric birational geometry is a special case of the study of the variation of geometric invariant theory quotients (VGIT) \cite{dolgachev1998variation, thaddeus1996geometric}.}

We note that triangulation theory considers secondary cones and the secondary fan to live in height space $\mathbb{R}^{|\mathbf{A}|} = \mathbb{R}^{|\Sigma(1)|}$ while toric geometry constructs them in the class group $\mathrm{Cl}(V_\Sigma)_\mathbb{R}$. The difference between the dimension of these spaces $|\mathbf{A}| - \dim \mathrm{Cl}(V_\Sigma)_\mathbb{R}$ is the rank of the set of vectors/rays (usually their dimension). In spite of this, the constructions are equivalent: the higher-dimensional triangulation-theoretic secondary cones/fan have \textit{lineality spaces}, or contain entire linear subspaces, which one can freely quotient out without loss of information, resulting in the lower-dimensional toric geometric secondary cones/fan.

Triangulation theory and toric geometry do each incorporate both manifestations of the secondary fan. In triangulation theory, quotienting the secondary fan in $\mathbb{R}^{|\mathbf{A}|}$ by the lineality space yields the \textit{chamber complex}, while $\mathbb{R}^{|\Sigma(1)|}$ in toric geometry is merely the space of torus-invariant divisors $\mathbb{Z}^{|\Sigma(1)|}$, tensored with $\mathbb{R}$. This also clarifies how to pass between the representations: one can pullback or pushforward using the map $\beta$ (the class group grading/GLSM charge matrix) from \cref{eq:SES}. The lineality space is exactly the kernel of $\beta$, or the trivial torus-invariant divisors, which are parameterized by the lattice $M$.

For the remainder of the text, we will let ``secondary fan'' refer strictly to the higher-dimensional fan in $\mathbb{R}^{|\mathbf{A}|} = \mathbb{R}^{|\Sigma(1)|}$, with lineality spaces, as this is how the term was originally defined. Because we will almost exclusively focus on fine triangulations/fans in this paper, we can exploit \cref{eq:nef} to refer to the lower-dimensional secondary cones as nef/K\"ahler cones.

We are especially interested in the collection of secondary cones in the secondary fan of $\mathbf{A}/\Sigma(1)$ which correspond to fine triangulations/fans. These are the fans that correspond to a birational equivalence class of toric varieties. It turns out that these cones organize into a convex subfan of the secondary fan (Cor. 5.3.14 \cite{De_Loera2010-ss} and Prop. 15.1.4 in \cite{cls}). In particular, the union of the nef/K\"ahler cones of fine fans in toric geometry is known as the cone of movable divisors $\mathrm{Mov}(V_\Sigma)$ (divisors with basepoint locus in codimension at least two), or the moving cone,\footnote{While the definition provided in \cref{eq:movable} doesn't obviously connect to the traditional notion of movable divisors (and the cone they generate), the two ideas indeed agree. This is Th. 15.1.10 in \cite{cls}, but the basic idea was already presented in \cref{sec:toric_bir}: a fan $\Sigma_D$ omits a ray $\rho$ if and only if the divisor $D_\rho$ is contained in the basepoint locus of $D$, so $\Sigma_D$ is fine if and only if $D$ has no codimension-one basepoint locus.} given as follows.
\begin{equation}
    \label{eq:movable}
    \begin{aligned}
        \mathrm{Mov}(V_\Sigma) 
        &= \bigcup_{\Sigma' \text{ fine}} \Gamma(V_{\Sigma'}) \\
        &= \bigcup_{\mathscr{T} \text{ fine}} \beta(\mathbf{C}(\mathbf{A}, \mathscr{T})) \\
        &= \bigcap_{\rho \in \Sigma(1)} \cone{D_{\rho'} \; | \; \rho' \neq \rho}.
    \end{aligned}
\end{equation}
The first equality is a definition, the second exploits that nef/K\"ahler cones are projections of secondary cones by $\beta$ from \cref{eq:SES}, and the third is a non-trivial result that allows one to avoid the computational problem of enumerating all fine fans/triangulations (\cite{cls}, Prop. 15.2.4). The toric varieties arising from the fine fans associated to some set of one-cones form a birational equivalence class in the sense defined in \cref{sec:toric_bir}.

We now turn our attention to how secondary cones are organized within the secondary fan. As we saw in \cref{sec:regular}, each facet of the secondary cone for $C(\mathbf{A},\mathscr{T})$ corresponds to an embedded circuit in triangulation theory, which either induces a flip transforming $\mathscr{T}$ into a distinct regular triangulation $\mathscr{T}'$ --- in which case the facet was $C(\mathbf{A},\mathscr{T}) \cap C(\mathbf{A},\mathscr{T}')$ --- or cannot be flipped, and demarcates a boundary of the secondary fan.\footnote{It's worth clarifying a subtlety at this point, first mentioned at the end of \cref{sec:dep_and_circ}. Facets of secondary cones given by intersections of two secondary cones always correspond to embedded circuits, but the converse is not always true. A flip corresponds to such a facet if and only if the flip relates two regular triangulations and the intermediate subdivision is also regular. If either condition fails, the flip is embedded but corresponds to no facet of the secondary cone, as either the flipped triangulation or the intermediate subdivision does not correspond to any height vector in the secondary fan. Neither of these scenarios are especially rare: for example, they both occur in the Kreuzer--Skarke database at small Hodge numbers.} In toric geometry, secondary cone facets correspond to maps called \textit{extremal contractions}, a notion from the minimal model program (see, e.g., \cite{kollar1998birational}, or \S15.4 in \cite{cls} for the toric case). These maps belong to one of three types, in a way that actually aligns with the three categories of \cref{sec:common_sig}.

If the facet is indeed shared by another secondary cone, say that of a simplicial fan $\Sigma'$, then we can ask if $V_\Sigma$ and $V_{\Sigma'}$ have the same Picard number --- i.e., $\Sigma(1) = \Sigma'(1)$. If so, this is a \textit{flip} or \textit{flipping contraction}, a birational map (in particular, an isomorphism in codimension one). This unfortunately overloads the term ``flip'' from \cref{eq:flip}, but this cannot be helped, as the name is standard in both triangulation theory and algebraic geometry. In this case, the associated circuit has signature $(a,b)$ for $a,b > 1$.

On the other hand, if the Picard numbers of $V_\Sigma, V_{\Sigma'}$ differ by one (the only other possibility) then this is a \textit{blowup/blowdown} or \textit{divisorial contraction}, a map involving the blowup/blowdown of an exceptional prime toric divisor (corresponding to the one-cone not shared by $V, V'$). In triangulation theory, this is known as a ``deletion/insertion flip'' (but we stress that this is not a flip in the algebraic geometric sense). In this case, the associated circuit has signature $(a,1)$ for $a > 1$.

Finally, if this facet intersected no other secondary cone and was therefore a boundary of the secondary fan, it corresponds to a \textit{fibering contraction} mapping $V_\Sigma$ to the lower-dimensional base of an (Iitaka) fibration enjoyed by $V_\Sigma$. This fibration will not be important for us, but we include a discussion of how to construct the base and fiber of this fibration in \cref{sec:translate}. In this case, the associated circuit has signature $(a,0)$ for $a > 1$. We mention in passing that for reflexive polytopes, it is well understood that these fibrations of the ambient toric variety are inherited by the anticanonical hypersurface \cite{Kreuzer:2000qv}. Thus, toric varieties admitting fibering contractions furnish examples of either genus-one ($a = 3$) or $T^4$/$K3$-fibered ($a = 4$) Calabi--Yau geometries, where a face of the nef/K\"ahler cone corresponds to shrinking fiber limit.\footnote{The case $a = 2$ is somewhat degenerate: the CY is fibered over a toric threefold with generic fiber given by a pair of points.}

We stress that by our definition of birational map from \cref{sec:toric_bir}, only flipping contractions are birational. Let us now illustrate these concepts in an example.

\begin{example}
    We return to the Hirzebruch surfaces $\mathscr{H}_r$. We recall that in each case, the nef/K\"ahler cone is $K = \cone{D_1, D_3} \subset \mathrm{Eff}(\mathscr{H}_r)$. Let us denote the vector configuration and triangulation associated to $\mathscr{H}_r$ as $\mathbf{A}_r$ and $\mathcal{T}_r$, respectively. We want to construct the complete secondary fan of $\mathbf{A}_r$, so we first ought to construct the secondary cone $C(\mathbf{A}_r, \mathcal{T}_r)$ associated to $\mathscr{H}_r$ (in $\mathbb{R}^4$, height space or the space of real-valued torus-invariant divisors) from which we can read off the embedded/flippable circuits. It is most conveniently defined via its hyperplanes. 
    \begin{equation}
        C(\mathbf{A}_r, \mathcal{T}_r) = \left\{ \omega \in \mathbb{R}^4 \; \Bigg| \; \begin{pmatrix}1 & 1 & 0 & -r \\ 0 & 0 & 1 & 1\end{pmatrix}\omega \geq 0 \right\}.
    \end{equation}
    This is a solid cone containing a two-dimensional lineality space (because $\dim \mathscr{H}_r = 2$) spanned by $(0, -r, 1, -1)$ and $(1, -1, 0, 0)$. These are the rows of the matrix $\mathbf{A}_r$ from \cref{eq:hirz_rays} giving the vectors/minimal generators for this configuration. They generate the lineality space because they are the two linearly independent linear evaluations associated to the vector configuration; equivalently, recalling \cref{eq:SES}, they represent the matrix $\alpha$ generating the kernel of the projection $\beta$ from torus-invariant divisors to divisor classes. 
    
    The two hyperplane normals specify the embedded circuits: they are each exactly the unique dependency of the associated circuit. Thus, we read off that the embedded circuits have support $\{1,2,4\}$ and $\{3,4\}$ which decompose into positive and negative components as $(\{1,2\}, \{4\})$ and $(\{3,4\}, \varnothing)$, respectively. In particular, this entails that their signatures are $(2,1)$ and $(2,0)$, respectively. Thus they are a deletion flip --- removing $u_4$ --- and an unflippable positive circuit, which correspond toric geometrically to a blowdown of the prime toric divisor $D_4$ and a fibering contraction, respectively. 

    At the present, we are interested in constructing the secondary fan, so we will only consider flippable circuits (we defer a discussion of the fibering contraction to \cref{ex:hirz_4}). Flipping this circuit to get a new triangulation $\mathcal{T}'_r$ involves removing $u_4$ and fusing the cones $\cone{u_1,u_4}$ and $\cone{u_2,u_4}$ to $\cone{u_1,u_2}$. The associated toric variety has charge matrix $(1,1,r)$, immediately recognizable as the weighted projective space $\mathbb{P}_{11r}$ (e.g., for $r = 1$, this is $\mathbb{P}^2$). This reproduces the familiar fact that $\mathscr{H}_r$ is the blowup of $\mathbb{P}_{11r}$ at a smooth point: the divisorial contraction associated to this triangulation flip was the associated blowdown.

    With respect to the four vectors of $\mathcal{T}_r$, the triangulation $\mathcal{T}'_r$ of $\mathbb{P}_{11r}$ then has secondary cone
    \begin{equation}
        C(\mathbf{A}_r, \mathcal{T}'_r) = \left\{ \omega \in \mathbb{R}^4 \; \Bigg| \; \begin{pmatrix}-1 & -1 & 0 & r \\ 1 & 1 & r & 0\end{pmatrix}\omega \geq 0 \right\}.
    \end{equation}
    which includes the circuit we could flip to return to $\mathscr{H}_r$ along with a non-flippable $(3,0)$ circuit. Thus, we have no more flippable circuits, and we have computed the full secondary fan. 

    Projecting down to the chamber fan using the charge matrix $Q$ from \cref{eq:hirz_charge} reproduces the K\"ahler cone of $\mathscr{H}_r$ and assigns $\cone{D_3, D_4}$ to $\mathbb{P}_{11r}$. This cone could've also been computed using \cref{eq:sec_cone}. In the end, we see that the chamber fan refines the effective cone $\mathrm{Eff}(\mathscr{H}_1) = \cone{(1,0),(0,1)}$. 

    There are no more remaining circuits to flip, so we conclude that the chamber fan has support $\cone{D_1, D_4}$ with two chambers: $\cone{D_1, D_3}$ for the fine triangulation/fan associated to $\mathscr{H}_r$ and $\cone{D_3, D_4}$ for the non-fine triangulation/fan associated to $\mathbb{P}_{11r}$. The full secondary fan (in our basis) is the pullback of this cone into $\mathbb{R}^4$ by the class group grading / GLSM charge matrix $Q$ from \cref{ex:hirz_1}.
    \label{ex:hirz_3}
\end{example}

\subsection{Computational Algorithm}

\label{sec:algo}

Having discussed theoretical aspects of the secondary fan, let us now explain how this knowledge can be applied computationally to solve the classification problem posed in \cref{sec:intro}. One can efficiently construct the secondary fan of some configuration $\mathbf{A}$ if they can compute all neighboring triangulations (those differing by $1$ flip) $\mathscr{T}'\in\mathrm{neighbors}(\mathscr{T})$ of some triangulation $\mathscr{T}$. See \cref{alg:secondary_fan}.\footnote{Since this algorithm involves constructing the neighbors of every triangulation $\mathscr{T}$, \cref{alg:secondary_fan} can easily be adapted to also output the flip graph of regular triangulations.}

\begin{algorithm}[H]
\caption{secondary fan}
\label{alg:secondary_fan}
\begin{algorithmic}
\State let $\mathscr{T}(\mathbf{A}, \omega)$ be a regular triangulation of a configuration, $\mathbf{A}$
\State let $\Sigma = \varnothing$ and $\text{checking} = \{\mathscr{T}(\mathbf{A},\omega)\}$
\While{$\text{len}(\text{checking}) > 0$}
    \State let $\mathscr{T}' = \mathrm{pop}(\text{checking})$
    \State update $\Sigma \to \Sigma\cup\{\mathscr{T}'\}$
    \For{$\mathscr{T}''\in \mathrm{neighbors}\left( \mathscr{T}' \right)$}
        \If{$\mathscr{T}''$ is regular and $\mathscr{T}''\notin\Sigma$}
            \State update $\text{checking}\to\text{checking}\cup\{\mathscr{T}''\}$
        \EndIf
    \EndFor
\EndWhile
\State \textbf{return} $\Sigma$
\end{algorithmic}
\end{algorithm}

Our work concerns fine triangulations, for which \cref{alg:secondary_fan} can be adapted by beginning with a fine $\mathscr{T}(\mathbf{A}, \omega)$ and restricting to fine neighbors $\mathscr{T}'$. This is valid because (fine) regular triangulations are connected by flips (\cite{De_Loera2010-ss}, Thm. 5.3.13). \cref{alg:secondary_fan} is simple, only requiring a subprocess to compute $\mathrm{neighbors}(\mathscr{T})$. I.e., one just needs to compute the flippable circuits of $\mathscr{T}$.

We begin with a conceptually simple, albeit (very) inefficient, procedure for generating the circuits. The final algorithm will only slightly modify this basic procedure. We can enumerate all flippable circuits of $\mathscr{T}(\mathbf{A})$ with the following procedure.
\begin{enumerate}
    \item For every subconfiguration $\mathbf{A}'\subseteq\mathbf{A}$,
    \item check that $\mathbf{A}'$ is minimally dependent.
    \item If this check passes, also check that both $|J_+|>0$ and $|J_-|>0$,
    \item if this check passes, also check that both $\mathscr{T}_\pm(\mathbf{A}')\subseteq\mathscr{T}(\mathbf{A})$ and that $\mathscr{T}_\pm(\mathbf{A}')$ has a constant link in $\mathscr{T}(\mathbf{A})$.
    \item If all checks passed, save $\mathbf{A}'$ as a flippable circuit.
\end{enumerate}
Steps $1$, $2$ have the interpretation of enumerating all circuits of $\mathbf{A}$, step $3$ checks whether each circuit is conceivably flippable (i.e., not an external wall of the secondary fan), and step $4$ checks whether the circuit is embedded in $\mathscr{T}(\mathbf{A})$.

The issue with this procedure is in step $1$ which, naively, involves iterating over $2^{|\mathbf{A}|}$ subsets. This is clearly unnecessary: if $\mathbf{A}'$ is minimally dependent and $\mathbf{A}''\supset\mathbf{A'}$, then $\mathbf{A}''$ is not minimally dependent. Thus, since any subset $\mathbf{A}'$ with $|\mathbf{A'}|=\dim(\mathbf{A})+1$ is dependent, it is unnecessary to check any subset of size $|\mathbf{A'}|>\dim(\mathbf{A})+1$. This effectively motivates the following lemma.

\begin{lem}[\cite{De_Loera2010-ss}, Lemma 4.4.9]
    Let $\mathscr{T}$ be a triangulation of a point configuration that has a flip supported at the circuit $Z$. If all elements of $Z$ are used in $\mathscr{T}$ (that is, ``unless the flip is an insertion flip'') then the flip has some witness wall.
\end{lem}

The circuit $Z$ in this lemma corresponds to a circuit (i.e., minimally dependent subconfiguration) $\mathbf{A}'\subseteq\mathbf{A}$ in our notation. In this lemma, a ``witness wall'' to a flip is simply a $\dim(\mathscr{T})-1$ dimensional simplex $\sigma\in\mathscr{T}$ such that the circuit $\mathbf{A}'\subseteq\text{st}_\mathscr{T}(\sigma)$. That is, the wall $\sigma$ is contained in two maximal simplices $\sigma',\sigma''\in\mathscr{T}$ for which $\mathbf{A}'\subseteq\sigma'\cup\sigma''$.  While the above lemma is stated specifically for point configurations, it also holds for vector configurations. This gives one form of the flippable circuit algorithm, \cref{alg:flippable_circuits}.

\begin{algorithm}[H]
\caption{flippable circuits}
\label{alg:flippable_circuits}
\begin{algorithmic}
\State let $\mathscr{T}(\mathbf{A},\omega)$ be a fine, regular triangulation of a (totally cyclic) vector configuration, $\mathbf{A}$
\State let $C = [\,]$
\For{each $\sigma$ a $\dim(\mathscr{T})-1$ dimensional simplex of $\mathscr{T}$}
    \State let $\mathbf{A}' = \text{st}_\mathscr{T}(\sigma)$ with associated label a $J'$
    \State compute the dependency $\lambda$ such that $\mathbf{A}'\lambda=0$
    \If{either $|J_+|=0$ or $|J_-|=0$}
        \State continue to next iteration
    \ElsIf{both $\mathscr{T}_+(\mathbf{A}')\not\subseteq\mathscr{T}(\mathbf{A})$ and $\mathscr{T}_-(\mathbf{A}')\not\subseteq\mathscr{T}(\mathbf{A})$}
        \State continue to next iteration
    \EndIf
    \State orient $\lambda$ such that $\mathscr{T}_+(\mathbf{A}')\subseteq\mathscr{T}(\mathbf{A})$
    \If{$\text{link}_{\mathscr{T}(\mathbf{A})}(\sigma)$ is constant for all maximal $\sigma\in\mathscr{T}_+(\mathbf{A}')$}
        \State let $\lambda'_j = \begin{cases}
            0,         & j\notin J'\\
            \lambda_j, & j\in J'
        \end{cases}$
        \State append $\lambda'$ to $C$
    \EndIf 
\EndFor
\State \textbf{return} $C$
\end{algorithmic}
\end{algorithm}

\cref{alg:flippable_circuits} is better than iterating over all subsets $\mathbf{A}'\subseteq\mathbf{A}$ since it only loops over
\begin{equation}
    \frac{1}{2}\cdot\dim(\mathbf{A})\cdot\bigg|\setbuilder{\sigma\in\mathscr{T}}{\dim(\sigma)=\dim(\mathscr{T})}\bigg|
\end{equation}
subsets, for a totally cyclic $\mathbf{A}$. I.e., the number of iterations in \cref{alg:flippable_circuits} scales linearly with the number of maximal simplices --- this is much better than the exponential scaling on $|\mathbf{A}|$. Together, \cref{alg:secondary_fan,alg:flippable_circuits} form a fairly efficient algorithm of generating the secondary fan $\Sigma$ of $\mathbf{A}$.

There is one further improvement which we sketch here. For each triangulation $\mathscr{T}$ seen in \cref{alg:secondary_fan}, the neighbors must be computed. That is, naively, \cref{alg:flippable_circuits} must be run on all such $\mathscr{T}$. This is wasteful, especially for $\mathscr{T}$ with a large number of maximal simplices. This is because flips are local modifications on triangulations so, if $\mathscr{T}'$ is a neighbor of $\mathscr{T}$, a large number of witness walls $\sigma\in\mathscr{T}$ will be unchanged and have $\text{st}_\mathscr{T}(\sigma) = \text{st}_\mathscr{T'}(\sigma)$. Such witness walls will define the same circuit, have the same link, and hence be embedded in $\mathscr{T}'$ if and only if they are embedded in $\mathscr{T}$. To avoid such redundant computations it suffices, when flipping, to copy all circuits $\mathbf{A}'$ of $\mathscr{T}$ to $\mathscr{T}'$, other than those for which there exists a simplex $\sigma$ that is being flipped out,
\begin{equation}
    \sigma\in\setbuilder{\sigma'\cup\sigma''}{\sigma'\in\mathscr{T}_+(\mathbf{A}'), \sigma''\in L}.
\end{equation}
This leaves only a small fraction of witness walls to be checked post-flip. 

\section{4D Reflexive Polytopes I: Theoretical Aspects}

\label{sec:ref_poly}

Having spent some time developing the theory of triangulations and toric geometry and studying the classification problem for fine regular triangulations/fans --- culminating in the computational methods of \cref{sec:algo} --- we now apply our newfound knowledge to better understand the toric varieties of the Kreuzer--Skarke database and their CY hypersurfaces, with the goal of constructing new CYs and mapping out more of their K\"ahler moduli spaces. We begin with a recollection of Batyrev's original theory of smooth CY hypersurfaces from FRSTs \cite{Batyrev:1993oya}, as well as some properties of these hypersurface CYs. This will then continue into a more general discussion of the birational geometry of toric varieties and CY hypersurfaces arising from four-dimensional reflexive polytopes: in particular, the role that vex triangulations play. 

For clarity, we note now that because in \cref{sec:pc_vs_vc} we saw how FRSTs induce simplicial fans, we will now proceed with a minor abuse of terminology by letting FRSTs --- formally point configuration triangulations --- refer also to the vector configuration triangulations/simplicial fans they induce.

\subsection{Review of Calabi--Yau Hypersurfaces from FRSTs}
\label{sec:frst_cy}

\subsubsection{Batyrev's Construction}
\label{sec:batyrev}

The anticanonical hypersurface in \textit{any} compact projective toric variety is compact, K\"ahler, and has trivial canonical class, but achieving smoothness is more difficult. In \cite{Batyrev:1993oya} Batyrev identified a set of sufficient conditions for the anticanonical hypersurface in a toric variety to be smooth, as we review now.

Singularities on a hypersurface $X \subset V$ can arise from intersection with singularities of $V$ or failures of the Jacobian criterion (i.e., solutions to $F = dF = 0$ for $F = 0$ the hypersurface equation). Bertini's theorem states that the Jacobian criterion is satisfied away from the basepoint locus of the class of the hypersurface in question. A sufficient condition for achieving the criterion everywhere in our case, then, is to impose that the anticanonical class is basepoint free. This is true for ample divisors. A smooth hypersurface contained in the smooth locus of a normal variety must also be Cartier\footnote{A smooth subvariety of a smooth variety is a regular embedding, and codimension-one regular embeddings are Cartier divisors. Thus, a smooth hypersurface $Y$ contained in the smooth locus of a variety is a Cartier divisor on that smooth locus. For normal varieties, whose singular loci are in codimension $\geq 2$, a Cartier divisor on the smooth locus induces a Cartier divisor on the entire variety. \label{foot:cartier_and_smooth}} so we are motivated to consider Gorenstein Fano toric varieties (\cite{cls}, Def. 8.3.1). Famously, the fan of such a toric variety must be the normal fan $\Sigma_\Delta$ of the polar dual $\Delta \subset M$ of the convex hull $\Delta^\circ \subset N$ of its minimal generators $u_\rho$. That is, $\Delta$ must be the anticanonical Newton polytope (ampleness) and $\Delta$ must be lattice (Cartier). Thus the rays of the fan $\Sigma_\Delta$ are exactly those generated by the vertices of a reflexive polytope $\Delta^\circ \subset N$ dual to $\Delta$, with higher-dimensional cones given by the faces of $\Delta^\circ$. That is, $\Sigma_\Delta$ is the central fan of $\Delta^\circ$.

This toric variety has at worst canonical singularities (\cite{cls}, Prop. 11.4.11) but generically does have many such singularities. Indeed, the associated fan is typically non-simplicial and the convex hulls $\Pi_\sigma$ over minimal generators of maximal cones $\sigma$ typically contain interior lattice points.
This is handled by blowing up all of the associated canonical singularities, adding in rays generated by these interior points (corresponding to new, exceptional prime toric divisors) and subdividing the cones of $\Sigma_\Delta$ to arrive at a simplicial fan $\Sigma'$.\footnote{A standard result is that prime toric divisors corresponding to points interior to facets do not intersect the generic anticanonical hypersurface: the typical argument exploits that the fan is a refinement of the central fan but this argument can be generalized to all fine fans. A consequence is that the canonical singularities that exist in lieu of such blowups do not affect generic CY hypersurfaces: thus, we do not need to consider points interior to facets in general for fine fans. See \cref{sec:intersect} for further discussion.} By definition, a resolution of (strictly) canonical singularities is \textit{crepant}, or preserves the canonical class, so the pullback of the anticanonical hypersurface in $V_{\Sigma_\Delta}$ remains the anticanonical hypersurface in $V_{\Sigma'}$.\footnote{This is essentially a consequence of the fact that adding interior points of $\Delta^\circ$ doesn't change the Newton polytope of the anticanonical class: i.e., $\Delta_{\overline{K}}$ is the same on $\Sigma_\Delta$ and $\Sigma'$. This follows from the result we stated earlier that the anticanonical Newton polytope depends only on the vertices of $\Delta^\circ$. If we had introduced a new vertex, the anticanonical Newton polytope would shrink, and correspondingly the pullback of the anticanonical class of $\Sigma_\Delta$ would differ from that of $\Sigma'$ by a non-trivial \textit{discrepancy}, or a linear combination of the exceptional blown-up prime toric divisor(s).} The resulting Gorenstein simplicial toric variety has at most terminal singularities (\cite{cls}, Prop. 11.4.12) so its singular locus is at codimension $\geq 4$ (\cite{cls}, Prop. 11.4.22) and is therefore at worst point-like on a four-dimensional toric variety, which will not intersect the generic anticanonical hypersurface. On the resulting toric variety, the anticanonical class will in general only be nef, but this is still enough for the class to be basepoint free (\cite{cls}, Th. 6.3.12), so the Jacobian criterion continues to hold. Thus, the generic anticanonical hypersurface is smooth.

\subsubsection{Interpretation with Triangulations}

We've just described a somewhat abstract algebraic geometric process for constructing smooth Calabi--Yau threefolds, but there is nice intuition to be gained by thinking about this process from the perspective of triangulations and the secondary fan. In particular, this will make the role of fine, regular, star (point configuration) triangulations clear. 

Consider the vector configuration $\mathbf{A}$ associated to a reflexive polytope $\Delta^\circ$. We would like to consider a subdivision of this configuration whose associated fan is the central fan. This choice is actually usually non-unique, because subdivisions and fans are not quite in bijection given our definitions (as discussed in \cref{sec:triang}), unlike the triangulations and simplicial fans (we discuss this in detail in \cref{sec:translate}). Nevertheless, there is a canonical choice of subdivision which we will call the \textit{central subdivision} of $\Delta^\circ$: namely, the regular subdivision $\mathscr{S}(\mathbf{A}, 1) := \mathscr{S}(\mathbf{A}, (1, \dots, 1))$. The fan associated to this subdivision is indeed the central fan.\footnote{In particular, we can identify the height vector $(1, \dots, 1)$ with the torus-invariant divisor $D = \sum_\rho 1 \cdot D_\rho$, which belongs to the anticanonical class $\overline{K}$ of the associated toric variety, so the normal fan of its Newton polytope is the central fan.} It's important to note that while the cells of $\mathscr{S}(\mathbf{A}, 1)$ are in bijection with the faces of $\Delta^\circ$, the cells $F$ associated to a facet $\Theta$ consist of \textit{all} lattice points on $\Theta$, including non-vertex lattice points of $\Theta$. That is, $\mathscr{S}(\mathbf{A}, 1)$ is a fine subdivision, but the non-vertex lattice points do not correspond to one-cones.

With this in hand, recall the non-unique process from the previous section of starting with the central fan $\Sigma_\Delta$ and ending with a simplicial fan $\Sigma'$. We have now associated $\Sigma_\Delta$ to the central subdivision $\mathscr{S}(\mathbf{A}, 1)$, and $\Sigma'$ corresponds to a unique fine regular triangulation $\mathscr{T}(\mathbf{A}, \omega)$ for some height vector $\omega$. In this language, the process of passing from $\mathscr{S}(\mathbf{A}, 1)$ to $\mathscr{T}(\mathbf{A}, \omega)$ is \textit{refinement}, as discussed in \cref{sec:regular}, because every cell in the latter is contained in a cell of the former by construction. Recall that, for a regular subdivision, a height vector on the boundary of its secondary cone gives rise to a coarsening of the original subdivision. Thus, since $\mathscr{T}(\mathbf{A}, \omega)$ (maximally) refines $\mathscr{S}(\mathbf{A}, 1)$, the (solid) secondary cone $\mathbf{C}(\mathbf{A}, \mathscr{T}(\mathbf{A}, \omega))$ contains the height vector $(1, \dots, 1)$ on its boundary. From \cref{sec:pc_vs_vc}, $\mathscr{T}(\mathbf{A}, \omega)$ induces a fine regular star triangulation (FRST) of the point configuration associated to $\Delta^\circ$ (and we recall that we will abuse notation by also calling $\Sigma'$ an FRST). FRSTs are then exactly the fine regular triangulations of $\mathbf{A}$ whose secondary cones contain $(1, \dots, 1)$. We note that though $\mathscr{S}(\mathbf{A}, 1)$ and $\mathscr{T}(\mathbf{A}, \omega)$ are both fine, the latter has more one-cones, as only in the latter do the non-vertex lattice points of $\Delta^\circ$ correspond to one-cones: the associated prime toric divisors are exactly the exceptional divisors resulting from blowing up the canonical singularities from the original singular toric variety associated to $\mathscr{S}(\mathbf{A}, 1)$. 

This triangulation theoretic discussion enables an algebraic geometric observation. A variety is \textit{weak-Fano} if the anticanonical class is nef and big. A divisor on a toric variety is big if its Newton polytope is of maximal dimension (see \S9.3 in \cite{cls}), which will always be true of the anticanonical Newton polytope. We've thus seen that toric varieties are associated to FRSTs if and only if they are weak-Fano. In particular, they are Fano if and only if $\Delta^\circ$ admitted a unique FRST (i.e., the central subdivision actually was a triangulation), which in turn holds if and only if the only non-zero lattice points of $\Delta^\circ$ were vertices and each of its $k$-faces has exactly $k+1$ vertices.

For a concrete illustration of this discussion of the secondary fan, we refer the reader to \cref{ex:192} and the accompanying plot in \cref{fig:192}, where a slice of the chamber fan illustrates how the ray generated by the class $\overline{K}$ (black dot), the image of the vector $(1,\dots,1)$ in the chamber fan, intersects five secondary cones, corresponding to the five FRSTs which refine the cells of the central fan.

\subsubsection{Calabi--Yau Hypersurface Properties}

\label{sec:frst_cy_prop}

We now briefly review some standard properties of Calabi--Yau hypersurfaces arising from FRSTs.

First, as introduced in \cref{ex:intro}, the diffeomorphism class of a CY is determined by its Wall data, or its two non-trivial Hodge numbers $h^{1,1}$ and $h^{2,1}$, its intersection form, and its second Chern class \cite{wall_classification_1966}. In the FRST case, CY hypersurface Hodge numbers can be computed, for example, using the stratification methods of Danilov and Khovanskii \cite{danilov_newton_1987} (summarized in our case by Th. 4.3.7 in \cite{Batyrev:1993oya}). The CY intersection numbers can be computed by contracting those of the ambient space with the anticanonical class, and the second Chern class can similarly be computed from the total Chern class of the ambient space by adjunction.

Having briefly touched on topology, let us now turn our attention to the (birational) geometry of CY hypersurfaces from FRSTs, elaborating on our discussion from the introduction. Facets of general CY threefold K\"ahler cones have been classified \cite{wilson1992kahler, Witten:1996qb}: either a disjoint union of rigid $\mathbb{P}^1$s belonging to a common curve class $C$ contracts (i.e., shrinks, or is blown down) to singular points, or effective divisor(s) are contracted. The former case corresponds to a flop transition, as the singular points can be blown up into the curve class $-C$, yielding a naively distinct CY topology $X'$ with its own K\"ahler cone $K(X')$. In particular, $K(X)$ and $K(X')$ share a facet consisting of K\"ahler forms $J$ satisfying $J \cdot C = 0$. In the latter case, there exists no such continuation of the geometric moduli space. 

We define a birational equivalence class of CYs to be the set of CYs related by flops,\footnote{This is consistent with our convention in \cref{sec:toric_bir} that birational maps be isomorphisms in codimension one.} the CY analog to a flip.\footnote{We now take a moment to clarify what's going on with the names ``flip'' and ``flop''. Both refer to a birational map achieved by contracting a curve whose class $C$ is an extremal ray of the Mori cone: that is, an extremal contraction. The difference between the two is that flips $X \dashrightarrow Y$ technically satisfy $K_X \cdot C \leq 0$ and $K_Y \cdot C \geq 0$ while flops are defined to be the special case where $K_X \cdot C = K_Y \cdot C = 0$. We join much of the literature in abusing terminology by referring to both a flip and its inverse as flips (technically one of the directions will not be a flip, but rather an ``anti-flip''). For a Calabi--Yau variety, $K_X = K_Y = 0$ so such an extremal contraction is always a flop, while for compact toric varieties it is always just a flip and never a flop.} The union of their K\"ahler cones forms a fan, the extended K\"ahler cone $\mathcal{K}(X)$. This is analogous to how the toric moving cone $\mathrm{Mov}(V_\Sigma)$ decomposes as the disjoint union of toric K\"ahler cones $\Gamma(V_\Sigma)$: indeed, $\mathcal{K}(X)$ is the moving cone of $X$. 

These decompositions are, in fact, related. There is an inclusion $H^2(V_\Sigma) \hookrightarrow H^2(X)$ of cohomology given by restriction, which descends to an inclusion of K\"ahler cones because K\"ahler forms on $V_\Sigma$ yielding non-negative holomorphic cycle volumes will, in particular, yield non-negative holomorphic cycle volumes on $X$. Thus, the K\"ahler cone $\Gamma(V_\Sigma)$ is a subcone of the K\"ahler cone $K(X)$ of $X$, bounding $K(X)$ from within. We say that $\Gamma(V_\Sigma)$ is an \textit{inner approximation} of $K(X)$, and we will continue to use this term to describe such inclusions. Indeed, there are facets of $\Gamma(V_\Sigma)$ whose K\"ahler forms assign zero volume to holomorphic cycles on $V_\Sigma$ which do not intersect the CY hypersurface: these facets do not descend to $K(X)$. In this way, some flips of $V_\Sigma$ restrict to flops of $X$ (if a curve shrinks on both $V_\Sigma$ and $X$) and some do not (if a curve only shrinks on $V_\Sigma$). 

In particular, there is a well-understood criterion for a flip between FRSTs to descend to a flop of CYs. The Wall data of a Calabi--Yau hypersurface arising from an FRST depends only on the restriction of an FRST to the two-faces of the polytope $\Delta^\circ \subset N$ \cite{Demirtas:2020dbm}: we say two FRSTs with the same two-face triangulations are two-face equivalent. Moreover, a flip descends to a flop if and only if it affects the two-face triangulations. That is, the flip contracts a toric subvariety that intersects the CY hypersurface only if it affects the two-face triangulations (i.e., the associated cone has minface dimension $\leq 2$, rather than $3$: see \cref{sec:intersect} for further discussion).\footnote{This actually provides a different way to see that two-face equivalent triangulations yield diffeomorphic Calabi--Yau hypersurfaces, without Wall's theorem: the sequence of two-face-preserving flips which relate the two triangulations induces a smooth homotopy of the Calabi--Yau hypersurfaces.}

This naturally leads to the toric inner approximation $K_\cup(X)$ \cite{Demirtas:2018akl,cox1999mirror} (``K-cup'') for CY hypersurface K\"ahler cones. Let $\mathscr{V}(X)$ denote the set of toric varieties arising from FRSTs related to $V_\Sigma$ by flips which do not descend to flops (do not affect two-faces): then
\begin{equation}
    \Gamma(V_\Sigma) \; \subseteq \; K_\cup(X) = \bigcup_{V_{\Sigma'} \in \mathscr{V}(X)} \Gamma(V_{\Sigma'}) \; \subseteq \; K(X).
\end{equation}
The lift of this cone to the space of torus-invariant divisors/heights --- i.e., the unions of secondary cones which triangulate the 2-skeleton of the polytope $\Delta^\circ$ in the same way --- was studied by one of the authors in \cite{macfadden2023efficient}, where an efficient algorithm for generating such unions of secondary cones was presented.

We analogously get an inner approximation $\mathcal{K}_\cup(X)$ of the extended K\"ahler cone.
\begin{equation}
    \mathcal{K}_\cup(X) := \bigcup_{\Sigma' \text{ from an FRST}} \Gamma(V_{\Sigma'}) \subset \mathcal{K}(X).
\end{equation}
This is how the combinatorics of toric geometry serves to map out the K\"ahler moduli space of CY hypersurfaces. 

\subsection{FRSTs vs. Vex Triangulations}
\label{sec:frst_vs_vex}

From the previous section we see how FRSTs are a very natural class of fans to consider for the purpose of generating smooth CY hypersurfaces. However, one is led to wonder about the fans which do not arise from FRSTs. We follow earlier references which studied such fans for the purposes of CY geometry and string theory \cite{Berglund:2016yqo, Berglund:2016nvh, Berglund:2022dgb, Berglund:2022zto, Berglund:2024zuz, Huang:2019pne, Jefferson:2022ssj, vex_notes}, and label non-FRST fans \textit{vex triangulations}.\footnote{One technicality ought to be noted: vex triangulations are defined such that their associated toric varieties must be Gorenstein (equivalently, for all maximal cones with rays $u_1, \dots, u_n$, the affine hyperplane containing each $u_i$ must have a lattice distance of $1$ from the origin in the sense of Def. 4.1.4 in \cite{cls}). We thank Tristan Hubsch for clarifying this for us. This property holds for any toric variety with smooth anticanonical hypersurface avoiding toric singularities (recalling \cref{foot:cartier_and_smooth}, such smooth divisors are Cartier, and Gorenstein means Cartier anticanonical class), so the Gorenstein property is a corollary of our main result --- \cref{main} --- that vex triangulations yield smooth CY hypersurfaces. We discuss this further in \cref{sec:sing_locus}. However, in general, a non-FRST fine regular simplicial fan of a polytope need not be a vex triangulation.} In this work, whenever we employ this term, it is implied that we restrict to fine, regular, simplicial fans. We know that such fans exist: indeed, we saw one in \cref{ex:intro}. We also now understand how to systematically generate FRSTs and vex triangulations alike, thanks to the theory we developed in \cref{sec:triang}, \cref{sec:toric}, and \cref{sec:sec_fan}, culminating in the algorithms of \cref{sec:algo}. The time has thus now come to understand the toric geometry of vex triangulations of reflexive polytopes and their anticanonical hypersurfaces. It is worth beginning by carefully understanding what distinguishes FRSTs and vex triangulations --- and in which ways they are essentially the same. Among other things, this will clarify which of the assumptions one often makes when working with FRSTs must be let go when transitioning to vex triangulations. 

First, though, we should stress that the general techniques and results of triangulation theory and toric geometry are largely agnostic to the FRST-vex distinction: e.g., we never needed to make reference to this distinction in \cref{sec:triang} or \cref{sec:toric}. Indeed, the methods for computing topological data such as the class group, intersection numbers, Chern classes, etc. which we illustrated in \cref{sec:toric} hold for all toric varieties. It is really only when one is interested in their Calabi--Yau hypersurfaces that one is encouraged to restrict to FRSTs. Nevertheless, though FRSTs are especially nice for constructing smooth Calabi--Yau threefolds, we will argue that vex triangulations are just as amenable to this purpose.

There are several simple criteria that distinguish FRSTs from vex triangulations, each with their own utility. We stress that every condition below is an if-and-only-if. At the end of the section, in \cref{ex:baby_vex}, we illustrate these criteria in a simple three-dimensional toy model.
\begin{itemize}
    \item \textbf{FRSTs refine the central subdivision.}
    \item \textbf{FRST secondary cones contain $(1, \dots, 1)$.}
    \item \textbf{FRSTs yield weak-Fano toric varieties.} The height vector $(1,\dots,1)$ can be interpreted as the torus-invariant divisor $\sum_\rho 1 \cdot D_\rho$, which we know belongs to the anticanonical class $\overline{K}$. Thus, by projecting from the secondary fan down to the chamber complex --- sending torus-invariant divisors to their divisor classes --- the secondary cone maps to the nef cone and contains $\overline{K}$. A variety is Fano if $\overline{K}$ is ample, and weak-Fano if $\overline{K}$ is nef, so the central fan gives rise to a Fano toric variety and all of its refinements give rise to (at least) weak-Fano toric varieties. 
    \item \textbf{FRST fans yield basepoint free $\overline{K}$.} This is an immediate corollary of the previous bullet, as divisors on compact toric varieties are basepoint free if and only if they are nef (\cite{cls}, Th. 15.1.1).
    \item \textbf{FRSTs are \textit{star}.} This is obvious from the name, but worth stressing: vex triangulations of reflexive polytopes, on the other hand, do induce triangulations of the point configuration associated to the polytope, but these are non-star, as discussed in \cref{sec:pc_vs_vc}. Similar to how we let the term ``FRST'' denote triangulations of both the associated point and vector configurations, one can think of vex triangulations either as fans or their induced non-star point configuration triangulations.
    \item \textbf{The primitive/minimal generators of any cone in an FRST are all contained in a common facet.} The only obstruction to a fine regular fan being a refinement of the central fan (or defining a star triangulation of the point configuration) is the presence of a cone whose primitive generators are not all contained in a common facet. This is not an immediate corollary of our discussion, but follows directly from Theorem 4.5.9 in \cite{De_Loera2010-ss}.
    Additionally, we will find in \cref{sec:basepoint} that the basepoint loci of $\overline{K}$ are exactly the toric subvarieties of basepoint cones. This is a toric geometric reason that a fine regular fan is not an FRST if and only if it features cones not contained in facets.
\end{itemize}

A couple of comments are in order. First, regarding point configuration triangulations (i.e., the second-to-last bullet above), let us make contact with a concept introduced in \cite{Berglund:2016nvh}: namely, \textit{VEX polytopes}. While vex triangulations induce non-star triangulations of reflexive polytopes, they induce star triangulations of a particular VEX polytope. In particular, recalling \cref{eq:simplex_poly}, to any fan $\Sigma$ constructed from an $n$-dimensional reflexive polytope $\Delta^\circ$ one can associate the union 
\begin{equation}
    \nabla = \bigcup_{\sigma \in \Sigma(n)} \Pi_\sigma.
\end{equation}
If $\Sigma$ was an FRST, $\nabla = \Delta^\circ$ because the $\Pi_\sigma$ fully triangulate $\Delta^\circ$. However, if $\Sigma$ is a vex triangulation, $\nabla$ is $\Delta^\circ$ with the non-star simplices of the associated point configuration triangulation excised. This non-convex union of simplices is what is called the VEX polytope \cite{Berglund:2016nvh}, and $\Sigma$ does triangulate it with the simplices $\Pi_\sigma$. In particular, $0$ belongs to each $\Pi_\sigma$, so we say the vex triangulation $\Sigma$ star-triangulates the VEX polytope $\nabla$. We will not employ VEX polytopes further.

Second, on the topic of point configuration triangulations, recall that in \cref{sec:pc_vs_vc} we introduced a map between point configuration triangulations and vector configuration triangulations of the homogenization. It's worth briefly understanding this from a toric geometry/physics perspective. Interpreting the fan of a toric variety $V$ as a point configuration triangulation --- possibly non-star --- one can consider the toric variety $V'$ associated to the homogenization. One can show that $V'$ is the total space of the canonical bundle of $V$, a non-compact toric Calabi--Yau fourfold and a natural environment for two-dimensional $N = 2$ supersymmetric theories (see, e.g., \S4 of \cite{Witten:1993yc}). The non-complete fan has a complete secondary fan/chamber fan: in particular, the charge matrix for $V'$ involves appending the non-effective class $K_V$ to the charge matrix for $V$. Thus, all of the original secondary cones from the secondary fan of $V$ exist for $V'$, but now there are additional secondary cones: in certain examples, these novel cones have been shown to correspond to non-geometric phases of compactifications of string theory on the CY hypersurfaces associated to the toric varieties of the original FRSTs/vex triangulations.

To conclude this section, we present a toy model exhibiting the FRST-vex dichotomy. We will consider an acyclic triangulation/non-complete fan (i.e., a non-compact toric variety) in three dimensions, for conceptual simplicity, but one could imagine this setup being embedded inside of a complete three-dimensional fan.
\begin{example}
    Consider the following vector configuration.
    \begin{equation}
        \mathbf{A} = \begin{pNiceMatrix}[first-row]
            1 & 2 & 3 & 4 \\
            0 & 1 & 0 & 1 \\
            0 & 0 & 1 & 1 \\
            1 & 2 & 2 & 2
        \end{pNiceMatrix}.
    \end{equation}
    See \cref{fig:toy}. We will denote the columns of this matrix by $u_i$. We can choose the following charge matrix $Q$.
    \begin{equation}
        Q = \begin{pNiceMatrix}[first-row]
            1 & 2 & 3 & 4 \\
            2 & -1 & -1 & 1
        \end{pNiceMatrix}.
    \end{equation}
    One way to understand the FRSTs of this vector configuration is to identify which fans induce full triangulations of the associated point configuration --- the three dimensional polytope $\Delta^\circ$ given as the convex hull of the $u_i$ and the origin, $0$ --- and cover all faces with strictly star simplices. $\Delta^\circ$ has six facets (codimension-one faces): four of them contain the origin, while two do not. 
    
    The four facets containing the origin will not be interesting to us here, which we can see in at least two ways. First, by the definition of a subdivision, the support of any fan must be the convex hull of these rays, so the four origin-containing facets must be facets of cones in our fan. Alternatively, as stated before the beginning of this example, one can imagine embedding this configuration inside of an acyclic/complete configuration, in which case these facets would no longer be relevant. The two important facets are thus $\conv{u_1,u_2,u_3}$ and $\conv{u_2,u_3,u_4}$.
    
    This vector configuration admits two fine triangulations/fans: an FRST $\Sigma_1$ which respects these two faces, and a vex triangulation $\Sigma_2$ which does not. 
    
    Let us begin with $\Sigma_1$. Its maximal cones $\sigma_{11}, \sigma_{12}$ are just the cones over the two facets: $\sigma_{11} = \cone{u_1,u_2,u_3}$ and $\sigma_{12} = \cone{u_2,u_3,u_4}$. In particular, $\Delta^\circ = \Pi_{\sigma_{11}} \cup \Pi_{\sigma_{12}}$, so $\Sigma_1$ induces a star point configuration triangulation of $\Delta^\circ$. $\Sigma$ results from a height vector $\omega = (1,1,1,1)$: it was the central fan and required no further refinement. In this basis, projecting to the chamber fan reveals that the nef/K\"ahler cone for $V_{\Sigma_1}$ is $\mathbb{R}_{\geq 0}$, and the anticanonical class is $1$, so $V_{\Sigma_1}$ is indeed weak-Fano (in this case, it is in fact Fano). 
    
    Now we turn our attention to $\Sigma_2$. Its maximal cones are $\sigma_{21} = \cone{u_1,u_2,u_4}$ and $\sigma_{22} = \cone{u_1,u_3,u_4}$. Now, crucially, $\Delta^\circ \neq \Pi_{\sigma_{21}} \cup \Pi_{\sigma_{22}}$. In particular, their difference is the ``divot'' $\Pi' = \conv{u_1,u_2,u_3,u_4}$, a non-star simplex, so the induced point configuration triangulation must include this cell, yielding a non-star point configuration triangulation. This is the consequence of both maximal cones not being contained in a single facet. The nef/K\"ahler cone for $V_{\Sigma_2}$ is $\mathbb{R}_{\leq 0}$, so the associated toric variety is not weak-Fano.

    In this case, the VEX polytope that $\Sigma_2$ star-triangulates is $\Pi_{\sigma_{21}} \cup \Pi_{\sigma_{22}}$ or more explicitly $\conv{0,u_1,u_2,u_4} \cup \conv{0,u_1,u_3,u_4} = \conv{0,u_1,u_2,u_3,u_4} \setminus \conv{u_1,u_2,u_3,u_4}$. \label{ex:baby_vex}
\end{example}

\begin{figure}
    \centering
    \begin{minipage}{0.48\linewidth}
        \centering
        \includegraphics[width=\linewidth]{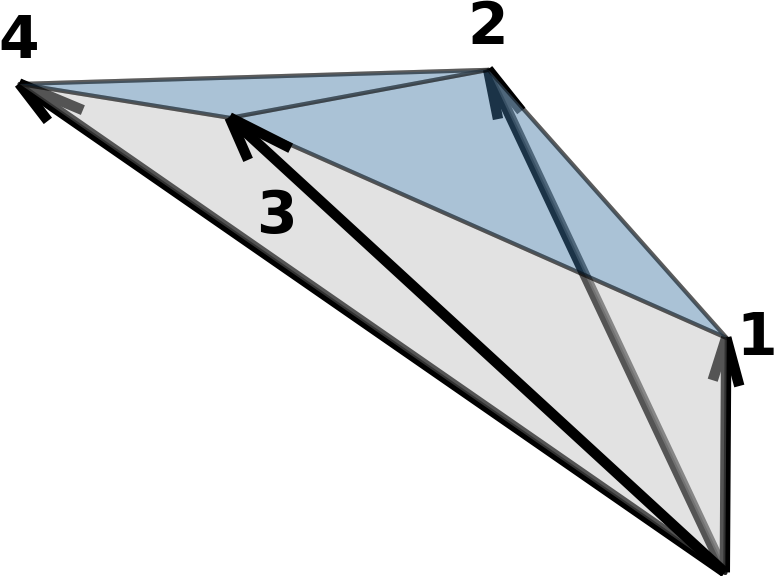}
    \end{minipage}
    \hfill
    \begin{minipage}{0.48\linewidth}
        \centering
        \includegraphics[width=\linewidth]{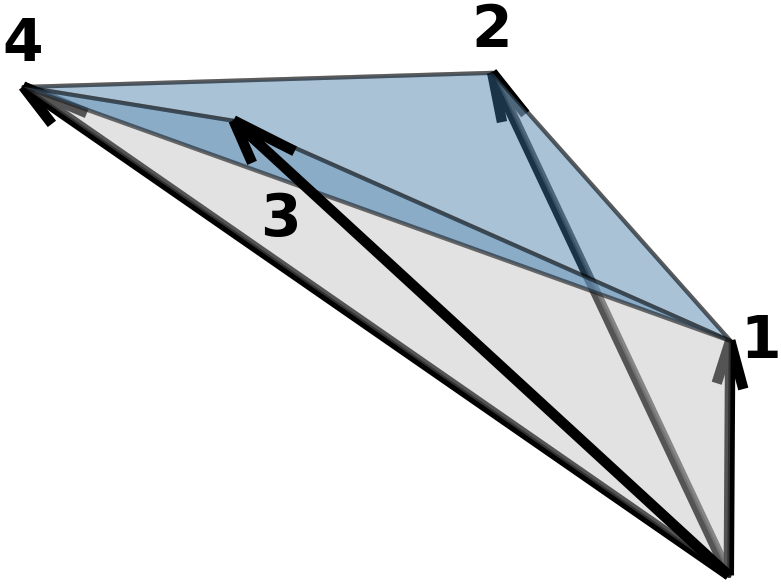}
    \end{minipage}
    \caption{The triangulations $\Sigma_1$ (left) and $\Sigma_2$ (right) of the vector configuration $\mathbf{A}$ from \Cref{ex:baby_vex}, which serve as toy models for FRSTs and vex triangulations, respectively. The polytopes $\Pi_{\sigma_{ij}}$ associated to the maximal cones are emphasized (their unique facet not containing the origin is colored in blue). For $\Sigma_1$, the union $\Pi_{\sigma_{11}} \cup \Pi_{\sigma_{12}}$ agrees with the convex hull $\Delta^\circ$ of the rays together with the origin, while for $\Sigma_2$ the union $\Pi_{\sigma_{21}} \cup \Pi_{\sigma_{22}}$ is strictly smaller than $\Delta^\circ$ by a divot --- this non-convexity is the characteristic property of vex triangulations.}
    \label{fig:toy}
\end{figure}

\subsection{Calabi--Yau Hypersurfaces from Vex Triangulations}
\label{sec:cy_hyper}

In this section, we study the topology and geometry of CY hypersurfaces arising from vex triangulations. We are immediately led to ask what changes from \cref{sec:frst_cy_prop} when considering vex triangulations. The intersection form and second Chern class can be computed for CY hypersurfaces exactly as they are for FRSTs, using the analogous data of the ambient non-weak-Fano toric variety. The methods of \cite{danilov_newton_1987} for computing Hodge numbers do not immediately apply, as they traditionally require that the hypersurface is basepoint-free. Moreover, it is \textit{a priori} unclear if the CY hypersurfaces are smooth. Additionally, if they are smooth, it is unclear when the vex triangulations of a fixed polytope are birationally related to the other triangulations (FRST or vex) of the same polytope.\footnote{For example, all fine regular simplicial fans of a four-dimensional reflexive polytope could give rise to smooth CY hypersurfaces, but a flip between two toric varieties could descend to the contraction of a divisor on the two CY hypersurfaces, meaning the two CYs were not necessarily birational. However, \cref{main} is the result that this does not happen.} In this section we will study the potential singularities of CY hypersurfaces in vex triangulations, culminating in the following result, which we prove in \cref{sec:smooth_birational}. 
\begin{prop}
    \label{main}
    (Fine, regular) fans constructed from a fixed four-dimensional reflexive polytope yield Gorenstein toric varieties with smooth Calabi--Yau hypersurfaces; moreover, all such CY hypersurfaces belong to the same birational equivalence class.
\end{prop}
Thus, one can freely construct fine, regular fans --- FRST or vex --- of polytopes in the Kreuzer--Skarke database in order to generate examples of smooth Calabi--Yau threefolds. A corollary of this is that Danilov and Khovanskii's familiar formula actually does hold, in spite of the non-trivial basepoint loci of the anticanonical class for vex triangulations, because Hodge numbers of Calabi--Yau threefolds (at generic complex structure) are a birational invariant.\footnote{It would be nice to adapt the methods of \cite{danilov_newton_1987} to accommodate basepoint loci and demonstrate that one recovers the same formula in general: we defer this to future work.}

Let us continue our discussion of how vex triangulations differ from their FRST counterparts. Analogous to FRSTs, we can approximate the K\"ahler cone of CY hypersurfaces from vex triangulations $\Sigma$ naively, using the toric K\"ahler cone $\Gamma_\Sigma$, and in an improved way using $\mathcal{K}_\cup(X)$: one merely lets $\mathscr{V}(X)$ denote the set of toric varieties arising from some fine regular simplicial fan related to $\Sigma$ by flips not descending to flops. This is slightly harder to compute for vex triangulations, as it no longer suffices to look at two-face triangulations: indeed, the restrictions of a vex triangulation to two-faces are not even triangulations, but subdivisions, due to the cones extending between faces. However, one can just check whether the toric subvarieties associated to the cones in the coarsening associated with the secondary cone facet intersect the CY hypersurface, which we explain how to do in \cref{sec:intersect}.

Importantly, a direct consequence of \cref{main} is that all toric K\"ahler cones across all fine regular simplicial fans correspond to CY hypersurfaces in the same birational class.
\begin{prop}
    \label{moving}
    The extended K\"ahler cone of the birational equivalence class of CY hypersurfaces associated to the toric varieties arising from fine regular simplicial fans of a four-dimensional reflexive polytope admits an inner approximation given by the moving cone of the ambient toric varieties.
\end{prop}
\begin{proof}
    An immediate corollary of \cref{main} is that the K\"ahler cones of all CY hypersurfaces associated to the toric varieties of fine regular simplicial fans of four-dimensional reflexive polytopes belong to the same extended K\"ahler cone. Each individual CY K\"ahler cone is approximated from the inside by the K\"ahler cone of its ambient toric variety, so the union of these toric K\"ahler cones is contained in the extended K\"ahler cone. This union is exactly the moving cone, by \cref{eq:movable}.
\end{proof}

Before continuing on to our discussion of singularities, we stress that while the toric moving cone is a good inner approximation to the extended K\"ahler cone, it remains an approximation: in general, the entire CY K\"ahler moduli space can be larger. At walls of the moving cone, divisors are contracted (blown down) on the toric variety, but this need not hold on the Calabi--Yau hypersurface. That is, while the toric variety changes topologically in codimension 1, the Calabi--Yau hypersurface doesn't have to: it may undergo a flop that truly does not descend from a flip of the ambient toric variety. We present an example featuring this phenomenon in \cref{ex:38}. What's more, the Calabi--Yau may actually even not become singular at all at the boundary of the toric moving cone. We defer a systematic study of faces of the toric moving cone to future work.

\subsubsection{Generalities of Singularities}

We will present two strategies for assessing whether CY hypersurfaces $X$ arising from vex triangulations $\Sigma$ are smooth. First, to prove \cref{main}, we will exploit the birational invariance of smoothness for CY threefolds: flops of smooth CY threefolds are smooth. This allows us to essentially ``extrapolate'' the smoothness of CYs from FRSTs to CYs from vex triangulations.\footnote{We are grateful to Jakob Moritz for proposing this method for determining smoothness to us.} Second, we will carry out a complementary analysis by directly studying the singular locus of $X$, characterizing when $X$ intersects singularities of the ambient variety and when the Jacobian criterion can fail, culminating in a set of sufficient conditions for a particular CY hypersurface to be smooth. We will not prove that these conditions always hold for vex triangulations, but they can be readily checked in examples and we will demonstrate in \cref{sec:count} by brute force that these conditions are not violated for $h^{1,1} \leq 7$. While the first strategy suffices to prove smoothness for vex triangulations of reflexive four-dimensional polytopes, it hinges on \textit{all} flips of triangulations of reflexive polytopes having certain properties, while the second strategy can be applied locally, to individual toric varieties. In particular, we believe the second strategy would generalize better to the study of smoothness of Calabi--Yau subvarieties of toric varieties more generally (e.g., at higher codimension or in non-reflexive canonical polytopes).

As it turns out, both strategies will depend crucially on the basepoint loci of the anticanonical class. For the former strategy, this will help us understand the subvarieties that can be contracted on the CY hypersurface by toric flips; for the latter, it controls when singular toric subvarieties can be contained in the generic anticanonical hypersurface and where the Jacobian criterion can fail. Thus, we begin with an analysis of this basepoint locus: from there, we will proceed by discussing our two strategies for probing smoothness in succession. Along the way we will discuss the intersection of toric subvarieties with the CY hypersurface as well as how birational geometry of the ambient toric variety descends to the CY hypersurface.

\subsubsection{Anticanonical Basepoint Loci}

\label{sec:basepoint}

Recall that for a vex triangulation $\Sigma$ of a four-dimensional reflexive polytope $\Delta^\circ \subset N$, the anticanonical class is not basepoint free: we now study its basepoint locus $\mathrm{Basepoint} \, \overline{K}$. This is the union of toric subvarieties whose equations are automatically satisfied when the generic anticanonical section vanishes. It must be the case that when the equation $x_1 = \dots = x_k = 0$ of any such toric subvariety is satisfied, the generic anticanonical section vanishes. Recalling our discussion of dual faces of reflexive polytopes in \cref{sec:line_bundle_newton}, this is equivalent to requiring that the cone associated to the toric subvariety has minface dimension $4$ --- i.e., the minface is the entire polytope --- so that its dual face has dimension $-1$ and is therefore empty. 

This motivates us to define a \textit{basepoint cone} of $\Sigma$ to be a cone $\sigma \in \Sigma$ such that $\sigma$ is not contained in any face of $\Delta^\circ$, and every face of $\sigma$ is contained in some face of $\Delta^\circ$. Any toric subvariety belonging to the basepoint locus of $\overline{K}$ will be contained in the toric subvariety of a basepoint cone. In particular, we have that
\begin{equation}
    \mathrm{Basepoint} \, \overline{K} = \bigcup_{\substack{\sigma \in \Sigma \\ \sigma \text{ is basepoint}}} V_{\Sigma,\sigma}
\end{equation}
Given some such $\sigma$, setting $x_\rho = 0$ for each $\rho \in \sigma$ causes each monomial section of the anticanonical bundle to vanish, so the generic anticanonical global section $f$ vanishes on $V_{\Sigma,\sigma}$. That is, $f$ factors as 
\begin{equation}
    \label{eq:factor}
    f = \sum_{\rho \in \sigma} x_\rho f_\rho
\end{equation}
where the $f_\rho$ are generic global sections of $\overline{K} - [D_\rho]$. FRSTs can have no basepoint cones cutting between facets, as they respect the face structure of $\Delta^\circ$ by construction; vex triangulations must have basepoint cones.

We now characterize the dimensions of basepoint cones. Our primary tool will be the following lemma.
\begin{lem}[\cite{cls}, Lemma 8.3.6]
    Two points on the boundary of a reflexive polytope $\Delta^\circ$ either share a common face or sum to a lattice point in $\Delta^\circ$. 
    \label{lem:cls}
\end{lem}
This can be used to argue that basepoint cones on reflexive polytopes cannot be ``too small.''
\begin{prop}
    A fine fan $\Sigma$ constructed from a reflexive polytope cannot have basepoint two-cones.
\end{prop}
\begin{proof}
    Assume by way of contradiction $\Sigma$ has such a cone $\sigma$: its two generators do not share a face, so their sum must be a lattice point on $\Delta^\circ$ by \cref{lem:cls}. But this lattice point would fall interior to $\sigma$, meaning $\Sigma$ wasn't fine.
\end{proof}

In a slightly more cumbersome fashion, we can prove that basepoint cones in general cannot be ``too large.''
\begin{prop}
    Any fine simplicial fan whose minimal generators constitute the lattice points on the boundary of an integral polytope $\Delta^\circ$ cannot have basepoint cones of maximal dimension.
\end{prop}
\begin{proof}
    Assume by way of contradiction a fine simplicial fan $\Sigma$ whose minimal generators are the boundary lattice points of an $n$-dimensional polytope $\Delta^\circ \subset N_\mathbb{R}$ has a cone $\sigma$ with minimal generators $u_1, \dots, u_n$ which are not contained in a single facet of $\Delta^\circ$, but such that any size $n-1$ subset of them is contained in a single facet. Any such subset is linearly independent and thus determines an affine hyperplane, so there is a vector $m \in M \otimes \mathbb{Q}$ which is an (origin-facing) normal to the hyperplane (recalling $N^* = M$). We normalize $m$ such that $\langle m, u_i \rangle = -1$ for $1 \leq i \leq n$. This affine hyperplane cuts through the interior of $\Delta^\circ$ (otherwise $\sigma$ wouldn't be a basepoint cone) so $\Delta^\circ$ must have a vertex $p$ satisfying $\langle m, p \rangle < -1$.
    
    The strategy of our proof is now as follows. Because $p \in \Delta^\circ$, $p$ must obey the affine hyperplane constraints which define $\Delta^\circ$. In particular, the $n$ facets containing the minimal generators $\{u_1, \dots, u_n\} \setminus u_i$ for $1 \leq i \leq n$ define such constraints. Such facets must be unique, as otherwise $\sigma$ would be contained in any facet repeated more than once. Taking these inequalities in conjunction with the constraint on $p$ from the previous paragraph, we will argue that this confines $p$ to live inside of $\sigma$, which yields a contradiction, because then $\Sigma$ wouldn't be fine (as $p$ then couldn't feature in any cone). 
    
    Let us now show that $p$ satisfies all hyperplane constraints of $\sigma$. Because $\sigma$ is simplicial, each facet is generated by all but one of its generators $\{u_i\}$. Pick one, and say $u_j$ is the generator omitted. Then there is a hyperplane normal $m_j \in M$ such that the hyperplane constraint is $\langle m_j, n \rangle \geq 0$ for $n \in \sigma$, where $\{u_i\} \setminus u_j$ saturate the inequality and $u_j$ does not. Let $\Theta$ be the facet of $\Delta^\circ$ containing $\{u_i\} \setminus u_j$. Thus, there is a $m_\Theta \in M \otimes \mathbb{Q}$ satisfying $\langle m_\Theta, u \rangle \geq -1$ for all lattice points $u$ in $\Delta^\circ$ such that $\{u_i\} \setminus u_j$ saturate the inequality while $u_j$ does not.\footnote{If $\Delta^\circ$ is reflexive, $m_\Theta$ is merely the vertex of $\Delta$ dual to $\Theta$.} Let us write $p = \sum_{i = 1}^n a_i u_i$: because $\langle m, u_i \rangle = -1$ and $\langle m, p \rangle < -1$, we know $-\sum_i a_i = \langle m, p \rangle < -1$. Using this, we can argue:
    \begin{equation}
        \begin{aligned}
            -1 \leq \langle &m_\Theta, p \rangle = -\sum_{i \neq j} a_i + a_j \langle m_\Theta, u_j \rangle \\
            &\implies \sum_{i \neq j} a_i - 1 \leq a_j \langle m_\Theta, u_j \rangle \\
            &\implies \sum_i a_i - 1 \leq a_j (1 + \langle m_\Theta, u_j \rangle) \\
            &\implies 0 < a_j (1 + \langle m_\Theta, u_j \rangle) 
        \end{aligned}
    \end{equation}
    The parenthetical is non-negative, so $a_j$ is positive. Thus $\langle m_j, p \rangle = \langle m_j, a_ju_j \rangle > 0$, so $p$ satisfies the hyperplane constraint. Thus $p \in \sigma$, a contradiction, so maximal basepoint cones cannot occur for the fans considered here.
\end{proof}
In particular, then, vex triangulations do not exist for reflexive polytopes of dimension $\leq 3$, and in the four-dimensional case we have shown the following.
\begin{cor}
    The anticanonical basepoint locus of toric varieties arising from vex triangulations of four-dimensional reflexive polytopes is one-dimensional (i.e., a union of toric curves).
\end{cor}

\subsubsection{CY Smoothness from Birational Geometry}

\label{sec:smooth_birational}

The results of the previous section enable us to prove \cref{main}. In particular, it essentially follows from the following lemma.
\begin{lem}
    Flips of toric varieties associated to fine, regular fans from four-dimensional reflexive polytopes do not contract divisors on the anticanonical hypersurface.
    \label{lem:flips_to_flops}
\end{lem}
\begin{proof}
    Consider a flip of fine regular simplicial fans of four-dimensional reflexive polytopes and the birational map $f : V \dashrightarrow V'$ that it defines. Let $X$ denote the CY hypersurface in $V$: if the exceptional locus $\mathrm{Ex} \, f$ on which the map $f$ is not defined intersects $X$, then $f$ descends to a rational map on $X$. First, we note that $\mathrm{Ex} \, f|_X \subset \mathrm{Ex} \, f$ (because the exceptional locus is merely where the map is not defined). Thus, $\mathrm{Ex} \, f|_X$ can be a divisor (codimension $1$) only if $\mathrm{Ex} \, f$ has codimension $\leq 2$. But $\mathrm{Ex} \, f$ has codimension $\geq 2$ (as we are considering a flipping contraction, not a divisorial contraction) so the only possibility is that the codimension is exactly two. For a codimension-two $\mathrm{Ex} \, f$ to intersect $X$ at a divisor, it must be a non-complete intersection: an entire codimension-two component of $\mathrm{Ex} \, f$ must be contained in $X$. For this to hold for a generic anticanonical hypersurface, such a component must belong to the basepoint locus of the anticanonical class. But we showed in the previous section that the anticanonical basepoint locus is not codimension two: it is codimension at least three. Thus, $\mathrm{Ex} \, f|_X$ cannot be a divisor, and $f$ descends to at worst a flop.
\end{proof}
From here, we can directly proceed to the proof of \cref{main}.
\begin{proof}[Proof of \cref{main}.]
    All fine regular simplicial fans with fixed rays are connected by flips. By \cref{lem:flips_to_flops}, these can only result in the contraction of curves on the CY hypersurface: never the contraction of divisors. Thus, the only non-trivial map that can be induced on a CY hypersurface by a flip is a flop (or a composition of flops). In particular, all CY hypersurfaces associated to the fine regular simplicial fans must be connected by flops. Flops of smooth CY threefolds are smooth, so the smoothness of FRST hypersurface CYs (along with the existence of at least one FRST per reflexive polytope) ensures that vex triangulation hypersurface CYs are smooth as well. Moreover, because these CYs are all related by flops, they are all birational.
\end{proof}

\subsubsection{CY Hypersurface Singular Locus}

\label{sec:sing_locus}

To supplement the above discussion, we now undertake a direct analysis of the singular locus of a Calabi--Yau hypersurface constructed from a vex triangulation. What could contribute to this locus? First, a singular toric subvariety of dimension $2$ or greater --- codimension $2$ or smaller --- would be expected to intersect the anticanonical hypersurface, which can render it singular. In general, such an intersection with the ambient singular locus need not result in a singularity of the hypersurface, but a sufficient condition for hypersurface smoothness is the total avoidance of ambient singularities, so we limit our scope to this. Additionally, the Jacobian criterion can fail: there can exist solutions to $F = dF = 0$ for $F$ the defining equation of the CY. \\

\noindent \textbf{Toric Singularities.} Here we describe the conditions required for the generic CY hypersurface to avoid singularities of the ambient toric variety. We believe that this may not be strictly necessary for smoothness in general, but it is a fairly weak assumption which simplifies matters greatly, and our large-scale brute-force enumeration of fans in \cref{sec:count} will find no counterexamples. We recall from \cref{foot:cartier_and_smooth} that the anticanonical class must be Cartier for its representatives (that avoid the singular locus) to be smooth, so it is necessary for the ambient toric variety to be Gorenstein. \textit{A priori}, though FRST toric varieties are Gorenstein, their vex triangulation counterparts need not be: the Cartier condition for a Weil divisor is not a birational invariant (see, e.g., Ex. 15.4.8 in \cite{cls}). The Gorenstein property, if satisfied, has the useful consequence that singularities are guaranteed to be in codimension $4$. In this case, toric singularities would be point-like, but they could still intersect the Calabi--Yau hypersurface. For FRSTs, the anticanonical class is basepoint free, and the hypersurface equation can always be tuned to avoid toric point singularities. Vex triangulations, however, give rise to basepoint loci for $\bar{K}$, and point-like singularities could fall on these loci, so we must ensure they do not. Naively, then, to avoid toric singularities it is sufficient to impose that the ambient toric variety is Gorenstein and that all (point-like) singularities avoid the anticanonical base locus. As it turns out, the latter implies the former, so we have the following result.
\begin{prop}
    Given a (fine, regular) vex triangulation of a four-dimensional reflexive polytope, if all maximal cones containing a basepoint cone are smooth, then the CY hypersurface does not intersect any singularities of the ambient toric variety.
    \label{prop:toric_sing}
\end{prop}
\begin{proof}
    If all maximal cones containing basepoint cones are smooth, then all toric points contained in basepoint loci are smooth. It therefore suffices to show that our hypothesis further implies that the toric variety is Gorenstein, as then there are no higher dimensional toric singularities. To see this, note that for any fine regular triangulation, the anticanonical class certainly has integral Cartier data for all maximal cones \textit{not} containing basepoint cones: such cones are contained in three-faces, and the Cartier data are furnished by the dual vertices of $\Delta$. Because all other maximal cones are smooth, they admit Cartier data for \textit{any} divisor class (all classes restrict to Cartier divisors on the associated affine open charts: more practically, the matrix of minimal generators associated to any maximal cone is invertible over the integers). Thus, the anticanonical class has a complete set of Cartier data and the toric variety is Gorenstein.
\end{proof}

\noindent \textbf{Jacobian Criterion Singularities.} By Bertini's theorem, the Jacobian criterion for the CY hypersurface can only fail on the basepoint locus of anticanonical class. This precludes failures for FRSTs. Let us now study the Jacobian criterion on anticanonical basepoint loci for vex triangulations. For a hypersurface of the form \cref{eq:factor}, the partial derivatives take the form
\begin{equation}
    \partial_\tau F = \sum_{\rho \in \sigma} \Big( x_\rho \partial_\tau f_\rho + \delta_{\tau\rho} f_\rho \Big)
\end{equation}
On $V_{\Sigma,\sigma}$ the first term vanishes, so $\partial_\tau F$ only survives for $\tau \in \sigma$. Therefore, for a fixed component $V_{\Sigma,\sigma}$ of the basepoint locus, $F = dF = 0$ is solved at the intersection of six divisors: the three prime torics $D_\rho$ for $\rho \in \sigma$ (yielding $V_{\Sigma,\sigma}$), and the three classes $\Tilde{D}_\rho \equiv \overline{K} - D_\rho$ to which the $f_\rho$ belong for each $\rho \in \sigma$. While this is naively empty on a fourfold, resulting in no singularities, it need not be a complete intersection: we must study basepoint loci. First, we should remove any $\Tilde{D}_\tau$ containing $V_{\Sigma,\sigma}$ in its base locus, which would entail that  $\partial_\tau F |_{V_{\Sigma,\sigma}} = f_\tau |_{V_{\Sigma,\sigma}} = 0$ identically. We avoid this for fixed $\tau \in \sigma$ if and only if there is a term in $f_\tau$ featuring no factors of $x_\rho$ for any $\rho \in \sigma$. Employing the notation of \cref{eq:poly_faces}, this is the condition that the intersection 
\begin{equation}
    \label{eq:jacobian_set}
    S_\tau = \bigcap_{\rho \in \sigma} \Theta_{\Tilde{D}_\tau, \rho} \subset M
\end{equation}
of faces of $\Delta_{\Tilde{D}_\tau}$ is non-empty. Equivalently, the normal fan of $\Delta_{\Tilde{D}_\tau}$ must contain the basepoint cone $\sigma$.

Now keeping only the $\Tilde{D}_\rho$ with non-empty $S_\rho$, we split into cases. If there is a single surviving $\Tilde{D}_\rho$, the singular locus is $V_{\Sigma,\sigma} \cap \Tilde{D}_\rho$, a set of points one can compute with intersection theory. Only if the intersection number vanishes is the hypersurface smooth. If there is more than one, the intersection naively vanishes, but the $\Tilde{D}_\rho$ may share a basepoint locus inside $V_{\Sigma,\sigma}$: i.e., a toric point. Note that all surviving $\Tilde{D}_\rho$ must have the same toric point basepoint locus for their intersection to be non-trivial for generic $f_\rho$, so it suffices to ensure that for each cone $\sigma' \supset \sigma$, at least one surviving $\Tilde{D}_\rho$ doesn't contain $V_{\Sigma,\sigma'}$ in its basepoint locus.

Thus, we arrive at the following result.
\begin{prop}
    Given a (fine, regular) vex triangulation of a four-dimensional reflexive polytope, if for each basepoint cone $\sigma = \ConeOp(\rho_1, \rho_2, \rho_3)$ it is true that either
    \begin{itemize}
        \item at least two of the sets $S_{\rho_i}$ are non-empty and the associated $\overline{K} - D_{\rho_i}$ have no common toric point $V_{\Sigma,\sigma'}$ in their basepoint loci (for $\sigma \subset \sigma' \in \Sigma(4)$), or
        \item a single $S_\rho$ is non-empty and $\overline{K} - D_\rho$ does not intersect $V_{\Sigma,\sigma}$,
    \end{itemize}
    then the Jacobian criterion is satisfied for CY hypersurfaces.
    \label{prop:jacobian}
\end{prop}

In summary, if the hypotheses of \cref{prop:toric_sing} and \cref{prop:jacobian} are satisfied, then the anticanonical hypersurface associated to a vex triangulation is smooth. Let's see these criteria in an example.

\begin{example}
    We can exhibit some of the concepts in this section through the second fan constructed from the polytope $\Delta^\circ$ considered in \Cref{ex:intro}: i.e., the vex triangulation. One can compute that $\Sigma$ has only a single basepoint cone $\sigma = \cone{u_1, u_4, u_5}$. This is a three-cone, as we'd expect. Thus, the basepoint locus of $\overline{K}$ on $V_\Sigma$ is $V_{\Sigma,\sigma} = D_1 \cap D_4 \cap D_5$, because the generic anticanonical polynomial factors as in \Cref{eq:factor}. 

    For the CY hypersurface to avoid toric singularities, we must show that all maximal cones containing $\sigma$ are smooth. These are the cones $\sigma_1 = \cone{u_1, u_2, u_4, u_5}$ and $\sigma_2 = \cone{u_1, u_3, u_4, u_5}$, and indeed they are smooth. This implies that the toric variety is Gorenstein. In particular, the Cartier data for all other maximal cones is furnished by vertices of the dual polytope $\Delta$, and because the basepoint cones are smooth, the matrix whose rows are the points $u_1, u_2, u_4, u_5$ can be inverted over the integers to solve for the Cartier datum for $\sigma_1$ --- yielding $(3,-1,-1,-3)$ --- and likewise for $\sigma_2$. Because the toric variety is Gorenstein, its singularities are in codimension at least four, so the basepoint curve associated to $\sigma$ must be smooth, and indeed it is. 

    Now we verify that the Jacobian criterion holds. Letting $\rho_i$ be the ray generated by $u_i$, and letting $\sum_{\rho \neq \rho_i} D_\rho$ be the torus-invariant divisor representing the class $\Tilde{D}_i$ (fixing the polytope $\Delta_{\Tilde{D}_i}$), we can compute that
    \begin{equation}
        \begin{aligned}
            S_{\rho_1} &= \varnothing, \\
            S_{\rho_4} &= \varnothing, \\
            S_{\rho_5} &= \{(2,0,-1,-1)\} \subset \Delta_{\Tilde{D}_5}.
        \end{aligned}
    \end{equation}
    In particular, only one of the relevant $S_\rho$ is non-empty, so the Jacobian criterion fails (i.e., $F = dF = 0$) on the complete intersection $D_1 \cap D_4 \cap D_5 \cap (\overline{K} - D_5)$. We must check that the associated intersection number vanishes. Using the basis $\{D_2, D_5\}$ set by our GLSM charge matrix we fixed in \cref{ex:intro}, we can compute that the only non-vanishing intersection numbers are $D_2 \cdot D_5^3 = 1$ and $D_5^4 = -4$. This is enough to compute that 
    \begin{equation}
        \begin{aligned}
            D_1 \cdot D_4 \cdot D_5 \cdot (\overline{K} - D_5) 
            &= (D_2 + D_5) \cdot (D_2 + D_5) \cdot D_5 \cdot (6D_2 + 3D_5) \\
            &= 12 D_2 D_5^3 + 3 D_5^4 \\
            &= 0.
        \end{aligned}
    \end{equation}
    Thus, the Jacobian criterion holds everywhere, and the CY hypersurface is smooth.
    \label{ex:sing_locus}
\end{example}

\subsubsection{CY Intersection with Toric Subvarieties}

\label{sec:intersect}

It is useful to briefly touch on the intersection structure of toric subvarieties with the CY hypersurface. For examples, this helps one determine $K_\cup$ for vex triangulations, as discussed at the beginning of \cref{sec:cy_hyper}. 

We recall that prime toric divisors corresponding to points interior to facets of $\Delta^\circ \subset N$ are understood to not intersect the generic anticanonical divisor $X$. Traditionally, this is argued by noting that if $u_\rho$ is interior to a facet --- that is, it has minface dimension three --- setting $x_\rho = 0$ results in all anticanonical monomials vanishing aside from the one associated to the vertex $m \in \Delta$ dual to the facet containing $u_\rho$. This monomial depends only on homogeneous coordinates $x_\tau$ such that $u_\rho, u_\tau$ do not share a three-face. For FRSTs, $u_\rho$ and $u_\tau$ cannot belong to the same cone, so $x_\rho = x_\tau = 0$ belongs to the exceptional set $Z(\Sigma)$, meaning $X$ and $D_\rho$ do not intersect. For FRSTs, this argument holds for any toric subvariety $V_{\Sigma,\sigma}$ whose associated cone has minface dimension three.

For vex triangulations, basepoint cones provide exactly the mechanism to circumvent the above argument. If $u_\rho$ were to belong to a basepoint cone with $u_\tau$ on a different three-face, $D_\rho \cap D_\tau \cap X$ would be non-empty. However, we can apply the previous result to rule out this possibility for reflexive fans: any such pair $u_\rho, u_\tau$ would form a basepoint two-cone, which cannot exist. Note, though, this only ensures prime toric divisors with minface dimension three do not intersect the anticanonical hypersurface: for example, it is common for one of the two-cones in a basepoint three-cone to have minface three, but it will intersect the anticanonical hypersurface thanks to the third ray in the basepoint cone. This is conveniently restated in the language of \cref{subsubsec:flippable_secondary_cone}: a toric subvariety $V_{\Sigma,\sigma}$ with minface dimension three intersects the CY hypersurface if and only if the star $\text{st}_\Sigma(\sigma)$ has minface dimension four. 

In general, then, toric subvarieties with minface dimension four are contained in the CY hypersurface (belong to the basepoint locus), subvarieties with minface dimension three intersect the CY according to the above condition, and subvarieties with minface dimension two or smaller intersect the CY generically.

\subsubsection{Birational Geometry}

\label{sec:cy_bir_geo}

As first noted in \cref{sec:common_sig}, it is useful to categorize flips of triangulations according to the signatures of their associated circuits. It is worth briefly doing this for four-dimensional reflexive polytopes, and the CY hypersurfaces of their associated toric varieties.

Applying the discussion of \cref{sec:common_sig}, we find that the circuits associated to four-dimensional fans can only have the following signatures.
\begin{equation}
    (|J_+|, |J_-|) \in \Big\{ (3,2), (2,2), (4,1), (3,1), (2,1), (5,0), (4,0), (3,0), (2,0) \Big\}
\end{equation}
Among these, the only ones corresponding to birational maps --- i.e., the flips, which are isomorphisms in codimension one --- are $(3,2)$ and $(2,2)$. Thus we see that the flips of toric fourfolds (and thus, the flops of toric CY threefold hypersurfaces, by \cref{lem:flips_to_flops}) organize into two categories.

The $(2,2)$ circuit is already well known. A prototype on a reflexive polytope is the classic scenario of four points $u_1, u_2, u_3, u_4$ forming a quadrilateral inside of a two-face of $\Delta^\circ$, with linear relation $u_1 + u_2 - u_3 - u_4 = 0$. In this case, the flip exchanges which diagonal pair of points are connected by an edge ($u_1$ and $u_2$ or $u_3$ and $u_4$). Such flips, where the circuit has minface dimension $2$, are inherited by the Calabi--Yau as a flop. Indeed, this is the signature of the unique circuit associated to the resolved conifold, the non-compact toric threefold, the local model for a simple flop of length one \cite{atiyah1958analytic, katz1992gorenstein}.

However, the $(3,2)$ signature is less familiar. It only exists in four or greater dimensions, and does not preserve the number of maximal cones, in stark contrast to the $(2,2)$ circuit: in particular, it turns two maximal cones into three, or vice versa. Perhaps they have historically been less emphasized in the context of reflexive polytopes because all flips between two FRSTs of $(3,2)$ signature do not descend to flops of the anticanonical hypersurface. To see this, recall that CYs arising from FRSTs are fully determined by the two-face triangulations induced by the FRST, so a flip between FRSTs descending to a flop must change the two-face triangulations. Thus, any flip between FRSTs affecting the CY must descend to a $(2,2)$ flip on some two-face, but then the original flip must've been $(2,2)$ itself, as the addition of any further points (to $J_+$ or $J_-$) spoils minimal dependency. A different way of saying this is that a $(3,2)$ circuit cannot have minface dimension $2$, and for a flip between FRSTs it must be minface dimension $3$, so it does not affect the CY hypersurface per the discussion of \cref{sec:intersect}.

Now consider flips between an FRST and a vex triangulation. Such a flip must have a circuit with minface dimension four, because in going from the FRST to the vex triangulation it must introduce a basepoint cone, which is a subset of the circuit. However, as a corollary of \cref{main}, we actually cannot have a signature $(2,2)$ circuit of minface dimension $4$, because the associated coarse non-simplicial toric variety would have a non-simplicial $3$-cone on the basepoint locus of the anticanonical class, yielding a singular curve on the CY hypersurface, which contradicts the result that all flips must descend to flops.\footnote{It may seem that the possibility of $(2,2)$ circuits with minface dimension four is a counterexample to our proof for \cref{main}. However, such a circuit would yield a facet of a CY K\"ahler cone where a singular curve would develop without the contraction of a divisor. Such a scenario does not feature in the taxonomy of CY threefold K\"ahler cone facets, so we conclude that it cannot occur.} The only other option is a $(3,2)$ circuit of minface dimension $4$. The associated coarse non-simplicial toric variety would have non-simplicial $4$-cone, or a singular point on the anticanonical basepoint locus, so this would descend to a flop of the CY. In particular, flips between FRSTs and vex triangulations must descend to flops: they can never contract subvarieties which collectively miss the generic anticanonical hypersurface.

Finally, consider flips between two vex triangulations. The circuit signatures $(2,2)$ and $(3,2)$ are both possible, and both can descend to flops of the anticanonical hypersurface (even if the minface dimension of the circuit is three, due to the presence of basepoint cones, as discussed in the previous section).

This discussion is summarized as follows. Flips between two FRSTs are of signature $(2,2)$ or $(3,2)$, and only the $(2,2)$ signature can descend to a flop on the CY, which happens exactly when the induced two-face triangulations are changed. Flips between an FRST and a vex triangulation are only of signature $(3,2)$, and must descend to a flop. Finally, flips between two vex triangulations are of signature $(2,2)$ or $(3,2)$ and may or may not descend to a flop.

We conclude by briefly commenting on the remaining circuit signatures. Those of the form $(n,1)$ correspond to blowdowns of prime toric divisors, which often results in that divisor shrinking on the anticanonical hypersurface, inducing a boundary of K\"ahler moduli space. Those of the form $(n,0)$ correspond to a degenerating $(n-1)$-dimensional fiber of the fibered ambient toric variety. We've already discussed in \cref{sec:sec_fan} how the anticanonical hypersurface inherits these fibrations, with fibers of dimension $n-2$: these degeneration limits often result in the entire CY shrinking to zero volume, again marking a boundary of moduli space. However, in spite of the harsh topological changes undergone by the toric variety in codimension one, the impact on the anticanonical hypersurface can be much milder, resulting in either a flop or no singularity at all. In the case of a flop, the flopped CY has no immediate interpretation as a hypersurface in a toric variety. In this sense they are ``non-toric,'' though this name is misleading as one can construct higher-codimension toric embeddings of these CYs. We will refer to these geometries as non-inherited.

\begin{figure}
    \centering
    \includegraphics[width=0.5\linewidth]{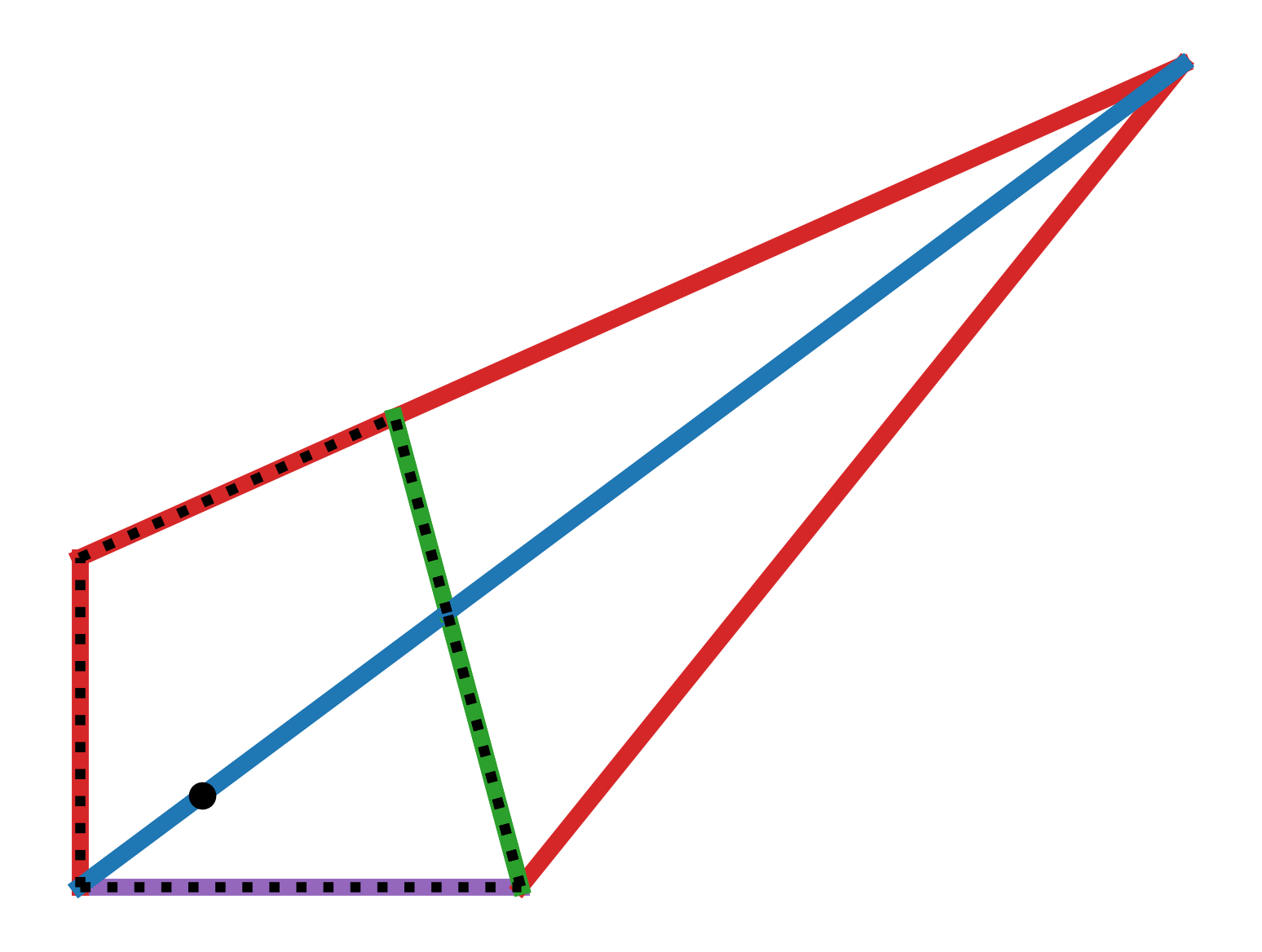}
    \caption{Two-dimensional slice of three-dimensional chamber fan (quotient of secondary fan by its lineality space) for the 4D reflexive polytope considered in \Cref{ex:38}, restricted to the moving cone (i.e., only showing the chambers associated to fine triangulations). That is, each polygon corresponds to a nef/K\"ahler cone of a toric variety. This polytope yields a CY hypersurface with $(h^{1,1}, h^{2,1}) = (3,69)$. Lines are colored according to circuit signature: green for $(3,2)$, blue for $(2,2)$, and red for $(n,1)$. In this example, purple is also $(n,1)$ but corresponds to a non-inherited flop of the CY hypersurface: the extended K\"ahler cone (moving cone) of the CY extends beyond the moving cone of the ambient toric variety. The black dashed region delineates the union of (toric) nef cones for FRSTs, and the black dot marks the ray generated by the anticanonical class, which lies at the intersection of all such cones.}
    \label{fig:38}
\end{figure}

\begin{example}
    Consider the four-dimensional reflexive polytope with non-zero points not interior to facets given by the columns of the following matrix. 
    \begin{equation}
        \begin{pNiceMatrix}[first-row]
             1 &  2 & 3 & 4 & 5 & 6 & 7\\
            -1 & 0 & -1 & 0 & 0 & 1 & 0 \\
            -1 & 0 & 2 & 0 & 1 & 0 & 1 \\
            -1 & 1 & 0 & 0 & 0 & 0 & 0 \\
            0 & 0 & -1 & 1 & 0 & 0 & -1
        \end{pNiceMatrix}
        \label{eq:poly_38}
    \end{equation}
    Because there are seven points, the height space for the associated vector configuration is seven dimensional. Because $\dim N = 4$, the secondary fan features a four-dimensional lineality space: performing a quotient by this subspace and passing to the chamber fan --- which we interpret toric geometrically as a refinement of the effective cone, in the class group --- yields a three-dimensional fan. We can restrict this fan to the moving cone, meaning we consider only cones corresponding to fine triangulations (for which the Picard group has rank $3$, yielding favorably embedded anticanonical hypersurfaces with $h^{1,1} = 3$). A two-dimensional slice of the resulting fan is shown in \Cref{fig:38}, with the plot style explained in the caption. 

    From \Cref{fig:38}, we see that there are four fine triangulations. In particular, there are two FRSTs and two vex triangulations. Recalling that the anticanonical class --- denoted by the black dot --- must be contained in the nef cone of the toric variety associated to any FRST, the left two chambers must correspond to FRSTs (the non-simplicial fan arising at the intersection of the two chambers is the central fan). They are related by a $(2,2)$ flip (blue) with $Z = (3,5,6,7)$. The remaining two vex triangulations are also related by this flip. There is additionally a $(3,2)$ flip (green) which maps FRSTs and vex triangulations to each other, with $Z = (1,2,3,4,5)$. 

    These four fine triangulations descend to four topologically distinct Calabi--Yau hypersurfaces with $(h^{1,1}, h^{2,1}) = (3, 69)$. By incorporating vex triangulations, we gained access to two new topologies, and improved our approximation of the extended K\"ahler cone from the union of the two FRST toric K\"ahler cones (black dotted region) to the entire toric moving cone (full quadrilateral). However, we have deliberately selected this example to showcase how this remains an approximation. This is not a complete birational class of Calabi--Yau threefolds. On each of the boundaries of the toric moving cone, a prime toric divisor is blown down on the toric variety, as a consequence of an $(n,1)$ circuit. For the top boundary, it is the seventh prime toric divisor $D_7$; for the right boundary, $D_3$; and for the bottom and left boundaries, $D_5$ (though via distinct circuits). On the top, right, and left boundaries, a divisor also shrinks to zero volume on the Calabi--Yau hypersurface, marking the end of the extended K\"ahler cone: however, this is not true for the bottom boundary, which is merely a flop wall in spite of it being a blowdown on the toric variety. Thus, there is a non-inherited phase on the other side of the purple boundary, not realized as a hypersurface in any member of the birational equivalence class of the ambient toric varieties. This illustrates how our improved toric approximation in general does not capture all topological members of a birational equivalence class nor the complete extended K\"ahler cone.
    \label{ex:38}
\end{example}

\section{4D Reflexive Polytopes II: Counting Triangulations}

\label{sec:count}

In this section we use the algorithm of \cref{sec:algo}, implemented and run using \texttt{CYTools} \cite{Demirtas:2022hqf} and \texttt{regfans} \cite{regfans}, to explicitly enumerate fine regular triangulations (i.e., fine regular simplicial fans) arising from the KS database for $h^{1,1} \leq 7$. Such a computation generalizes the counting of FRSTs performed in \cite{Altman:2014bfa, Crino:2022zjk} to include vex triangulations. We also present some particular examples and compute upper bounds on the total number of fine regular triangulations in KS. Finally, we present upper bounds on fine regular triangulations in the KS database, generalizing the methods of \cite{Demirtas:2020dbm}.

\subsection{Enumeration at Small $h^{1,1}$}

\label{sec:enumerate}

\begin{table}[]
    \centering
    \begin{tabular}{|c|c|c|c|}
        \hline
        $h^{1,1}$ & \# Polytopes & \# Fine Regular Triangulations & \# Vex Triangulations  \\
        \hline
        \hline
        1 & 5 & 5 & 0 \\ 
        \hline
        2 & 36 & 52 & 4 \\ 
        \hline
        3 & 244 & 694 & 168 \\ 
        \hline
        4 & 1197 & 9639 & 4291 \\ 
        \hline
        5 & 4990 & 137620 & 80570 \\ 
        \hline
        6 & 17101 & 1786783 & 1196698 \\
        \hline
        7 & 50376 & 22089147 & 16014955 \\
        \hline
        Total & 73949 & 24023940 & 17296686 \\
        \hline
    \end{tabular}
    \caption{Counts of all fine regular triangulations (i.e., fine regular simplicial fans) and of vex triangulations across all reflexive polytopes from the Kreuzer--Skarke database with $h^{1,1} \leq 7$.}
    \label{tab:count}
\end{table}

By applying the algorithm discussed in \cref{sec:algo}, by brute force we have explicitly solved our classification problem from \cref{sec:intro} and enumerated all fine regular simplicial fans arising from 4D reflexive polytopes for which the anticanonical hypersurface satisfies $h^{1,1} \leq 7$. This computation took a few days on a standard modern laptop. Our counts are summarized in \cref{tab:count}, and we plot the fraction of fine regular simplicial fans that are vex triangulations in \cref{fig:count}. In particular, we display this fraction in two ways. In blue, we combine the counts of fans across all polytopes for fixed $h^{1,1}$, yielding the total fraction of vex triangulations for that Hodge number. In orange, we compute the fraction for each polytope and present the average fraction for each $h^{1,1}$, with error bands providing the standard error. The orange curve is lower because there is a fraction of polytopes at each Hodge number which feature no vex triangulations. From these figures we see that for $h^{1,1} \geq 5$, there are more vex triangulations than FRSTs. Moreover, as $h^{1,1}$ increases, the ratio of vex triangulations to FRSTs increases monotonically: we anticipate that this trend continues beyond $h^{1,1} = 7$. 

\begin{figure}
    \centering
    \includegraphics[width=\linewidth]{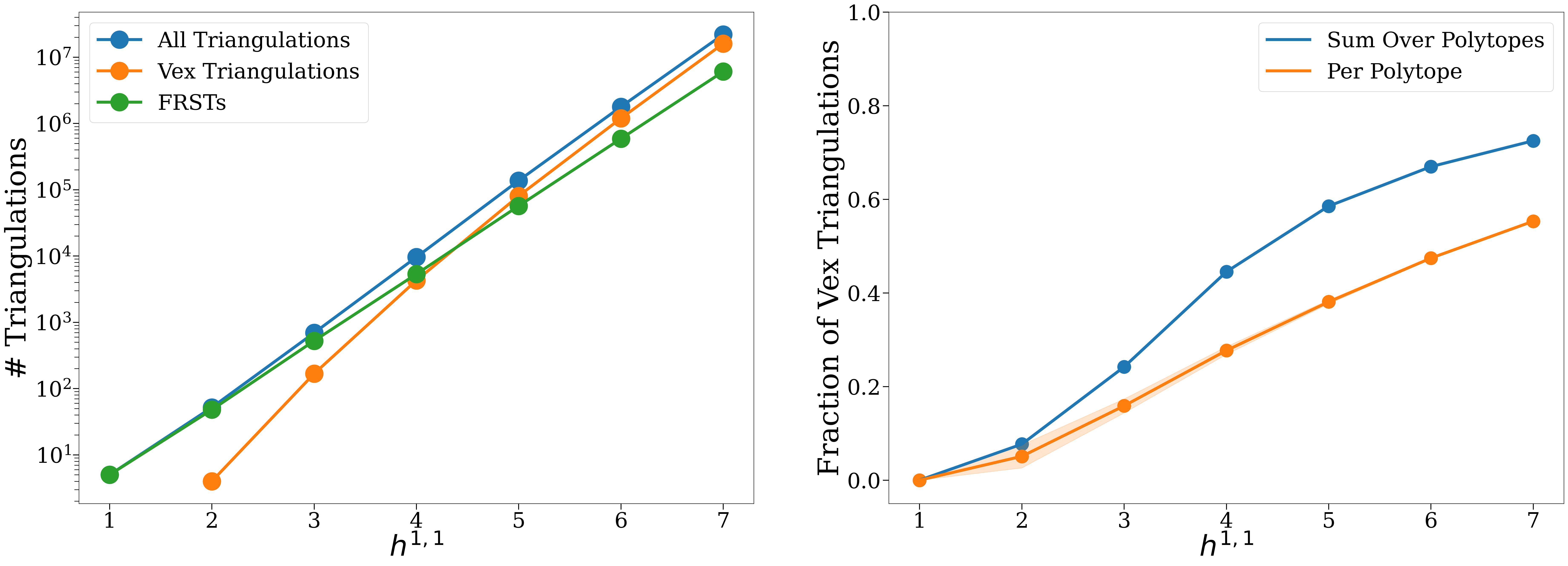}
    \caption{Results of exhaustively enumerating all fine regular triangulations (fine regular simplicial fans) of all 4D reflexive polytopes with $h^{1,1} \leq 7$: total counts (left) and fraction of triangulations that are vex (right). In particular, in the right panel, in blue we take the union of vex triangulations across all polytopes to compute one ratio per $h^{1,1}$; in orange, we compute the ratio per polytope, plotting the mean with error bands denoting the standard error.}
    \label{fig:count}
\end{figure}

While \cref{main} is the result that all of these toric varieties have smooth CY hypersurfaces, as a sanity check we verify that all vex triangulations have anticanonical hypersurfaces with empty singular loci. A sufficient condition for this is that the hypotheses of \cref{prop:toric_sing} and \cref{prop:jacobian} are satisfied, as this precludes intersection with toric singularities and failures of the Jacobian criterion, respectively. We indeed verify that this sufficient condition holds for all vex triangulations with $h^{1,1} \leq 7$.\footnote{We comment in passing that our condition for satisfying the Jacobian criterion is quite often non-trivially satisfied: in the language of \cref{prop:jacobian}, there are many examples where only a single $S_\rho$ is non-empty, so the singular locus is naively $0$-dimensional and an intersection theory computation must yield $0$ (and does), as in \cref{ex:sing_locus}.} Additionally, we noted in \cref{sec:cy_bir_geo} that $(2,2)$ embedded circuits with minface dimension $4$ should never occur. This is also verified in our exhaustive search. 

\subsection{Examples}

Having taken the time to enumerate fans at small $h^{1,1}$, we take some time to present a pair of interesting CY birational classes (i.e., polytopes) at $h^{1,1} = 3$. In \cref{ex:7} we present the birational class with the most unique toric hypersurface CY manifolds at $h^{1,1} = 3$, illustrating exactly how much one can gain by considering vex triangulations. In \cref{ex:192} we present a class which illustrates some qualitative differences between the decomposition of the moving cones of the ambient toric variety and the CY hypersurface into nef cones as well as some subtleties in counting distinct diffeomorphism classes. On this note, we stress that \cref{tab:count} should not be interpreted as a count of topologically distinct Calabi--Yau threefolds: we expect that there is significant degeneracy \cite{Gendler:2023ujl,Chandra:2023afu}.

\begin{figure}[h]
    \centering
    \includegraphics[width=0.5\linewidth]{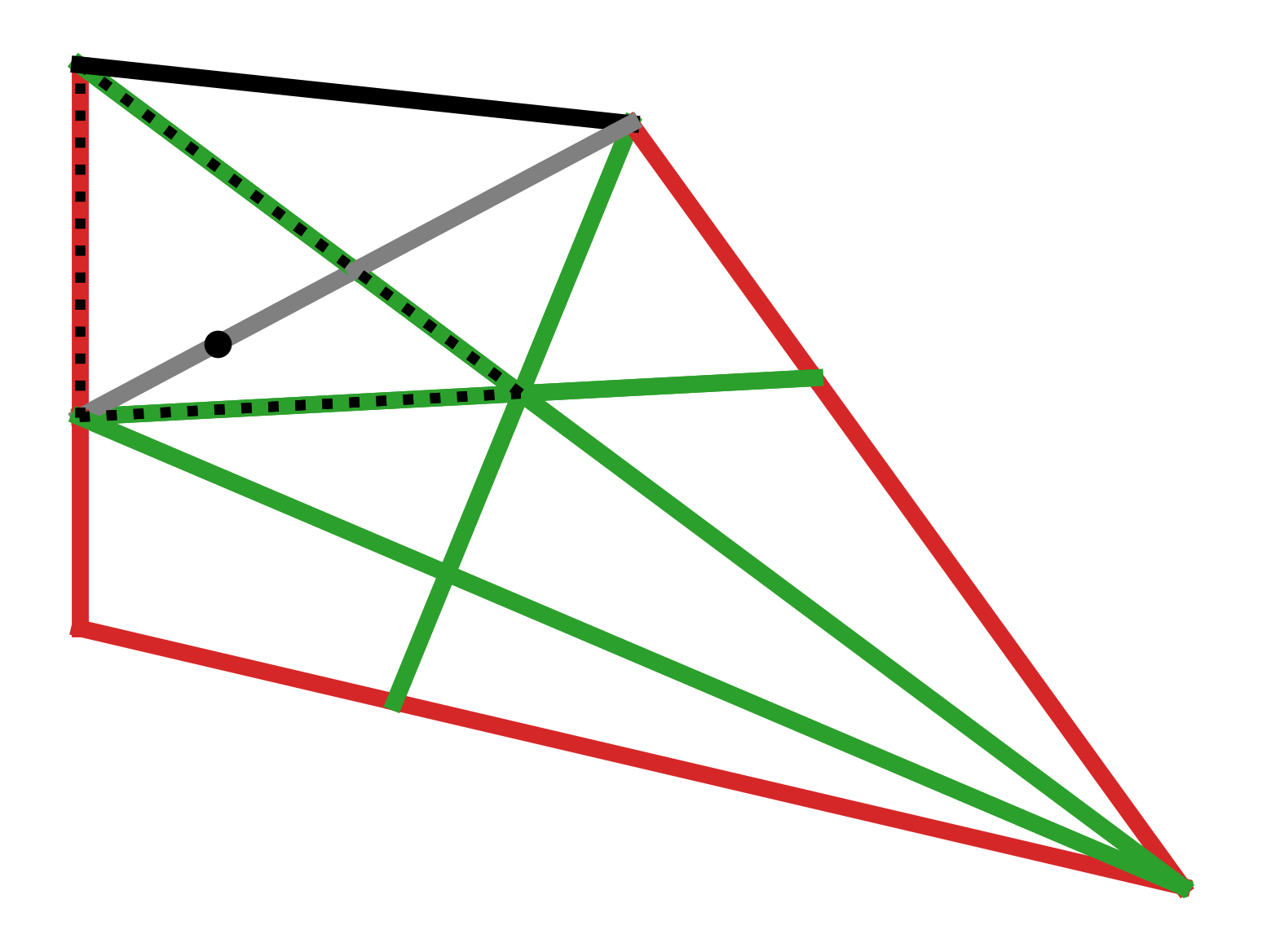}
    \caption{Two-dimensional slice of three-dimensional chamber fan (quotient of secondary fan by its lineality space) for the 4D reflexive polytope considered in \Cref{ex:7}, restricted to the moving cone (i.e., only showing the chambers associated to fine triangulations). This polytope yields a CY hypersurface with $(h^{1,1}, h^{2,1}) = (3,57)$. The styling is in agreement with \Cref{fig:38}, with a few additions: gray boundaries denote a flip which does not descend to a flop on the CY, and black boundaries denote a $(n,0)$ circuit.}
    \label{fig:7}
\end{figure}

\begin{example}
    Consider the four-dimensional reflexive polytope with non-zero points not interior to facets given by the columns of the following matrix. 
    \begin{equation}
        \begin{pNiceMatrix}[first-row]
            1 & 2 & 3 & 4 & 5 & 6 & 7 \\
            0 & 1 & -1 & 0 & 0 & -1 & 0 \\
            0 & 0 & -1 & -1 & 1 & 0 & 0 \\
            1 & 0 & 0 & -1 & 0 & -1 & 0 \\
            0 & 0 & 2 & 0 & 0 & -1 & 1
        \end{pNiceMatrix}
        \label{eq:poly_7}
    \end{equation}
    This is a setup analogous to \Cref{ex:38}, so we display a similar plot in \Cref{fig:7}: namely, a two-dimensional cross-section of the three-dimensional chamber fan --- restricted to the toric moving cone --- descending from a seven-dimensional secondary fan. We see that there are ten fine triangulations: in particular, two FRSTs (two chambers containing the anticanonical class, denoted by the black dot) and eight vex triangulations. These fine fans are related strictly by five distinct $(3,2)$ circuits.

    These ten fine triangulations descend to eight topologically distinct Calabi--Yau hypersurfaces with $(h^{1,1}, h^{2,1}) = (3, 57)$. This reduction in number arises because one of the $(3,2)$ circuits --- marked in gray --- has minface dimension $3$ and thus does not descend to a flop on the CY, following the discussion of \cref{sec:cy_bir_geo}. Thus, the two FRSTs correspond to a single CY --- whose K\"ahler cone is the union of the toric nef/K\"ahler cones of the two toric varieties --- and the two vex triangulations adjoining the gray $(3,2)$ circuit analogously correspond to a single CY. 

    This is the example with the greatest number topologically distinct CYs at $h^{1,1} = 3$ in the Kreuzer--Skarke dataset. By incorporating vex triangulations, we gained access to seven new topologies and seven new chambers of the extended K\"ahler cone. In fact, in this example, the toric moving cone agrees with the extended K\"ahler cone, so we captured an entire birational equivalence class.   
    
    Analogous to \Cref{ex:38}, the red boundaries of the toric moving cone indicate that a prime toric divisor is blown down, and on the CY a divisor shrinks as well. However, the black boundary corresponds to a $(3,0)$ circuit, and is therefore also a boundary of the toric effective cone (the support of the chamber fan), which agrees with the CY effective cone in this example. We recall from the end of \cref{sec:cy_bir_geo} the toric variety and CY hypersurface whose nef/K\"ahler cones contain this facet are fibered --- in this case, the CY is genus-one fibered --- and approaching that facet geometrically corresponds to the shrinking fiber limit.
    \label{ex:7}
\end{example}

\begin{figure}[h]
    \centering
    \includegraphics[width=0.5\linewidth]{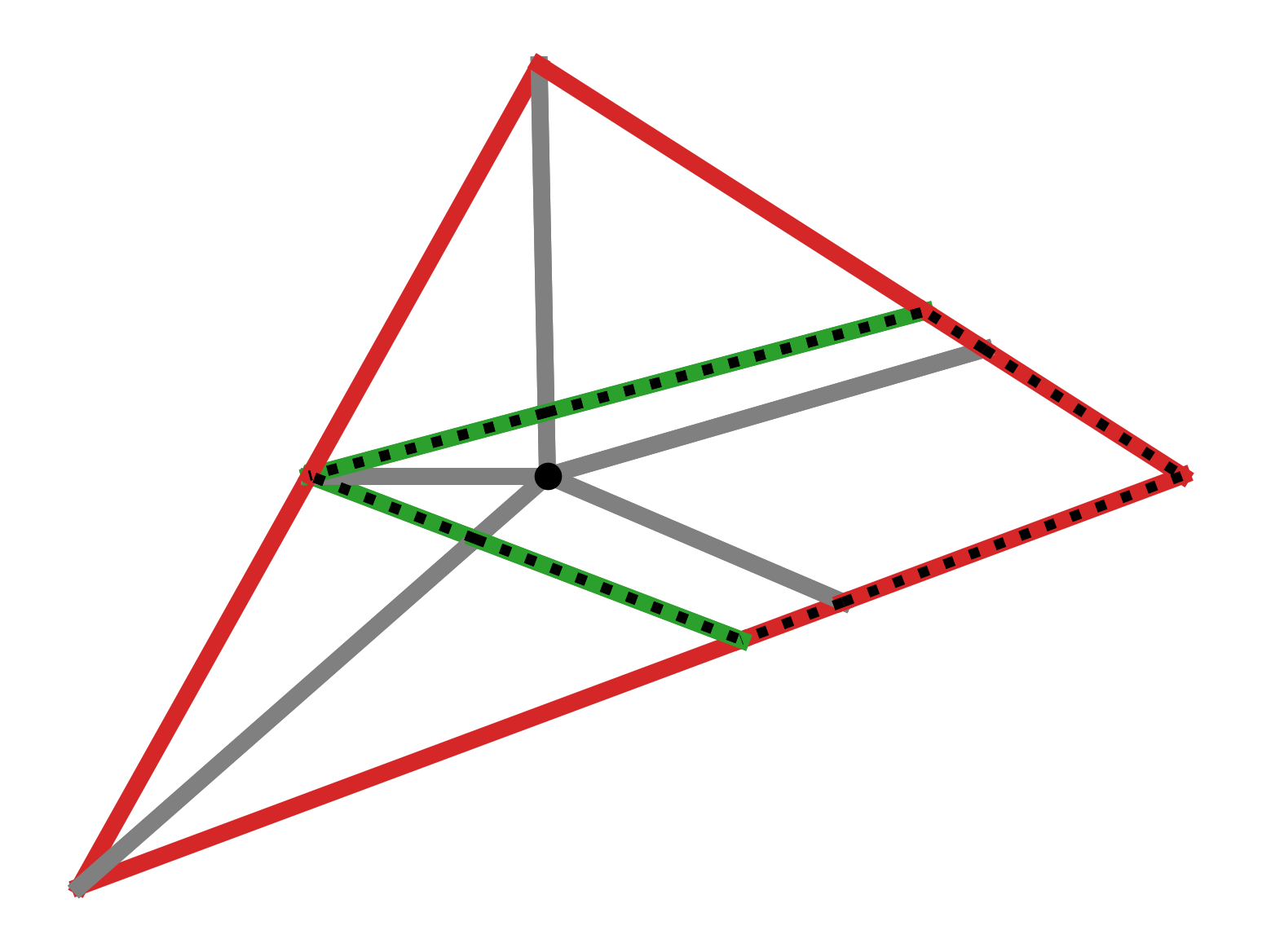}
    \caption{Two-dimensional slice of three-dimensional chamber fan (quotient of secondary fan by its lineality space) for the 4D reflexive polytope considered in \cref{ex:192}, restricted to the moving cone (i.e., only showing the chambers associated to fine triangulations). This polytope yields a CY hypersurface with $(h^{1,1}, h^{2,1}) = (3,111)$. The styling is in agreement with \Cref{fig:7}.}
    \label{fig:192}
\end{figure}

\begin{example}
    Consider the four-dimensional reflexive polytope with non-zero points not interior to facets given by the columns of the following matrix. 
    \begin{equation}
        \begin{pNiceMatrix}[first-row]
            1 & 2 & 3 & 4 & 5 & 6 & 7 \\
            1 & 0 & 1 & -3 & -3 & 0 & 0 \\
            0 & 1 & 0 & -1 & -1 & 0 & 0 \\
            0 & 0 & 1 & -1 & 0 & 0 & 1 \\
            0 & 0 & 1 & 0 & -1 & 1 & 0
        \end{pNiceMatrix}
        \label{eq:poly_157}
    \end{equation}
    This is a setup analogous to \Cref{ex:38} and \Cref{ex:7}, so we display a similar plot in \Cref{fig:192}: namely, a two-dimensional cross-section of the three-dimensional toric moving cone. We see that there are nine fine triangulations: in particular, five FRSTs (five chambers containing the anticanonical class, denoted by the black dot) and four vex triangulations. 
    
    These nine fine triangulations naively descend to three topologically distinct Calabi--Yau hypersurfaces with $(h^{1,1}, h^{2,1}) = (3, 111)$. We again employ gray to denote flips of the toric variety which do not descend to flops of the CY hypersurface. Similar to \Cref{ex:7}, all of the FRSTs actually only correspond to a single CY. In particular, the five circuits which relate two FRSTs --- four $(3,2)$ circuits and a $(2,2)$ circuit, the five gray lines converging on the anticanonical class --- each have minface dimension three. 

    However, there are actually only two topologically distinct CYs in this birational class. This polytope has an automorphism given by exchange of the last two coordinates, which swaps $u_4 \leftrightarrow u_5$ and $u_6 \leftrightarrow u_7$. This induces an action on the class group of the toric variety, given by the swapping of the associated prime toric divisors. The extended K\"ahler cone of the CY experiences this symmetry as well. In particular, the K\"ahler cone of the unique CY arising from the FRSTs of this polytope is fixed by this action (though the top and bottom generators are exchanged), but the K\"ahler cones of the two remaining CYs are mapped to each other, so they correspond to the same CY topology. 

    We conclude by stressing an important difference between the secondary fan/chamber fan and the extended K\"ahler cone. In the former, the hyperplanes delineating chambers can end: we see this clearly in \Cref{fig:192}, where each of the five hyperplanes marked by gray end at the anticanonical class, and do not extend further. A circuit will only define a hyperplane for a secondary cone if it is embedded in the associated triangulation; this need not hold for every triangulation whose secondary cone the hyperplane would naively intersect. However, in the extended K\"ahler cone, hyperplanes delineating chambers correspond to effective curve classes on the CY with non-trivial Gopakumar-Vafa invariants, which do not vary across the K\"ahler moduli space: thus, these hyperplanes cannot end, and must cut all the way through the extended K\"ahler cone. That is, the extended K\"ahler cone is always a \textit{hyperplane arrangement}. In particular, if a hyperplane in the secondary fan does not extend through the entire moving cone, it must not descend to a flop on the CY hypersurface (i.e., the associated curve class on the CY must not be effective). This is exemplified in \Cref{fig:192} by the fact that the five hyperplanes which end at the anticanonical class do not descend to flops.
    \label{ex:192}
\end{example}
In this last example we see how the degeneracy induced by automorphisms makes counting topologically distinct CYs difficult. We mention that another source of automorphisms is the symmetric flop --- a flop which relates two CYs which are actually topologically equivalent --- which are explored in \cite{Brodie:2021ain, Brodie:2021toe, Gendler:2022qof}. The map exchanging the K\"ahler cone of the two topologies related by a symmetric flop lifts to an automorphism of the extended K\"ahler cone. A further source of degeneracy arises from the fact that that distinct polytopes give rise to the same birational class of Calabi--Yau threefolds, as found in \cite{Gendler:2023ujl,Chandra:2023afu}. We defer the full problem of enumerating topological distinct CYs arising from vex triangulations, along the lines of \cite{Gendler:2023ujl,Chandra:2023afu}, to future work.

\subsection{Upper Bounds}
\label{subsubsec:upper}

Recall from \cref{thm:pc_to_vc,thm:vc_to_pc} the maps between FRSTs of a point configuration $\mathbf{A}_\mathrm{PC}$ and fine, regular triangulations (FRTs) of their associated vector configuration $\mathbf{A}_\mathrm{VC}=\mathbf{A}_\mathrm{PC}\setminus0$. Notably, the map from FRTs of $\mathbf{A}_\mathrm{PC}$ to FRTs of $\mathbf{A}_\mathrm{VC}$ is surjective. Hence, the number of FRTs of $\mathbf{A}_\mathrm{VC}$, denoted $N_\mathrm{FRT}(\mathbf{A}_\mathrm{VC})$, is bounded from above by $N_\mathrm{FRT}(\mathbf{A}_\mathrm{PC})$. This is nice because \cite{Demirtas:2020dbm} demonstrated a method to bound $N_\mathrm{FRT}(\mathbf{A}_\mathrm{PC})$ for any reflexive polytope,
\begin{equation}
    N_\mathrm{FRT}(\mathbf{A}_\mathrm{PC}) \leq \binom{(\dim(\mathbf{A}_\mathrm{PC})+1)V-1}{n-1}
\end{equation}
for $V$ the volume of $\conv{\mathbf{A}}$ and $n$ the number of lattice points not interior to facets of $\conv{\mathbf{A}}$. This is a modification of \cite{Demirtas:2020dbm}'s Eqs. (3.3) and (3.4) so as not to restrict to star triangulations. This bound shows that there are $\leq10^{979}$ FRTs of the vector configurations in the Kreuzer--Skarke database, analogous to \cite{Demirtas:2020dbm}'s bound of $\leq10^{928}$ FRSTs.

It should be stressed: these upper bounds are often brought up for counting the number of possibly distinct CYs that one can produce. \cite{Demirtas:2020dbm} demonstrated that there are many fewer CYs producible from FRSTs, only $\leq10^{429}$, because of diffeomorphism classes being repeated by two-face equivalent FRSTs (see \cref{sec:frst_cy_prop}). In fact, \cite{bigicy} subsequently improved this bound on CYs from FRSTs down to $\leq 10^{296}$. So, while the bound $\leq 10^{979}$ similarly restricts the number of inequivalent CYs that can be produced from FRTs of the vector configurations in the Kreuzer--Skarke database, this is likewise anticipated to be a significant overestimate due to repetitions of diffeomorphism classes.

\section{Conclusion}

In this work, in order to study the birational geometry of smooth toric hypersurface Calabi--Yau threefolds --- counting the distinct topological representatives and mapping out the K\"ahler moduli space --- we considered the classification problem of enumerating fine regular triangulations (simplicial fans) associated with a fixed vector configuration (set of rays). We performed this analysis from the complementary perspectives of triangulation theory and toric geometry, reviewing both theories in detail and explaining how they both ultimately converge on the geometric notion of the secondary fan. In particular, we note an efficient computational algorithm for exhaustively enumerating regular triangulations/fans, which resolves the classification problem.

We then applied this theory and methodology to four-dimensional reflexive polytopes, studying the associated toric varieties and their Calabi--Yau hypersurfaces. In particular, we demonstrated how considering all fine regular simplicial fans naturally leads one to incorporate fans outside the purview of the FRSTs originally considered by Batyrev: i.e., vex triangulations. We then studied non-weak-Fano Gorenstein toric varieties arising from vex triangulations of four-dimensional reflexive polytopes and their Calabi--Yau hypersurfaces. In particular, we identified how they compared and contrasted with the FRST case.

A large component of our analysis revolved around the possible singularities of CY hypersurfaces from vex triangulations. We considered two complementary perspectives: ``extending'' smoothness from FRST toric varieties via flops, and explicitly computing the singular locus of a toric CY hypersurface. The former approach culminated in a proof of smoothness (\cref{main}) with the important corollary that all fine regular triangulations of a four-dimensional reflexive polytope yield birational CY hypersurfaces (\cref{moving}). Our theoretical analysis was corroborated by a brute-force enumeration of fans with few rays: i.e., toric varieties whose anticanonical hypersurfaces satisfy $h^{1,1} \leq 7$, of which we found over $24$ million. In particular, we found that as $h^{1,1}$ increases, vex triangulations constitute an increasing fraction of all fine regular triangulations. 

We conclude that the incorporation of vex triangulations into analysis of the Kreuzer--Skarke database will provide access to many new discrete topological classes of Calabi--Yau threefolds and provide a toric description for novel continuous regions in K\"ahler moduli space. In particular, these topologies represent novel environments for performing conjecture testing, identifying rare topological or geometric features, or constructing new and interesting string compactifications. Likewise, these newly toric regions in K\"ahler moduli space represent new regimes for, e.g., the numerical optimization of phenomenologically desirable string theoretic observables. In this way, vex triangulations benefit the mathematics and physics communities alike. 

To close, we note some relevant future directions for research in this area (beyond those of the previous paragraph).
\begin{enumerate}
    \item \textbf{Generalizing two-face equivalence.} We recalled in \cref{sec:cy_hyper} that two FRSTs of a reflexive four-dimensional polytope $\Delta^\circ$ yield the same CY3 if their restrictions to two-faces (i.e., the induced triangulations of the 2-skeleton of $\Delta^\circ$) agree. Is there a generalization of this simple combinatorial criterion which holds for all fine regular triangulations? This would enable generalizations of the methods of \cite{macfadden2023efficient}, which would likely provide a significant speed-up to exhaustive enumeration of toric CY hypersurfaces for fixed polytopes at larger $h^{1,1}$.
    \item \textbf{Bounding vex triangulations (and their CY hypersurfaces).} Completing the previous task would enable generalizations of \cite{Demirtas:2020dbm} and \cite{bigicy}, which bound inequivalent CY3s from KS more intelligently by incorporating two-face equivalence. This would presumably dramatically improve the upper bound presented in \cref{subsubsec:upper}.
    \item \textbf{Counting CY3 diffeomorphism classes.} More generally, it would be worthwhile to continue the task began in \cite{Gendler:2023ujl, Chandra:2023afu} of counting inequivalent CY3s using invariants including --- but certainly not limited to --- some generalization of two-face equivalence. In particular, it would be insightful to understand the number of non-inherited CY geometries relative to hypersurfaces arising from FRSTs and vex triangulations.
    \item \textbf{Further extension of toric descriptions.} In this work, we provide a better toric inner approximation for the K\"ahler moduli space of toric hypersurface CY3s, but this is still only an inner approximation. Can the non-inherited phases (mentioned, e.g., in \cref{ex:38}) be rendered toric in general, possibly as higher-codimension subvarieties in distinct toric varieties?\footnote{For example, if one has a toric hypersurface description of the smooth deformation associated to a flop between toric and non-inherited CY threefolds, then there is an algorithm for constructing the non-inherited phase as a codimension-two complete intersection in a toric variety, generalizing \S{III} in \cite{Candelas:1987kf}. How far can this kind of approach be extended?}
    \item \textbf{Studying Reid's fantasy in Kreuzer--Skarke.} At the end of \cref{sec:pc_vs_vc}, we briefly discussed to how insertion/deletion flips may induce maps (i.e., extremal transitions, such as conifold transitions) on toric hypersurface CY3s whose composition connects the entirety of Kreuzer--Skarke (because any two polytopes can be related by the insertion and deletion of finitely many points). Explicitly realizing Reid's fantasy \cite{reid1987moduli} in Kreuzer--Skarke in this way would be valuable (and would complement the known result for CICYs \cite{Green:1988bp, wang2016connectedness}).
    \item \textbf{Extending toric hypersurface CY3 methods.} While the toric varieties induced by FRSTs and vex triangulations are largely similar, certain methods in Calabi--Yau geometry developed specifically for FRST toric varieties may need to be adapted/generalized before they can be directly applied to vex triangulations (e.g., the computational implementation of HKTY mirror symmetry in \cite{Demirtas2024}). 
    \item \textbf{Considering non-reflexive polytopes.} We briefly mentioned canonical polytopes \cite{Kasprzyk:sca} in \cref{sec:toric}, or polytopes with a unique interior point (taken to be the origin). FRSTs of such polytopes will yield non-Gorenstein varieties (i.e., non-Cartier anticanonical classes) whose anticanonical hypersurfaces will be singular, but vex triangulations of these polytopes are perfectly capable of furnishing smooth CY hypersurfaces.
    \item \textbf{Continuing development of computational tools.} The construction of the Kreuzer--Skarke database led to the development of computational tools for studying FRSTs of polytopes and their associated toric/CY geometry, such as PALP \cite{Kreuzer:2002uu} and CYTools \cite{Demirtas:2022hqf}: continued computational advancements (enabling, e.g., systematic analysis of vex triangulations and/or non-reflexive polytopes) would be valuable.
\end{enumerate}

\label{sec:conc}

\section*{Acknowledgments}

We are grateful to Mike Stillman for originally inspiring this project and many useful conversations. We thank Washington Taylor and collaborators for sharing their unpublished notes on vex triangulations \cite{vex_notes}, Per Berglund and Tristan Hubsch for insightful correspondences about their theory of VEX polytopes, and Andreas Schachner for the provision of computational resources for the enumeration undertaken in \cref{sec:enumerate}. We thank Per Berglund, Tristan Hubsch, Doddy Marsh, Jakob Moritz, Andres Rios-Tascon, Michael Stepniczka, Mike Stillman, and Washington Taylor for comments on a draft. We thank the anonymous referee for pointing out the relevance of \cite{Gross1997DeformationObstructed, Ruan1996TopologicalSigma}. ES is grateful to Federico Carta, Federico Compagnin, Sebastian Vander Ploeg Fallon, Liam McAllister, Jakob Moritz, and Rachel Webb for many helpful conversations. The authors are supported in part by NSF grant PHY-2309456. This material is based upon work supported by the National Science Foundation Graduate Research Fellowship under Grant No. 2139899.

\appendix

\section{Translation}
\label{sec:translate}

We hope that the content of \cref{sec:triang}, \cref{sec:toric}, and \cref{sec:sec_fan}, as well as the dictionary presented in \cref{tab:dict}, largely clarify how to pass back and forth between the vocabulary of triangulation theory and toric geometry. For example, like in \cref{sec:sec_fan}, fix a vector configuration $\mathbf{A}$ in $\mathbb{R}^n$ with the vectors of $\mathbf{A}$ being integral and primitive. Let $\Sigma(1)$ be the one-cones (or rays) generated by the vectors of $\mathbf{A}$: then for all triangulations of $\mathbf{A}$ the one-cones will be a subset of $\Sigma(1)$. In particular, there is a bijection between the triangulations of $\mathbf{A}$ and the fans whose one-cones are a subset of $\Sigma(1)$. 

It is important to bear in mind the following dichotomy. To define a triangulation, the fundamental objects are the configuration and cells (sets of labels); these then define geometric regions by their support. In toric geometry it is the opposite: the fundamental objects are fans (sets of geometric regions, or cones, in the sense of \cref{def:fan}) and the minimal generators are derived from the fan (as lattice generators of one-cones). The bijection of the previous paragraph, though, means that translation between triangulation theory and toric geometry can often be fluid and straightforward. This is important: for example, it allows us to apply triangulation theory to solve problems in toric geometry (such as our main classification problem posed in \cref{sec:intro}). However, there are some subtleties in this translation process in general, and in this appendix, we will take some time to discuss some subtleties in passing back and forth between triangulation theory and toric geometry. \\

\noindent \textbf{Subdivisions vs Fans.} First, we will explain the failure of bijectivity between regular subdivisions and regular fans. To achieve this, let us be precise about how we can map back and forth between these objects. Take the vector configuration $\mathbf{A}$ we introduced at the beginning of the section: with this fixed, we have two maps, namely $\mathtt{Fan}$ (mapping vector configuration subdivisions of  $\mathbf{A}$ to fans with one-cones contained in $\Sigma(1)$) and $\mathtt{Sub}$ (mapping the other way). As we've been doing, let $u_\rho$ denote the minimal generator in $N$ associated to the one-cone $\rho \in \Sigma(1)$. These two maps are defined as follows.
\begin{align}
    \mathtt{Fan} = \mathscr{T} &\mapsto \{ \supp{F} \; | \; F \in \mathscr{T} \}, \\
    \mathtt{Sub} = \Sigma &\mapsto \{ \{u_\rho \; | \; \rho \subset \sigma \} \; | \; \sigma \in \Sigma \; \}.
\end{align}
We have that $\mathtt{Fan} \circ \mathtt{Sub}$ is the identity but $\mathtt{Sub} \circ \mathtt{Fan}$ is not, though the restriction of $\mathtt{Sub} \circ \mathtt{Fan}$ to vector configuration \textit{triangulations} is the identity. Thus, for regular subdivisions, $\mathtt{Sub} \circ \mathtt{Fan}$ differs from the identity only on a measure-zero subset of the secondary fan, namely some of the codimension $\geq 1$ faces of the secondary fan. 

When does $\mathtt{Sub} \circ \mathtt{Fan} \neq 1$? The underlying mechanism is that one can construct two subdivisions featuring different cells that have the same support, and the subdivisions are formally different (e.g., if regular, they have different secondary cones). In such a case, only one can be in the image of $\mathtt{Sub}$. Consider the following toy example.
\begin{example}
    Consider an acyclic vector configuration $\mathbf{A} = \{v_1, v_2, v_3\}$ with labels $1, 2, 3$ such that $v_1 \in \cone{v_2, v_3}$. The subdivision $\mathscr{S}$ of $\mathbf{A}$ consisting of the cell $\{2,3\}$ is formally different from the subdivision $\mathscr{S}'$ consisting of the cell $\{1,2,3\}$, even though $\cone{v_2,v_3} = \cone{v_1, v_2, v_3}$. For example, the former subdivision is a triangulation while the latter is not. Note that $\mathscr{S}'$ cannot contain the cell $\{1\}$ (its intersection with $\{1,2,3\}$ is not a face of $\{1,2,3\}$, violating \cref{def:subdivision}), so $v_1$ is a vector that doesn't correspond to a one-cone (ray) in $\mathtt{Fan}(\mathscr{S}')$. 

    We can see that $\mathtt{Fan}(\mathscr{S}) = \mathtt{Fan}(\mathscr{S}')$, with both being the fan given by $\cone{v_1, v_2, v_3} = \cone{v_2, v_3}$ and its faces. Applying $\mathtt{Sub}$ to this fan gives back $\mathscr{S}$, so that $(\mathtt{Sub} \circ \mathtt{Fan})(\mathscr{S}') =  \mathscr{S}$.
    
    This has some toric geometric consequences: for example, though $\mathscr{S}$ is a regular subdivision with three vectors that is not a triangulation, the associated toric variety is simplicial and has two minimal generators (and thus two prime toric divisors, etc.). 
    \label{ex:toy_deletion}
\end{example}
We note that given a regular subdivision $\mathscr{T}(\mathbf{A}, \omega)$, we can always construct a possibly distinct regular subdivision $\mathscr{T}(\mathbf{A},\omega')$ which has the same cells as $(\mathtt{Sub} \circ \mathtt{Fan})(\mathscr{T}(\mathbf{A}, \omega))$. In particular, $\omega'$ is given as follows.
\begin{equation}
    \omega'_i = \begin{cases}
        \omega_i & \mathbf{A}_i \text{ is extremal in all of its cells,} \\ \infty & \text{otherwise.}
    \end{cases}
\end{equation}
That is, one just lifts the non-extremal vectors in $\mathscr{T}(\mathbf{A}, \omega)$ to infinity, removing them from the cells of $\mathscr{T}(\mathbf{A}, \omega)$. We see that $\omega' = \omega$ if all vectors are extremal: this follows, for example, if $\mathscr{T}(\mathbf{A}, \omega)$ is a triangulation.

A useful guiding principle is that subdivisions associated to (codimension $\geq 1$) faces of the secondary fan can never be triangulations in triangulation theory: they must be non-simplicial. Sometimes these subdivisions have non-simplicial cones (see, e.g., \cref{ex:simple_secondary_cone}) but sometimes the cones are simplicial and only their cells are non-simplicial (see, e.g., the discussion preceding \cref{def:fan}). This is how insertion/deletion flips work: indeed, the subdivisions $\mathscr{S}$ and $\mathscr{S}'$ in \cref{ex:toy_deletion} are exactly the ``$-$'' and ``$0$'' constituents of a deletion flip with circuit signature $(2,1)$. \\

\noindent \textbf{Fans with lineality spaces.} We've now seen a subtlety that arises, for example, on secondary cone facets corresponding to insertion/deletion flips with circuit signature $(n,1)$. As it turns out, so too do we have intricacies for facets associated with $(n,0)$ circuits: or, more generally, faces of the support of the secondary fan. Consider a subdivision of $\mathbf{A}$ associated to such a face. One can show that all cones contain a common non-trivial linear subspace (i.e., lineality space).\footnote{In particular, one can show that if a single cone in a fan has a lineality space, then every cone does, such that the lineality space --- rather than the origin $\{0\}$ --- is the smallest face/cone in the fan.} This is a perfectly fine subdivision of $\mathbf{A}$ as-is, but the toric variety associated to this subdivision is actually given by the projection of this subdivision onto the quotient of the ambient space $N_\mathbb{R}$ by the lineality space of the subdivision. In particular, fans with non-trivial lineality spaces are called \textit{generalized fans} in \cite{cls}.

Let us illustrate the previous paragraph through an example.
\begin{example}
    Let $\mathscr{T}$ be a triangulation of $\mathbf{A}$ which admits an embedded $(n,0)$ circuit with $J_+ = {1, \dots, n}$ and $J_- = \varnothing$ (i.e., a fibering contraction). Now there exists no triangulation theoretic flip: we just let $\mathscr{T}_0$ denote the fan which arises on the associated facet of the secondary cone of $\mathscr{T}$, which falls on the boundary of the secondary fan. One finds that all cones in the subdivision $\mathscr{T}_0$ feature lineality spaces: that is, they contain a linear subspace (which is then the minimal face of the subdivision). In particular, they contain the subspace generated by $u_1, \dots, u_n$. In toric geometry, the toric variety assigned to this generalized fan is the toric variety of the fan achieved by projecting $\mathscr{T}_0$ onto the quotient $N_\mathbb{R} / \mathbb{R}\{u_1, \dots, u_n\}$.

    This comes with a nice geometric interpretation. In toric geometry, fibrations of toric varieties with fan $\Sigma \subset N_\mathbb{R}$ correspond to short exact sequences of the form
    \begin{equation}
        0 \to N' \to N \to N'' \to 0,
    \end{equation}
    and the fan of the fiber (base) is the preimage (image) of $\Sigma$ in $N'$ ($N''$). In this case, the $(n,0)$ circuit induces such a fibration for $\Sigma$: namely, we have
    \begin{equation}
        0 \to \mathbb{Z}\{u_1, \dots, u_n\} \to N \to N'' \to 0,
    \end{equation}
    and $\Sigma_0$ is exactly the image of $\Sigma$ in $N''$, so it is the base of this fibration. In particular, $\dim V_{\Sigma_0} = \dim N'' = \dim N - (n - 1)$. We can therefore interpret this facet of the secondary cone of $\Sigma$ as the limit where the fiber in $V_\Sigma$ shrinks, leaving only the base $V_{\Sigma_0}$. 
    
    It is also easy to understand this from the perspective of normal fans of Newton polytopes. Divisor classes $D$ falling on the boundary of the effective cone are not big, which in the toric context means their Newton polytopes $\Delta_D$ have strictly smaller dimension than $\dim M$ (see \S9.3 in \cite{cls}). The inclusion $\iota : M'' \to M$ of the smallest sublattice $M''$ containing $\Delta_D$ is dual to the projection $N \to N''$, so in particular the normal fan of $\Delta_D$ in $M''$ is the toric fan $\Sigma_0$ of the base of the fibration of $V_\Sigma$. 
\end{example}
Before concluding, it's worth looking at examples of $(n,1)$ and $(n,0)$ circuits in a more concrete example, with actual numbers.
\begin{example}
    We return to the Hirzebruch surface $\mathscr{H}_r$ one last time, because we recall from \cref{ex:hirz_3} that the walls of its secondary cone correspond exactly of an $(n,1)$ circuit and an $(n,0)$ circuit, providing a convenient example for the discussion of this section. Let us first consider the facet associated to the $(n,1)$ circuit. This facet has an associated subdivision with cells
    \begin{equation}
        \{1,2,4\}, \{1,3\}, \{2,3\}.
    \end{equation}
    Here, the first cell is non-simplicial, and $u_4$ falls in its relative interior. Thus, this subdivision is not a triangulation. However, the associated fan is simplicial: it is merely the fan of $\mathbb{P}_{11r}$. Just as we saw in \cref{ex:toy_deletion}, the ``$-$'' and ``$0$'' constituents of this deletion flip correspond to different subdivisions (the former has the cell $\{1,2\}$ while the latter has $\{1,2,4\}$) that induce the same fan.
    
    Now we turn our attention to the other facet of the secondary cone, the positive circuit. This facet has an associated subdivision with cells
    \begin{equation}
        \{1,3,4\}, \{2,3,4\}.
    \end{equation}
    These are the upper-half and lower-half plane in $\mathbb{R}^2$, respectively. In particular, both cones contain the lineality space spanned by $(1,0)$, which is the linear subspace spanned by the positive circuit $\{u_3, u_4\}$. To this subdivision, toric geometry assigns the toric variety associated to the fan achieved by taking the quotient by the lineality space. This quotient maps $u_3, u_4$ to $0$ and $u_1,u_2$ to $1$ and $-1$ in $\mathbb{Z}^2/\mathbb{Z} \cong \mathbb{Z}$, yielding the fan for $\mathbb{P}^1$. Geometrically, then, we recover the fact that $\mathscr{H}_r$ is a $\mathbb{P}^1$-fibration over $\mathbb{P}^1$: in particular, the ``quotient fan'' that toric geometry associated to this facet is the base of this fibration.
    \label{ex:hirz_4}
\end{example}

\section{Gale Duality}

\label{sec:gale}

Many of the objects and results reviewed in this paper can be understood as the consequence of Gale duality: while it wasn't necessary to introduce this duality in the text, we briefly discuss it here for completeness, and because it's a nice, elegant phenomenon.

At its most basic level, Gale duality is essentially the observation that if one has two matrices $A$ and $B$, then
\begin{equation}
    \label{eq:basic_gale}
    BA^\top = 0 \quad \Longleftrightarrow \quad AB^\top = 0
\end{equation}
In particular, fix a choice of vectors $a_1, \dots, a_n$ and a choice of basis for their linear relations $b_1, \dots, b_k$ (i.e., $\sum_i (b_j)_i a_i = 0$ for each $j$). If one constructs the matrix $A$ with $a_i$ as columns and the matrix $B$ with the $b_i$ as rows, then \cref{eq:basic_gale} is the statement that the rows of $A$ are also a basis for the linear relations of the columns of $B$. This is the reason, for example, that in \cref{ex:hirz_3} we were able to use the rows of the matrix $\mathbf{A}_r$ (whose columns were the minimal generators) as a basis of the linear relations of the class group grading / GLSM charge matrix $Q$ (to get the kernel of $\beta$). Let $b'_i$ denote the $n$ \textit{columns} of $B$: we say that the $a_i$ and the $b'_i$ are Gale dual to each other. Note that Gale duals are not unique, as $B$ is determined by $A$ only up to a change of basis. In some sense, then, the orthogonal linear subspaces spanned by the columns of $A$ and rows of $B$ are the objects that are intrinsically dual to each other, though nice results follow when bases are non-uniquely fixed for these subspaces.

We can organize these matrices into a short exact sequence as follows.
\begin{equation}
    0 \to \mathbb{R}^m \stackrel{A^\top}{\to} \mathbb{R}^n \stackrel{B}{\to} \mathbb{R}^k \to 0
\end{equation}
By comparison with \cref{eq:SES}, we can conclude that the vectors of a vector configuration $\mathbf{A}$ (i.e., minimal generators of a fan) and the columns of the GLSM charge matrix $Q$ (the classes of the associated prime toric divisors) are Gale dual to each other. That is, we can take $A = \mathbf{A}$ and $B = Q$. Moreover, from \cref{eq:sec_cone} we see that one constructs the chamber fan of a configuration using a Gale dual of that configuration. One might even hope that the minimal generators of the secondary fan are exactly a Gale dual of the original configuration, but this is too good to be true: the secondary fan will in general have strictly more generators.

An important fact is if a configuration is acyclic (i.e., its fans are non-complete, and the associated toric varieties are non-compact) then the Gale dual configuration is totally cyclic, and vice versa. This has the important consequence that acyclic (totally cyclic) configurations have totally cyclic (acyclic) secondary fans. This explains why any height vector results in a valid triangulation of an acyclic configuration (the support of the secondary fan is the entire vector space) but the same doesn't hold for a totally cyclic configuration (because then not every height vector belongs to the support of the secondary fan). For example, this explains why in \cref{sec:frst_vs_vex}, when we constructed the total space of the canonical bundle of a compact toric variety $V$ (turning a totally cyclic vector configuration into an acyclic one), it was natural that we added the previously non-effective canonical class $K_V$ to the charge matrix, as this turned the vector configuration associated to the secondary fan from an acyclic vector configuration to a totally cyclic one.

\bibliographystyle{JHEP}
\bibliography{ref}

\end{document}